\documentclass[11pt]{article}
 
\usepackage{xr,epstopdf,soul,array,booktabs,amsmath,amssymb,chemarr,subfigure,amsthm,authblk,mathrsfs}
\usepackage{graphicx}

\usepackage{tikz}

\usetikzlibrary{arrows}

\newcommand{\rank}{\mbox{rank}}
 \newtheorem{theorem}{Theorem}
  
  \newtheorem{remark}{Remark}
    \newtheorem{proposition}[theorem]{Proposition}
    \newtheorem{corollary}[theorem]{Corollary}
    \newtheorem{lemma}[theorem]{Lemma}
        
        \def\@seccntformat#1{\csname the#1\endcsname 33 \quad}

\definecolor{darkred}{rgb}{0.7,0,0.0}
\newcommand{\red}[1]{\textcolor{black}{#1}}
\newcommand{\emr}[1]{\textcolor{black}{#1}}
\newcommand{\rv}[1]{\textcolor{black}{#1}}

\title{\Huge Multi-modality in gene regulatory networks with slow promoter kinetics}

\author[a]{M. Ali Al-Radhawi}
\author[a]{Domitilla Del Vecchio}
\author[b,c]{Eduardo D. Sontag}

\affil[a]{Department of Mechanical Engineering, Massachusetts Institute of Technology, Cambridge, MA 02139-4307, USA.}
\affil[b]{Departments of Bioengineering and of Electrical and Computer Engineering, Northeastern University, Boston, MA 02115, USA. Email: \texttt{sontag@sontaglab.org}}
\affil[c]{Harvard Program in Therapeutic Science, Harvard Medical School, Boston, MA 02115, USA. }

\date{This manuscript was compiled on \today}

\begin{document}

\maketitle

 \begin{abstract}
{Phenotypical variability in the absence of genetic variation often reflects
complex energetic landscapes associated with underlying gene regulatory
networks (GRNs).  In this view, different phenotypes are associated with alternative
states of complex nonlinear systems: stable attractors in deterministic models
or modes of stationary distributions in stochastic descriptions. We provide
theoretical and practical characterizations of these landscapes, specifically
focusing on stochastic slow promoter kinetics, a time scale relevant when
transcription factor binding and unbinding are affected by epigenetic
processes like DNA methylation and chromatin remodeling.  In this case,
largely unexplored except for numerical simulations, adiabatic approximations
of promoter kinetics are not appropriate.  In contrast to the existing
literature, we provide rigorous analytic characterizations of multiple modes.
A general formal approach gives insight into the influence of parameters and
the prediction of how changes in GRN wiring, for example
through mutations or artificial interventions, impact the possible number,
location, and likelihood of alternative states.  We adapt tools from the
mathematical field of singular perturbation theory to represent
stationary distributions of Chemical Master Equations for GRNs as
mixtures of Poisson distributions and obtain explicit formulas for the
locations and probabilities of metastable states as a function of the
parameters describing the system.  As illustrations, the theory is used to
tease out the role of cooperative binding in stochastic models in comparison
to deterministic models, and applications are given to various model systems,
such as toggle switches in isolation or in communicating populations, \rv{a synthetic oscillator} and a trans-differentiation network.
} \\ \strut
\end{abstract}

A gene regulatory network (GRN) consists of a collection of genes that
transcriptionally regulate each other through their expressed proteins.
Through these interactions, including positive and negative feedback
loops, GRNs play a central role in the overall control of cellular life
\cite{alon,naturereview,ddv_book}.
The behavior of such networks is stochastic due to the random nature of
transcription, translation, and post-translational protein modification
processes, as well as the varying availability of cellular components
that are required for gene expression.
Stochasticity in GRNs is a source of phenotypic variation among genetically
identical (clonal) populations of cells or even organisms
\cite{munsky12}, and is considered to be one of the mechanisms facilitating
cell differentiation and organism development
\cite{zhou11}.
This phenotypic variation may also confer a population  {an} advantage when facing fluctuating
environments
\cite{arkin98,stamatakis09}.
Stochasticity due to randomness in cellular components and transcriptional
and translational processes have been thoroughly researched
\cite{kepler01,kaern05}.

The fast equilibration of random processes sometimes allows stochastic
behavior to be ``averaged out'' through the statistics of large numbers
at an observational time-scale, especially when genes and proteins are found in
large copy numbers.
In those cases, an entire GRN, or portions of it, might be adequately
described by a deterministic model.  Stochastic effects that occur
at a slower time scale, however, may render a deterministic analysis inappropriate
and might alter the steady-state behavior of the system.
This paper addresses a central question about GRNs: how many different
``stable steady states'' can such a system potentially settle upon, and how does
stochasticity, or lack thereof, affect the answer? 
To answer this question, it is necessary to understand the possibly different
predictions that follow from stochastic versus deterministic models of gene
expression.
Indeed, qualitative conclusions regarding the steady-state behavior of gene
expression levels in a GRN are critically dependent on whether a
deterministic or stochastic model is used (see \cite{hahl16} for a recent
review).  {It follows that the mathematical characterization of phenomena
such as non-genetic phenotype heterogeneity, switching
behavior in response to environmental conditions, and
lineage conversion in cells, will depend on the choice of the model.}


In order to make the discussion precise, we must clarify the meaning of the
term ``stable steady state'' in both the deterministic and stochastic
frameworks.
Deterministic models are employed when molecular concentrations are large, or
if stochastic effects can be averaged out.
They consist of systems of
ordinary differential equations describing averaged-out approximations of the
interactions between the various molecular species in the GRN under study.
For these systems, steady states are the zeroes of the vector field defining
the dynamics, and ``stable'' states are those that are locally asymptotically
stable.  The number of such stable states quantifies the degree of
``multi-stability''  of the system.
Stochastic models of GRNs, in contrast, are based upon continuous-time Markov
chains which describe the random evolution of discrete
molecular count numbers.
Their long-term behavior is characterized by a stationary probability
distribution that describes the gene activity configurations and the protein numbers recurrently visited.
Under weak ergodicity assumptions, this stationary distribution is unique
\cite{norris}, so
multi-stability in the sense of multiple steady states of the Markov chain
is not an interesting notion.
A biologically meaningful notion of ``multi-stability'' in this context, and
the one that we employ in our study, is ``multi-modality,'' meaning the
existence of multiple modes (local maxima) of stationary distributions.

Intuitively, given a multi-stable deterministic system, adding noise may help
to ``shake'' states, dislodging them from one basin of attraction of one
stable state, and sending them into {the} basin of attraction of another
stable state.

Therefore, in the long run, we are bound to see the various deterministic
stable steady states with higher probability, that is to say, we expect that
they will appear as modes in the stationary distribution of the Markov chain of
the associated stochastic model.
This is indeed a typical way in which modes can be interpreted as
corresponding to stable states, with stochasticity responsible for the
transitions between multiple stable states \cite{gardiner}.
However, new modes could arise in the stationary distribution of a
stochastic system besides those associated with stable states of the
deterministic model, and this can occur even if the deterministic model had
just a single stable state.
 This phenomenon of ``stochastic multi-stability'' has attracted considerable
attention lately, both in theoretical and experimental work
\cite{kepler01,kaern05,eldar10,symmons16,bressloff17}.
Stochastic multi-stability has been linked to behaviors such as
transcriptional bursting/pulsing \cite{pirone04,raj06} {and} GRN's binary response
\cite{karmakar04}{.  Furthermore, }multi-state gene transcription \cite{munsky12}
has been used to propose explanations for phenotypic heterogeneity in
isogenic populations.

A common assumption in gene {regulation} models is that transcription factor (TF) to gene binding/unbinding is
significantly faster than the rate of protein production and decay
\cite{alon}.
However, it has been proposed \cite{kaern05,shahrezaei08} that the emergence
of new modes in stochastic systems in addition to those that arise from the
deterministic model
might be due to low gene copy numbers and
\emph{slow} promoter kinetics, which means that the process of binding and unbinding of TFs to promoters is slow.
Thus, the emergence of multi-modality may be due to the slow TF-gene binding and unbinding. \rv{Already in prokaryotic cells, where DNA is more accessible to TF binding than in eukaryotic cells, some transcription factors can take several minutes to find their targets, comparable or even higher than the time required for gene expression \cite{tabaka13},\cite{potapov15}.}
{This is \rv{more} relevant in eukaryotic cells, in which transcriptional regulation is often mediated by an additional regulation layer dictated by DNA methylation and histone modifications, commonly referred to as chromatin dynamics. }  For example, the presence of
nucleosomes makes binding sites less accessible to TFs
and therefore TF-gene binding/unbinding is modulated by the process of chromatin opening \rv{\cite{briegel96}},\cite{paldi03,kaern05,miller11,voss14}.
{DNA methylation, in particular,   has \red{also} been reported to
slow down TF-gene binding/unbinding \cite{yuan16}.}
Several experiments have consolidated the role of the aforementioned complex transcription processes
in slow promoter kinetics \cite{raj06,maheshri10,mariani10,yuan16}.

In summary, new modes may appear in the stationary distribution that do not
correspond to stable states in the deterministic model.  Conversely, multiple
steady states in the deterministic model may collapse, being ``averaged out''
by noise, with a single mode representing their mean.
It is a well-established fact that, in general, multi-stability of the
deterministic description of a biochemical network and multi-modality of the
associated stochastic model do not follow from each other \cite{ebeling79}.
This is especially true in low copy number regimes with slow promoter
kinetics. \red{Figure \ref{f.intro_example} gives two examples for the emergence of new modes due to slow promoter kinetics, and it shows that equilibria derived from the corresponding deterministic model do not provide relevant information on the number and locations of the modes.  }

\begin{figure}[t]
\subfigure[]{ \raisebox{0.6\height}{\includegraphics[width=0.075\textwidth]{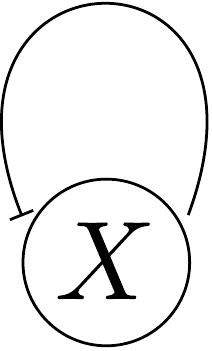}}}
\subfigure[]{\includegraphics[width=0.425\textwidth]{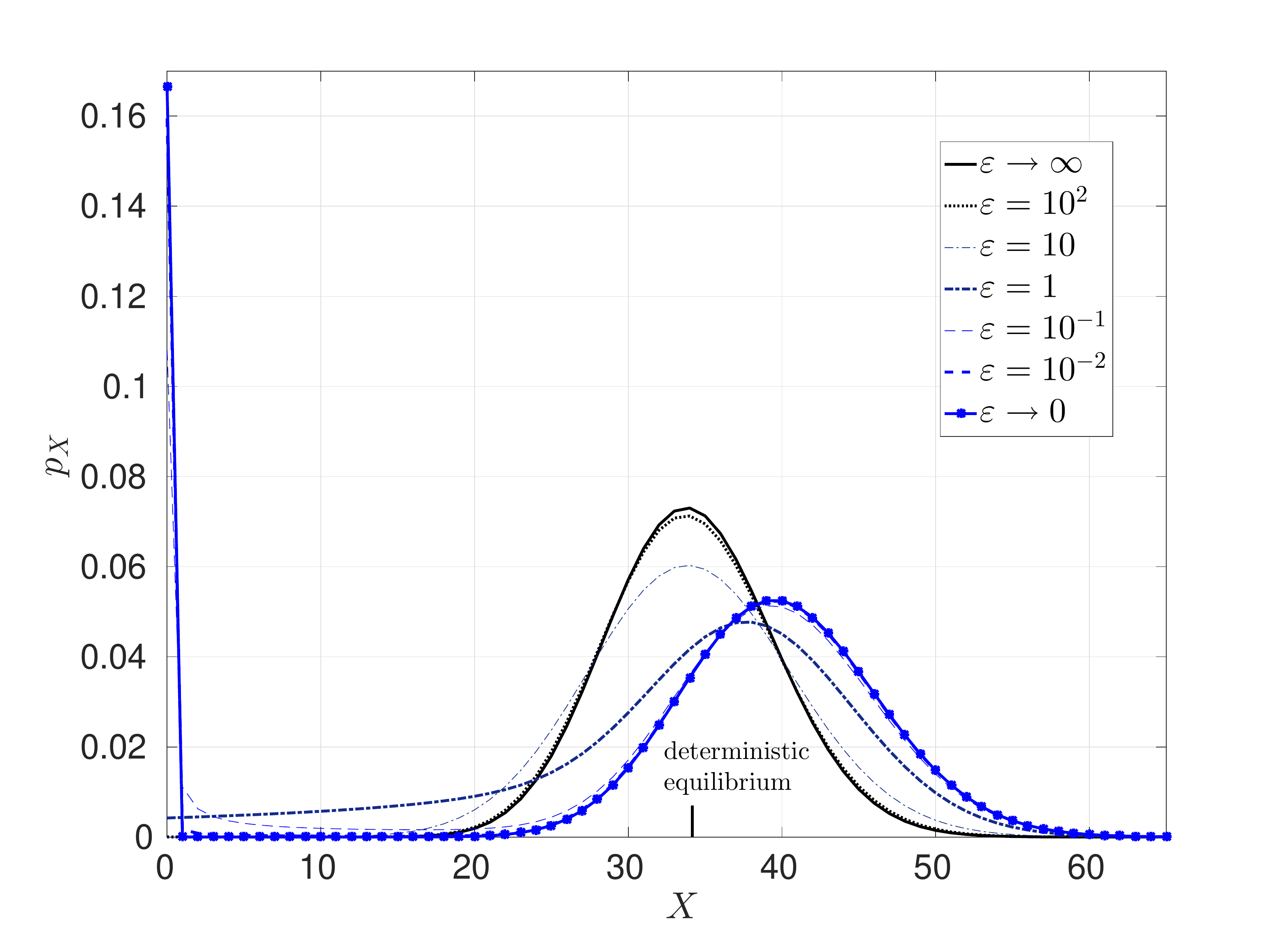}}
\subfigure[]{\raisebox{0.4\height}{\includegraphics[width=0.075\textwidth]{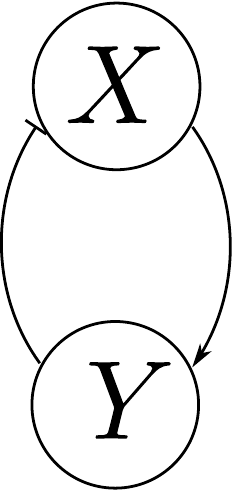}}}
\subfigure[]{\includegraphics[width=0.425\textwidth]{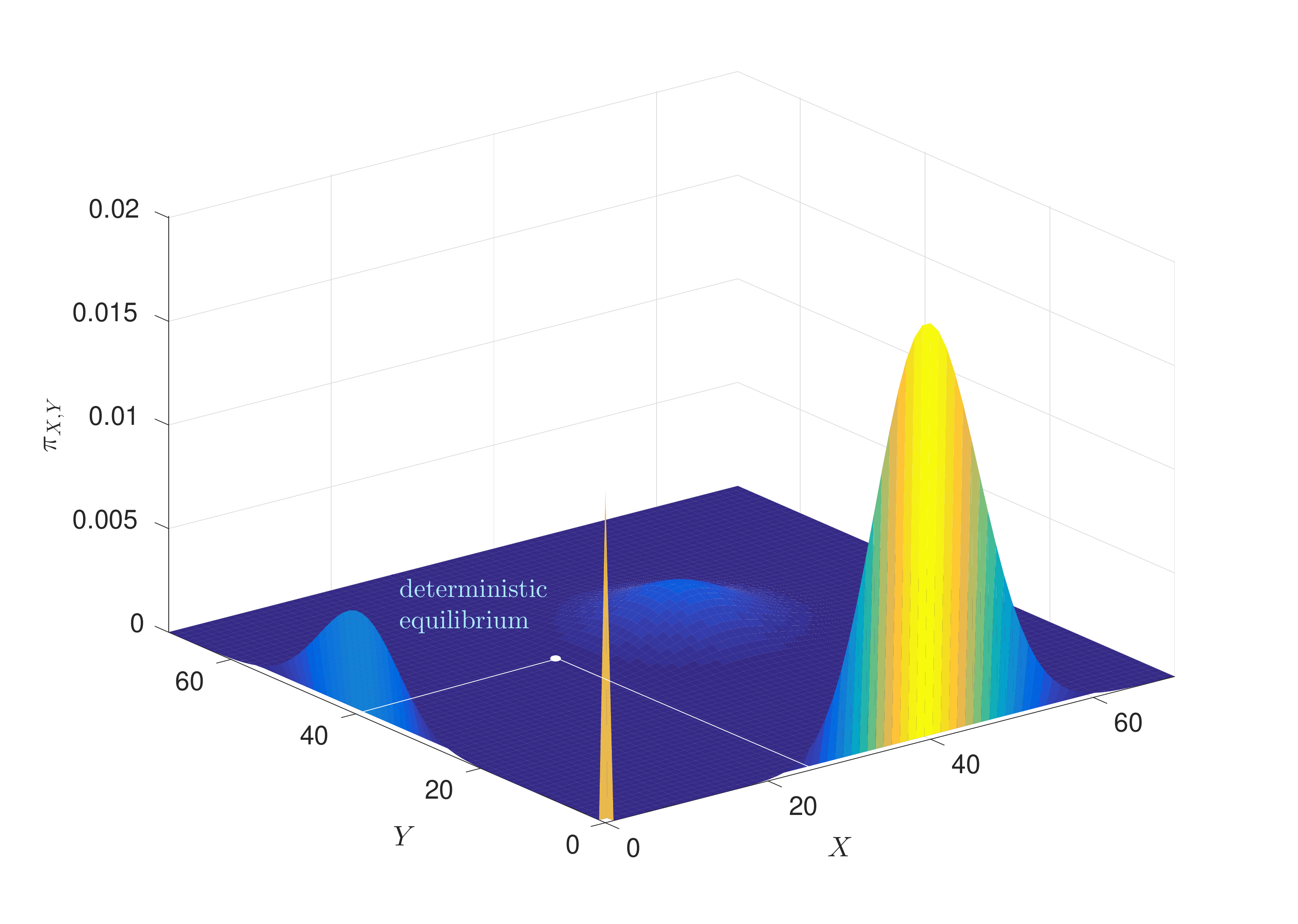}}
\caption{\textbf{Emergence of multi-modality due to slow promoter kinetics.} (a) A \red{diagram of a} self-repressing gene. (b) \red{The stationary probability distribution for different $\varepsilon$ which shows the} transition from fast promoter kinetics,  {i.e., $\varepsilon \to \infty $}, to slow promoter kinetics,  {i.e., $\varepsilon \to 0 $}, in a non-cooperative self-repressing gene.  {The stationary distribution is bimodal for small $\varepsilon$ and unimodal for large $\varepsilon$.} The deterministic equilibrium coincides with the fast kinetics mode,  {refer to SI-\S 5.1}. \red{The slow kinetic limit is calculated via \eqref{e.mixture_m}, while the fast kinetics limit is calculated by iterating the recurrence relation in Proposition SI-2. The remaining curves are computed by a finite projection solution \cite{khammash06} of the master equation.}   (c) A \red{diagram of} a repression-activation  {two-node network}. (d) Slow promoter kinetics gives rise to four modes while the deterministic model admits a unique stable equilibrium, \red{refer to SI-\S 5.2}. \red{The surface is plotted using \eqref{e.mixture_m}}. }

\label{f.intro_example}
\end{figure}

Here, we pursue a mathematical analysis of the role of slow promoter kinetics
in producing multi-modality in GRNs \red{and show analytically how the shape of the stationary distribution is dictated by key biochemical parameters}.
Previous studies of the chemical master equation (CME)
for single genes has already observed the emergence of bimodality with slow
TF-gene binding/unbinding \cite{hornos05,shahrezaei08,qian09,iyer09}. This phenomenon
was also studied by taking the limit of slow promot{e}r kinetics using the linear noise approximation \cite{grima14} or hybrid stochastic models  of gene expression \cite{potoyan15,liu15}. \red{The basic constitutive gene expression model (refer to \eqref{aut_swi})
has been validated for transcriptional bursting \cite{raj06}}.
However, and despite
its application relevance, mathematical analysis of the CME for multi-gene networks
with slow promoter kinetics has been missing, and only numerical solutions of the CME have been reported  \cite{feng12,chen16}.

{
In this work, an underlying theoretical contribution is the partitioning of the state space into weakly-coupled
ergodic classes \cite{norris} which, in the limit of slow binding/unbinding, results in the
reduction of the infinite-dimensional Markov chain into a finite-dimensional
chain whose states correspond to ``{promoter} states''.  In this limit, the
stationary distribution of the network can be expressed as a mixture
of Poisson distributions, each corresponding to conditioning
 the chain on a certain promoter configuration.
The framework proposed here enables us to analytically determine how the number of modes, their locations, and weights depend on the biophysical parameters.  Hence, the proposed framework can be applied to GRNs to predict the
different phenotypes that the network can exhibit with low gene copy numbers
and slow promoter kinetics.

The results are derived by introducing a new formalism to model GRNs with
arbitrary numbers of genes, based on continuous-time Markov chains. Then,
we analyze the stationary solution of the associated CME through a
systematic application of the method of singular perturbations \cite{campbell79}.}  {Specifically, we study
the slow promoter kinetics limit by letting the ratio of kinetic rate constants of the TF-gene binding/unbinding reactions with respect to protein reactions
approach zero}. \red{The stationary solution is computed by applying the method of
singular perturbations to the CME.}

In order to illustrate the practical significance of our results, we work out several
examples, some of which have not been studied before in the literature.  As a
first application, we discover that, with slow promoter kinetics, a self-regulating gene can
exhibit bimodality even with non-cooperative binding to the promoter site. 
We then investigate the role of cooperativity.
\emr{In contrast to} deterministic systems, we find that {cooperativity} does not change the number of modes.  Nevertheless, {cooperativity} adds extra
degrees of freedom by allowing the network to tune the relative weight of each
mode without changing its location.

As a second application, we revisit the classical toggle switch, under slow
TF-gene binding/unbinding.  It has been reported before that, with fast
TF-gene binding/unbinding, the
toggle switch with single-gene copies can be ``bistable'' without cooperative
binding \cite{lipshtat06}. We show that this can also happen with slow
prom{oter} kinetics, and{, moreover,} that a new mode having
both proteins {at} high copy numbers can
emerge. \emr{We provide a method to calculate the weight of each mode and show that the third mode is suppressed for sufficiently high kinetic rates for the dimerization reactions.}

A third application that we consider is a simplified model of
 {synchronization of communicating}
toggle switches. In bacterial populations, quorum sensing has been proposed
\cite{miller01} as a way for bacterial cells to broadcast their internal
states to other cells {in order} to facilitate synchronization. Quorum sensing
communication has been adopted also as a tool in synthetic biology
\cite{strogatz04,collins04}. Mathematical analysis of coupled toggle switches
designs usually employs deterministic models \cite{sontag16}. We study a
{simplified}
stochastic model of coupled toggle switches with slow promoter kinetics and
compare the resulting number of modes with deterministic equilibria.

Our final, and potentially most significant, application is motivated by cellular differentiation.
A well-known metaphor for cell lineage specification arose from the 1957
work of Waddington
\cite{waddington},
who imagined an ``epigenetic landscape'' with a series of branching valleys
and ridges depicting stable cellular states.
In that context, the emergence of new modes in cell fate circuits is often
interpreted as the creation of new valleys in the epigenetic landscape, and
(deterministic) multi-stability is employed to explain cellular
differentiation \cite{zhou11}.
However, an increasing number of studies have suggested stochastic
heterogeneous gene expression as a mechanism for differentiation
\cite{eldar10,balazsi11,norman15}.  Numerical analysis of the CME for the  canonical cell-fate circuit have shown the emergence of
new modes due to slow promoter kinetics in such models \cite{feng12,chu17}. \emr{This general category of cell-fate circuits includes pairs
such as PU.1:GATA1, Pax5:C/EBP$\rm\alpha$ {and} GATA3:T-bet
\cite{graf09}.
Cell fate circuits are characterized by TF cross-antagonism. However, their behavior is affected by the promoter configurations available for binding, the cooperativity index of the TFs, and the relative ratio of production rates. Hence, we study two models that differ in the aforementioned aspects and we highlight the differences between our findings and the behavior predicted by the corresponding deterministic model.  The first model employs independent cooperative binding. {We show that such a network can exhibit more than four modes. In contrast, the deterministic model predicts up to four modes only with cooperativity \cite{elife} .} The second network is a PU.1/GATA.1 network which employs non-cooperative binding and a restricted set of promoter configurations. The deterministic model is monostable, while the parameters of the stochastic model can be chosen to have additional modes including the cases of bistability and tristability.}

\rv{Although we formulate our study in terms of steady state probability distributions, one may equally well view our results as describing the typical dynamic behavior of realizations of the stochastic process.  These recapitulate the form of the steady state distributions: modes are reflected in metastable states along sample paths, states in which the system will stay for prolonged periods until switching to other states corresponding to alternative modes.  In the SI, we provide Monte-Carlo simulations showing such metastable behavior along sample paths.  We do so for the toggle switch as well as for a version of a well-studied genetic circuit \cite{elowitz00} which exhibits oscillatory behavior along sample paths even though the corresponding deterministic model cannot admit oscillations.}

\subsection*{The Reaction Network Structure}

In this paper, a GRN will be formally defined as
{a set of nodes (genes) that are connected with each other through regulatory interactions via the proteins that the genes express. The regulatory proteins are called \emph{transcription factors} (TFs). A TF regulates the expression of a gene by reversibly binding to the gene's promoter and by either enhancing expression or repressing it.}

The formalism we employ in order to describe GRNs at the elementary level is that of Chemical Reaction Networks (CRNs) \cite{erdi}.  A CRN consists of \emph{species} and \emph{reactions}, which we describe below.

\paragraph{Species:}
  The species in our context consist of promoter   configurations for the various genes
participating in the network, together with the respective TFs expressed from these genes and
some of their multimers. A  configuration of a promoter is  characterized by the possible locations and number of TFs bound to the promoter at a given time. If a promoter is expressed constitutively, then there are two   configurations {specifying} the expression activity state, active or inactive. A multimer is a compound consisting of a protein binding to itself several times. For instance, dimers and trimers are 2-mers and 3-mers, respectively. If a protein forms an $n^{\rm th}$-order multimer then we say that it has a cooperativity index of $n$. If species is denoted by $\rm X$, then its copy number is denoted by $X$.

 For simplicity we assume the following:
 \begin{enumerate}
 \item[\textbf{A1}] Each  {promoter} can have up to {two TFs binding to it.}
     \item[\textbf{A2}] Each {TF}  {is a single protein} that has a fixed cooperativity index, i.e, it cannot act as a TF with two different cooperativity indices.
  \item[\textbf{A3}] Each gene is present with only a \emph{single copy}.
  \end{enumerate}

   All the above assumptions can be relaxed. We make these assumptions only in order to simplify the notations and mathematical derivations.  \S SI-4,5 contain generalizations of the results to heterogenous TFs, and arbitrary numbers of gene copy numbers.

Consider  the $i^{\rm th}$ promoter. The expression rate of a gene is dependent on {the current configuration of its promoter}. We call the set of all possible such configurations the \emph{binding-site set} $B_i$. Each member of $B_i$ corresponds to a configuration \red{that translates into a specific species} $\mathrm D_j^i, j \in B_i$. If a promoter has just one or no regulatory binding sites, then we let $B_i=\{0,1\}$. Hence, the promoter configuration can be represented by \emph{two species}: the unbound species $\mathrm D_0^i$ and the bound species $\mathrm D_1^i$. \red{If the promoter has no binding sites then the promotor configuration species are interpreted as the inactive and active configurations, respectively.}   On the other hand, if the promoter has \emph{two binding} sites then $B_i=\{00,01,10,11\}$\footnote{We interpret the elements of the binding set as integers in binary representation.}. The first digit in a member of $B_i$ specifies whether the first binding site is occupied, and the second digit specifies the occupancy of the second binding site. Hence, the promoter configuration can be represented by four species \emph{$\mathrm D_{00}^i,\mathrm D_{10}^i,\mathrm D_{01}^i,\mathrm D_{11}^i$}. Note that in general we need to define $2^\kappa$ species for a promoter with $\kappa$ binding sites.

The species that denotes the protein produced by the $i^{\rm th}$ gene is  $\mathrm X_i$. A protein's multimer is denoted by $\mathrm X_{ic}$. If protein $\mathrm X_i$ does not form a multimer then $\mathrm X_{ic}:=\mathrm X_i$.

Therefore, the set of species in the network is $\mathscr S= \bigcup_i \left ( \{\mathrm D_j^i, j \in B_i\} \cup \{\mathrm X_i,\mathrm X_{ic}\}\right ).$

\paragraph{Reactions:} In our context, the reactions consist of TFs binding and unbinding with promoters and the respective protein expression (with transcription and translation combined in one step), decay, and $n$-merization.

For each gene, we define a \emph{gene expression block}. Each block consists of a set of \emph{gene reactions} and a set of \emph{protein reactions} as shown in Figure \ref{f.block}.

\begin{figure}
\centering
\includegraphics[width=0.5\columnwidth]{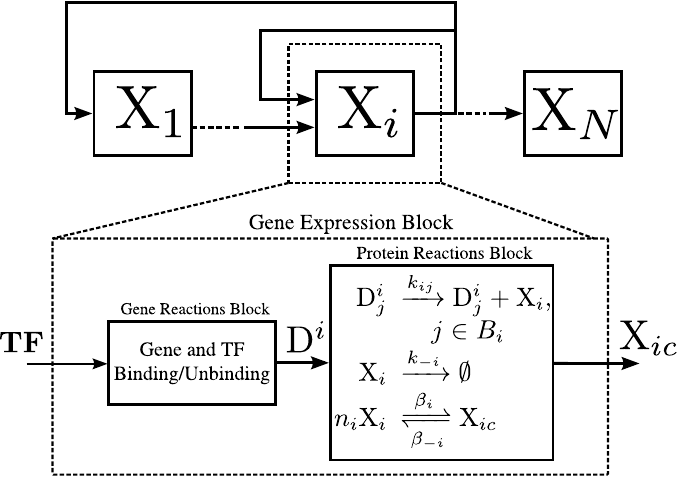}
 \caption{\textbf{A Gene Expression Block}. \red{A generic regulatory network that consists of gene expression blocks. A block consists of a gene reactions block and a protein reactions block. The gene reactions are described in the text.} \textbf{TF} is a vector of TFs which  can be monomers, dimers, or higher order multimers.  {$\mathrm D^i$ is a vector whose components consist of the $\mathrm D_j^{i}$'s}. The dimension of \textbf{TF} is equal to the number of binding sites of the gene. }
   \label{f.block}
 \end{figure}

If promoter is constitutive, i.e. it switches  between two configurations {autonomously} without an explicitly modeled TF-promoter binding, then   $B_i=\{0,1\}$ and the gene  {reactions} block consists of:
 \begin{equation}\label{e.ba} \mathrm D_0^i \xrightleftharpoons[ \alpha_{-i}]{ \alpha_{i}}  \mathrm D_1^i.\end{equation}
 We refer to $\mathrm D_0^i$ and $D_1^i$ as the \emph{inactive} and \emph{active} configurations, respectively.
If the promoter has one binding site, then also $B_i=\{0,1\}$ and the gene {reactions}  block consists of just two reactions:
 \begin{equation}\label{e.b0} {\rm TF}+\mathrm D_0^i \xrightleftharpoons[ \alpha_{-i}]{ \alpha_{i}}  \mathrm D_1^i,\end{equation}
where $\mathrm D_0^i$ and $D_1^i$ are {are the promoter configurations when \emph{unbound} and \emph{bound} to the TF}, respectively. Note that we did not designate a specific species as the active one since it depends on whether the TF is an activator or a repressor. \red{Specifically, when TF is an activator, $D_1^i$ will be the active configuration and $D_0^i$ will be the inactive configuration, and vice versa when TF is a repressor.}

 Finally, if the promoter has two \red{TFs binding to it, then they can bind \emph{independently}, \emph{competitively}, or \emph{cooperatively}}. {Cooperative binding is discussed in SI-\S 4.3.1.} \red{If they bind independently, then the promoter has two}  binding sites. Hence, $B_i=\{00,01,10,11\}$ and the gene  block {contains the following reactions}:
 \begin{align}\label{e.b00x} {\rm TF}_1+\mathrm D_{00}^i &\xrightleftharpoons[ \alpha_{-i1}]{ \alpha_{i1}}  \mathrm  D_{10}^i \\ \label{e.b01} {\rm TF}_1+\mathrm D_{01}^i& \xrightleftharpoons[ \alpha_{-i2}]{ \alpha_{i2}}  \mathrm D_{11}^i,\\ \label{e.b00y} {\rm TF}_2+  \mathrm  D_{00}^i& \xrightleftharpoons[ \alpha_{-i3}]{ \alpha_{i3}}  \mathrm  D_{01}^i,\\ \label{e.b10} {\rm TF}_2+ \mathrm D_{10}^i& \xrightleftharpoons[ \alpha_{-i4}]{ \alpha_{i4}}  \mathrm  D_{11}^i.\end{align}
The activity of each configuration species is dependent on {whether the TFs are activators or repressors}, and \red{ on how they behave jointly. This can be characterized fully by assigning a production rate for each configuration as will be explained below}.

\red{In the case of competitive binding, two different TFs compete to bind to the same location. This can be}  {modeled similarly} \red{to the previous case except that the transitions to $\mathrm D_{11}^i$, i.e. the configuration where both TFs are bound, are not allowed. Hence, the gene reactions block will have only two reactions: \eqref{e.b00x}, \eqref{e.b00y}, and the binding set reduces to $B_i=\{00,01,10\}$.}

  We assume that RNA polymerase and ribosomes are available in high copy numbers, and that we can lump transcription and translation  {into} one simplified ``production'' reaction.  The rate of production is dependent on the promoter's configuration. So for each configuration $\mathrm D_j^i, j \in B_i$ the production reaction is:
  \begin{equation}\label{production} \mathrm D_j^i \xrightarrow{k_{ij}} \mathrm D_j^i+\mathrm X_i,\end{equation}
 where the kinetic constant $k_{ij}$ is a non-negative number.  \red{The case} $k_{ij}=0$ means that when the promoter configuration is $\mathrm D_j^i$ there is no protein production, and hence $D_j^i$ is an inactive configuration. \red{The promoter configuration can be ranked from the most active to the least active by ranking the corresponding production kinetic rate}  {constants}.

\red{Consequently}, the character of a TF is manifested  {as follows}:  if the \red{maximal} protein production  {occurs at a configuration with the TF being bound we say that} the TF is \emph{activating}, and if the reverse holds it is \emph{repressing}.  {And, if the production is maximal with multiple configurations such that the TF is bound in some of them and unbound in others then the TF is neither repressing nor activating.}

   We model decay and/or dilution as a single reaction:
 \begin{equation}\label{e.decay} \mathrm X_i \xrightarrow{k_{-i}} \emptyset. \end{equation}

 The expressed proteins can act as TFs. They may combine to form dimers or higher order multimers \red{before acting as TFs}. \red{The numbers of copies of the TF needed to form a multi-mer is called the \emph{the cooperativity index} and we denote it by $n$}. Hence, we model the cooperativity reactions as given in Figure \ref{f.block} as follows:
  \begin{equation}\label{dimerization} n \mathrm X_i \xrightleftharpoons[ \beta_{-i}]{ \beta_{i}}  \mathrm X_{ic}. \end{equation}

 {If the cooperativity index of $\mathrm X_i$ is 1, then the species $\mathrm X_{ic}:=\mathrm X_i$, and the multimerization reaction becomes empty.}

\rv{Higher order multi-merization processes can be modelled as multi-step or sequential reactions \cite{ferrell09}. We discuss how our theory includes this case in \S SI-4.3.3,
by showing how an equivalent one-step model with can be formulated.}

\paragraph{Kinetics:}

   In order to keep track of
molecule counts, each species $\mathrm Z_i \in \mathscr S$ is associated with a copy
number $z_i \in \mathbb Z_{\ge 0}$.

To each reaction $\mathrm R_j$ one associates a propensity function $R_j$. We use \emph{Mass-Action Kinetics}.  If $\mathrm R_j$ has a single reactant species  $Z_i$ with stoichiometry coefficient $\alpha_{ij}$, then \cite{anderson15}:
\[ R_j(z_i) = k_j \binom{z_i}{\alpha_{ij}}=k_j \frac{z_i(z_i -1 ) ... (z_i -\alpha_{ij})}{\alpha_{ij}!},\]
where $k_j$ is a kinetic rate constant. Note if $\alpha_{ij}=1$, then $R_j(z_i)=k_j z_i$.

{The only bimolecular reaction we need is binding of a TF to a promoter, which has}  unity stoichiometry coefficients  {for} each reactant species, i.e. the left side of the reaction is of the form  $\mathrm Z_{i_1}+\mathrm Z_{i_2}$. In this case, the propensity function is:
\[ R_j(z_{i_1},z_{i_2}) = k_j z_{i_1} z_{i_2}.\]

\paragraph{A gene regulatory network:} Consider a set of $N$ genes, binding sets $\{B_i\}_{i=1}^N$, and  kinetic constants $k_j$'s. A \emph{gene expression block}\red{, as shown in Figure \ref{f.block},} is a set of gene reactions and protein reactions as defined above. Each gene block has an output {that} is either   the protein or its $n$-mer, and it is designated by $\mathrm X_{ic}$. The input to each gene expression block is a subset of the set of the outputs of all blocks.
Then, a GRN is an arbitrary interconnection of a gene expression blocks {(Figure \ref{f.block})}. SI-\S 4 defines a more general class of network that we can study.

A {directed} graph can be associated with a GRN as follows. Each vertex corresponds to a gene expression block. There is a directed edge from vertex A to vertex B if the output of A is an input to B.
In order to simplify the presentation, we assume the following:
 \begin{enumerate}
    \item[\textbf{A4}] The graph of gene expression blocks is connected.
 \end{enumerate}
 Note that if A4 is violated, our analysis can be applied to each connected component.

 \paragraph{Time-Scale Separation:}
 As mentioned in the introduction, we assume that the gene reactions \eqref{e.ba}-\eqref{e.b10} are considerably slower than the protein reactions \eqref{production}-\eqref{dimerization}. In order to model this assumption, we write the kinetic rates of gene reactions in the form $\varepsilon k_j$, where $0<\varepsilon \ll 1$ and assume that all other kinetic rates (for protein production, decay and multi-merization) are { $\varepsilon^{-1}$-times faster}.

Events in biological cells usually take place at different time-scales \cite{alon}, and hence singular perturbation techniques are widely used in deterministic settings in order to reduce models for analysis. On the other hand, model-order reduction by time-scale separation in stochastic processes has been mainly used in the literature for computational purposes, for example to accelerate the stochastic simulation algorithm \cite{rao03,sontag17}, or to compute finite-space-projection solutions to the CME \cite{khammash06}. In this work, we use a singular perturbation approach for the analytical purpose of characterizing the form of the stationary distribution in the regimes of slow gene-TF binding/unbinding.

   In the case of a finite Markov chain,  the CME is a finite-dimensional linear ODE, and reduction methods for linear systems can be used  \cite{campbell79} and applied to Markov chains \cite{kokotovic81,yin97}. For continuous-time Markov chains on a countable space, as needed when analyzing gene networks,  there are difficult and open technical issues.  Exponential stochastic stability \cite{meyn93} needs to be established for the stationary solution in order {to guarantee the existence of the asymptotic expansion in $\varepsilon$} \cite{altman04}. Although it has been shown for a class of networks \cite{gupta14}, the general problem needs further research.
   In this paper, we will not delve into  technical issues of stochastic stability {; we} assume that these expansions exist and \red{that} the solutions converge to a unique equilibrium solution.

 \subsection*{Dynamics and the Master Equation}
 The dynamics of the network refers to the manner in which the \emph{state} evolves in time, where  {the} state $Z(t) \in \mathbf Z \subset \mathbb Z_{\ge 0}^{|\mathscr S|}$ is the vector of  copy numbers of the species of the network at time $t$. The standard stochastic model for a CRN is that of a continuous Markov chain. Let $\mathbf Z$ denote the state space. Consider a time $t$ and let the state be $Z(t)=z \in \mathbf Z$. The relevant background is reviewed in SI-\S 1.1.

Let $p_z(t)=\Pr[Z(t)=z|Z(0)=z_0]$ be the stationary distribution for any given initial condition $z_0$.   Its time evolution is given by the \emph{Chemical Master Equation} (CME).

Since our species are either gene species or protein species, we split the stochastic process $Z(t)$ into two subprocesses: \emph{the gene process} $D(t)$ and \emph{the protein process} $X(t)$, as explained below.

  For each gene we define one process $D_i$ such that $D_i(t)
 \in B_i$.  $D_i(t)=j$ if and only the promoter configuration is encoded by $j \in B_i$.
Collecting these into a vector, define the gene process $D(t):=[D_1(t),...,D_N(t)]^T$ where $D(t) \in \prod_{i=1}^N B_i $. The $i^{\rm th}$ gene can be represented by $|B_i|$ states, so $L:={\prod_{i=1}^N |B_i|}$ is  the total number of promoter configurations in the GRN. {With} abuse of notation, we write also $D(t) \in \{0,..,L-1\}$ in the sense of the bijection between $\{0,..,L-1\}$ and $\prod_{i=1}^N B_i $ defined by interpreting $D_1...D_N$ as a binary representation of an integer. Hence, $d \in \{0,..,L-1\}$ corresponds to $(d_1,...,d_N) \in B_1 \times .. \times B_N$ and we write $d=(d_1,..,d_N)$.

 Since each gene expresses a corresponding protein,  we define $X_{i1}(t)\in \mathbb Z_{\ge 0}, i =1,..,N$ protein processes. If the multimerized version of the $i^{\rm th}$ protein participates in the network as an activator or repressor then we define $X_{ic}(t)$ as the corresponding multimerized protein process, and we denote $X_i(t):=[X_{i1}(t),X_{ic}(t)]^T$. If there is no multimerization reaction then we define $X_i(t):=X_{i1}(t)$. Since not all proteins are necessarily multimerized, the total number of protein processes is  $N \le M \le 2N$. Hence,
   the \emph{protein  process} is $X(t)=[X_1^T(t),..,X_N^T(t)]^T \in {\mathbb Z}_{\ge 0}^M $ and the state space can be written as $\mathbf Z= {\mathbb Z}_{\ge 0}^M \times \prod_{i=1}^N B_i$.

\section*{Results}

\subsection*{Decomposition of the Master Equation}

It is crucial to our analysis to represent the linear system of differential equations given by  the CME  as an interconnection of weakly coupled linear systems. \red{To this end,} we present the appropriate notation in this subsection.

Consider the joint probability distribution:
\begin{equation}\label{e.jointpdf}
  p_{d,x}(t)= \Pr[X(t)=x, D(t)=d],
\end{equation}
\red{which represents the probability at time $t$ that the protein process $X$ takes the value} $x \in {\mathbb Z}_+^M$ \red{and the gene process $D$ takes the value} $d\in \{0,..,L-1\}$.  {Recall} \red{that $x$ is a vector of copy numbers for the protein processes while $d$ encodes the configuration of each promoter in the network.} Then, we can define for each fixed $d$:
\begin{equation}\label{p_d} p_d (t):= [p_{dx_0}(t),p_{dx_1}(t),....]^T, \end{equation}
\red{representing the vector enumerating  the probabilities \eqref{e.jointpdf} for all values of $x$ and for a fixed $d$, }where $x_0,x_1,..$ is an indexing of ${\mathbb Z}_{\ge 0}^M$. Note that $p_d(t)$ can be thought of as an infinite vector with respect to the aforementioned indexing. Finally, let \begin{equation}\label{e.vectordecomp} p(t):=[p_0(t)^T,...,p_{L-1}^T(t)]^T \end{equation}
\red{representing a concatenation of the vectors \eqref{p_d} for $d=0,..,L-1$.} Note that $p(t)$ is a finite concatenation of infinite vectors.

\red{The joint stationary distribution $\bar\pi$ is defined as the following limit,  {which we assume to exist and is independent of the initial distribution:}
\begin{equation}\label{pi_def}
\bar\pi = \lim_{t \to \infty } p(t).
\end{equation}}
Note that $\bar\pi$ depends on $\varepsilon$.

Consider a given GRN.  The CME \red{is defined over a countable state space $\mathbf Z$. Hence, the  CME can be interpreted as an infinite system of} differential equations with an infinite infinitesimal generator matrix $\Lambda$ which contains the reaction rates (see SI-\S 1.1).

\red{Consider} partitioning the probability distribution vector as in \eqref{e.vectordecomp}. Recall that reactions have been divided into two sets: slow gene reactions \eqref{e.ba}-\eqref{e.b10} and fast protein reactions \eqref{production}-\eqref{dimerization}. \rv{This allows us to write $\Lambda$ as a sum of a slow matrix $\varepsilon\hat \Lambda$ and a fast matrix $\tilde\Lambda$, which we call a fast-slow decomposition. Furthermore,  $\tilde \Lambda$ can be written as a block diagonal matrix with $L$ diagonal blocks which correspond to conditioning the Markov chain on a specific gene state $d$}. This is stated in the following basic proposition (see SI-\S 2.1 for the proof):
\begin{proposition}\label{th.decomp}
Given a GRN. Its CME can be written as
\begin{align}\label{e.masterinfinite}
\dot p(t)&=\Lambda_\varepsilon p(t)=\left (  \tilde\Lambda + \varepsilon \hat\Lambda\right ) p(t), \end{align}
 {
where \begin{align}\label{e.masterinfinite2} \tilde \Lambda  =  \begin{bmatrix} \Lambda_0 & & \\ & \ddots & \\ & & \Lambda_{L-1} \end{bmatrix}, \ \text{and} \ p(t)= \begin{bmatrix}p_{0}(t) \\ \vdots \\ p_{L-1}(t) \end{bmatrix},
\end{align}}
where   $\tilde \Lambda$ is the  fast matrix,  $\hat \Lambda$ is the slow matrix, and $\Lambda_0, .. , \Lambda_{L-1}$ are stochastic matrices.
\end{proposition}

\subsection*{Conditional Markov Chains} \red{For each $d$, consider modifying the Markov chain}  {$Z(t)$} \red{defined in the previous section by replacing the stochastic process $D(t)$ by a deterministic constant process $D(t)=d$. This means that the resulting}  {chain} \red{does not describe the gene process dynamics, it only describes the protein process dynamics \emph{conditioned on} $d$.}  {Henceforth}, \red{we refer to the resulting Markov chain as the \emph{Markov chain conditioned on $d$.} The infinitesimal generator of a chain conditioned on $d$ is }  {denoted} \red{by $\Lambda_d$,}  {and is identical to the corresponding block on the diagonal of $\tilde\Lambda$ as given in \eqref{e.masterinfinite2}.}
 In other words, fixing $D(t)=d \in \{0,..,L-1\}$, the dynamics of the network can be described by a CME:
\begin{equation} \label{e.subgenerator}\red{\dot p_{X|d} = \Lambda_d p_{X|d}, } \end{equation}
where \red{$p_{X|d}$ is a vector that enumerates the conditional probabilities $p_{x|d}=\Pr[X(t)=x | D(t)=d]$ for a given $d$}. \red{ The conditional stationary distribution}  {is denoted by}: \red{ $\pi_{X|d}^{(J)}=\lim_{t \to \infty} p_{X|d}(t)$, where $(J)$ refers to the fact that it is joint in the protein and multimerized protein processes. Note that $\pi_{X|d}^{(J)}$ is independent of $\varepsilon$.}
{This notion of a conditional Markov chain is useful since}, {at the slow promoter kinetics limit}, $D(t)$ stays constant. \red{It can be noted from} {\eqref{e.masterinfinite2}} \red{that when $\varepsilon=0$ the dynamics of \red{$p_{d}$ decouples and becomes independent of $p_{\tilde d}, \tilde d=0,..,L-1, \tilde d \ne d$}.}

We show below that each conditional Markov chain has a simple  structure. Fixing \red{the promoter configuration} $D(t)=d=(d_1,..,d_N)$, the network consists of \emph{uncoupled} birth-death processes.
So for each $d_i$, {the protein reactions \eqref{production}-\eqref{dimerization} corresponding to the $i^{\rm th}$ promoter can be written as follows without multimerization:}
\begin{equation}\label{e.birthdeath}
  \emptyset \xrightleftharpoons[  k_{-i}]{  k_{id_i}} \mathrm X_i,
\end{equation}
 {where the subscript $id_i$ refers to the production kinetic constant corresponding to the configuration species $\mathrm D_{d_i}^i$,}
or, if there is a multimerization reaction, it takes the form:
\begin{equation}\label{e.birthdeathDim}
  \emptyset \xrightleftharpoons[ k_{-i}]{ k_{id_i}} \mathrm X_i, \ n_i \mathrm X_i \xrightleftharpoons[  \beta_{-i}]{ \beta_{i}} \mathrm X_{ic}.
  \end{equation}
Note that the stochastic processes $\mathrm X_i(t),i=1,..,N$ conditioned on $D(t)=d$ are independent of each other. Hence, the conditional stationary distribution $\pi_{X|d}^{(J)}$ can be written as a product of stationary distributions and the individual stationary distributions have Poisson expressions. The following proposition gives the analytic expression of the conditional stationary distributions: (see SI-\S 2.2 for proof)
\begin{proposition}\label{th.cond_dis} Fix $d \in \{0,..,L-1\}$. Consider \eqref{e.subgenerator}, then there exists a conditional stationary distribution $\pi_{X|d}^{(J)}$ and it is given by
\begin{equation}\label{e.cond_dis} \pi^{(J)}_{X|d}(x)= \prod_{i=1}^N \pi_{X|di}(x_i),\end{equation}
where
\begin{equation}\label{e.cond_dis2} \pi^{(J)}_{X|d_i}(x_i)= \left \{ \begin{array}{ll}\displaystyle \mathbf P\left ( x_{i1},x_{i2}; \frac{k_{id_i}}{k_{-i}}, \frac{k_{id_i}^{n_i}\beta_i}{n_i! k_{-i}^{n_i} \beta_{-i}} \right ) 
& \mbox{if}\  \mathrm X_i \  \mbox{is multimerized} \\
\mathbf P\left ( x_{i}; \frac{k_{id_i}}{k_{-i}} \right ) ,& \mbox{otherwise}
\end{array}\right . ,\end{equation}
where  {$(J)$ refers to the joint distribution in multimerized and non-multimerized processes, $x_{i1}$ refers to the copy number of $\mathrm X_i$, while $x_{i2}$ refers to the copy number of $\mathrm X_{ic}$,} $\mathbf P(x;a):= \frac {a^x}{x!} e^{-a},  \mathbf P(x_1,x_2;a_1,a_2):= \frac {a_1^{x_1}}{x_1!}\frac {a_2^{x_2}}{x_2!} e^{-a_1-a_2}$.
\end{proposition}

\begin{remark} \label{rem.marginal} The conditional distribution in \eqref{e.cond_dis} is a joint distribution in the protein and multimerized protein processes. If we want to compute a marginal stationary distribution for the protein process only, then we average over the multimerized protein processes $X_{ic}, i=1,..,N$ to get a joint Poisson in $N$ variables. Hence, the formulae \eqref{e.cond_dis} {-\eqref{e.cond_dis2}} can be replaced by:
\begin{equation}\label{e.marginal_conditional}
 \pi_{X|d}(x):=  {\sum_{i=1}^{M-N}\sum_{x_{i2}=0}^\infty \pi_{X|d}^{(J)}(x)}= \prod_{i=1}^N \mathbf P\left (x_i;\frac{k_{id_i}}{k_{-i}}\right ),
 \end{equation}
 \red{where $M-N$ is the number of $n$-merized protein processes, and $\pi_{X|d}$  is the marginal stationary distribution for the protein process}.
\end{remark}

\subsection*{Decomposition of The Stationary Distribution}
 Recall the slow-fast decomposition \red{of the CME} in \eqref{e.masterinfinite} and the joint stationary distribution \eqref{pi_def}. In order to emphasize the dependence on $\varepsilon$ we \red{denote $\bar\pi^\varepsilon:=\bar\pi(\varepsilon)$}. \red{Hence,} $\bar\pi^\varepsilon$ is the unique stationary distribution that satisfies
 $ \Lambda_\varepsilon \bar\pi^\varepsilon=0$, $\pi^\varepsilon>0$, and $\sum_z \pi_z^\varepsilon = 1 $,  {where the subscript denotes the value of the stationary distribution at $z$}.

  Our objective is to characterize the stationary distribution as $\varepsilon \to 0$.
 Writing $\bar\pi_\varepsilon$ as an asymptotic expansion  {to first order} in terms of  $\varepsilon$, we have
\begin{equation}\label{e.expansion} \bar\pi^\varepsilon = \bar\pi^{(0)} + \bar\pi^{(1)} \varepsilon + o(\varepsilon).\end{equation}

Our aim is to find $\bar\pi^{(0)}$. We use singular perturbations techniques to derive the following theorem (see SI-\S 2.3):

\begin{theorem}\label{th} Consider a \red{given} GRN with $L$ \rv{promoter states} with the CME \eqref{e.masterinfinite}. Writing \eqref{e.expansion}, then the joint stationary distribution  {$\bar\pi:= \lim_{\varepsilon \to 0^+} \bar\pi^{\varepsilon}$} {can be written as}:
\[\bar\pi(x,d) = \sum_{d=0}^{L-1} \lambda_d \bar\pi_{X|d}(x,d),\]
where $\lambda=[\lambda_0,..,\lambda_{L-1}]^T$ is the principal normalized eigenvector of:
\begin{equation}\label{e.solution} \Lambda_r :=\begin{bmatrix}\mathbf 1^T & \mathbf 0^T &...  &  \mathbf 0^T \\ \mathbf 0^T & \mathbf 1^T &...  &  \mathbf 0^T \\ & & \ddots & \\ \mathbf 0^T & \mathbf 0^T &...  &  \mathbf 1^T \end{bmatrix} \hat\Lambda \, [\bar\pi_{X|0} \ \bar\pi_{X|1} \ ... \ \bar\pi_{X|L-1} ]  ,
\end{equation}

where $\bar\pi_{X|0},...,\bar\pi_{X|L-1}$ are the extended conditional stationary distributions defined as: $\bar\pi_{X|d}(x,d)=\pi_{X|d}(x)$, and $\bar\pi_{X|d}(x,d')=0$ when $d'\ne d$.
\end{theorem}
 {
The result characterizes the stationary solution of \eqref{e.masterinfinite} which is a joint distribution in $X$ and $D$. However, we are particularly interested in the marginal stationary distribution of the protein process $X$ and the marginal stationary distribution of the non-multimerized protein process, since these distributions are typically experimentally observable.   Therefore, we can use Remark  \ref{rem.marginal} to write the stationary distribution as mixture of $L$ Poisson distributions with weights $\{\lambda_d\}_{d=0}^{L-1}$:}

\begin{corollary} \label{cor}  {
Consider a  {given} GRN with $L$ genes with the CME \eqref{e.masterinfinite}. Writing \eqref{e.expansion}, let $\pi_{X|0},...,\pi_{X|L-1}$ be the conditional stationary distributions of $\Lambda_0,...,\Lambda_{L-1}$, where explicit expressions are given in \eqref{e.cond_dis}. Then, we can write the following:
\begin{equation}\label{e.mixture}\pi^{(J)}(x):= \lim_{\varepsilon \to 0^+}\lim_{t\to\infty} \Pr[X(t)=x]= \sum_{d=0}^{L-1} \lambda_{d} \pi_{X|d}^{(J)}(x),\end{equation}
where $\lambda=[\lambda_0,..,\lambda_{L-1}]^T$ is as given Theorem \ref{th}.\\
Furthermore,  the marginal stationary distribution of the non-multimerized protein process can be written as:
\begin{equation}\label{e.mixture_m}\pi(x):= \sum_{d=0}^{L-1} \lambda_d \pi_{X|d}(x)=\sum_{d=0}^{L-1} \lambda_d \prod_{i=1}^N \mathbf P\left(x_i;\frac{k_{id_i}}{k_{-i}}\right).\end{equation}}
\end{corollary}

\begin{remark} In the remainder of the Results section, when we refer to the ``stationary distribution'' we mean the  marginal stationary distribution of the non-multimerized protein process given in \eqref{e.mixture_m}.
\end{remark}
\begin{remark}If a mode is defined as a local maximum of a stationary distribution, then this does not necessarily imply that the stationary distribution has $L$ modes since the peak values of two Poisson distributions can be very close to each other. In the remainder of the paper we will call each Poisson distribution in the mixture as a ``mode`` in the sense that it represents a component in the mixture distribution.  The number of local maxima of a distribution can be found easily given the expression \eqref{e.mixture_m}.\end{remark}

\subsection*{The Reduced-Order Finite Markov Chain}
\rv{The computation of the weighting vector $\lambda$ in Theorem \ref{th}
  requires computing the $L \times L$ matrix  $\Lambda_r$ in
  \eqref{e.solution} which can be interpreted as the infinitesimal generator
  of an $L$-dimensional Markov chain. The expression in \eqref{e.solution}
  involves evaluating the product of infinite dimensional matrices.    Since
  the structure of the GRN and the form of the conditional distribution in \eqref{e.cond_dis} are known,   { an easier algorithm to compute $\Lambda_r$} for our GRNs is given in Proposition SI-2. The algorithm provides an intuitive way to interpret Theorem \ref{th} and can be informally described as follows.}

   \rv{Assume $D(t)=d$, the algorithm implies that each binding reaction of the form:
\[ \mathrm{TF}+\mathrm D_{d_i}^i  \xrightarrow{\alpha} \mathrm D_{d_{i'}}^i,\]
gives the rate $\alpha \mathbb E [{TF}|D=d] $,  {where $\mathbb E$ denotes mathematical expectation.} Hence it corresponds to a reaction of the following form in the reduced-order Markov chain:
\begin{equation} \label{e.reversible_gene} \mathrm D_{d_i}^i  \xrightarrow{\alpha \mathbb E [ \text{TF}|D=d]} \mathrm D_{d_{i'}}^i.\end{equation}
Using Proposition \ref{th.cond_dis}, we can write:
\begin{equation}\label{e.rates}\mathbb E [ \text{TF}|D=d]={\frac{\alpha}{n_{ i}!}  \frac{ \beta_{ i  }}{  \beta_{- i}}\left (\frac{k_{ i d_{ i }}}{k_{- i}} \right )^{n_{ i}}}.
\end{equation}
 Refer to \S SI-2.4 for a more precise statement.}

\subsection*{Basic Example: Gene Bursting Model}
\label{s.aut_swi}
The simplest network is the {autonomous} TF-gene binding/unbinding model, and it has been used for transcriptional bursting \cite{raj06} and studied using time-scale separation in \cite{qian09,grima15}.
Consider:
  \begin{align}\label{aut_swi}
\nonumber \mathrm D_0 &\xrightleftharpoons[\varepsilon \alpha_{-}]{\varepsilon \alpha} \mathrm D_1 \\
  \mathrm D_1 &\mathop\rightarrow\limits^{k} \mathrm X+\mathrm D_1, \\ \nonumber
\mathrm X &  \mathop\rightarrow\limits^{k_-} 0.
\end{align}

Referring to Figure \ref{f.block}, we identify a single gene block with two states. Using \eqref{e.cond_dis}, the conditional stationary distributions are two Poissions at $0$ and $k/k_-$.
The reduced Markov chain is a binary Bernoulli process with a rate of $\alpha/(\alpha+\alpha_-)$.
Then the stationary distribution of $X$ can be written using \eqref{e.mixture_m} as: (see SI-\S 3.1)

\[\pi(x) =\frac{\alpha_{-}}{\alpha+\alpha_{-}}  \mathbf P(x;0) +  \frac{\alpha}{\alpha+\alpha_{-}} \mathbf P(x;k/k_-),\]
which is a bimodal distribution with peaks at 0 and $k/k_-$.
 {The fast promoter kinetics model is obtained, instead, by} reversing {the} time-scale separation  {such that the protein reactions become slow and gene reactions become fast. In that case, the resulting stationary distribution is} a Poisson with mean $\frac{\alpha}{\alpha+\alpha_-} \frac{k}{k_-}$ which is the same as the deterministic equilibrium (with the  conservation law $D_1(t)+\mathrm D_0(t)=1$). Although two models share the mean, the stationary distributions differ drastically.

\subsection*{The Role of Cooperativity}
 A TF is said to be cooperative if it acts only after it  forms a dimer or a higher-order $n$-mer that binds to the gene's promoter \cite{ferrell09}. In standard deterministic modelling, a cooperative activation  {changes} the {form of the quasi-steady state activation rate from} a Michaelis-Menten {function} into a Hill {function}.   {Cooperativity is often necessary for a} network to have   multiple equilibria in some kinetic parameter ranges.  For example, a non-cooperative self- {activating} gene can only be mono-stable, while  {its} cooperative counterpart can be multi-stable  for some parameters.

   {Corollary \ref{cor} and \eqref{e.rates}} show that cooperativity {plays in the context of} slow promoter kinetics a role  {that is very} different from the deterministic setting. This is since the stationary distribution is a mixture of $L$ Poisson processes,  {independent} of whether the activations are multimerized  {or not}, and  \eqref{e.marginal_conditional} are also independent of the dimerization rates.
 Nevertheless, the multi-merization can tune the weighting coefficients  in \eqref{e.rates}. In the non-cooperative case,  a certain mode can be made more probable only by changing either the location of the mode or {the} {dissociation} ratio (the ratio of the binding to unbinding kinetic constants). On the other hand, a multimerized TF gives extra tuning parameters, namely the multimerization ratio and the cooperativity index. Hence, a certain mode can be  {made} more or less probable by modifying the multimerization ratio and/or the cooperativity index  {without} changing the location of the peaks or the dissociation ratio.

  In order to illustrate the above idea, we analyze a self-regulating gene with slow promoter kinetics with and without cooperativity.

 \subsection*{A Self-Regulating Gene}

Consider a non-cooperative self-regulating gene:
 \begin{align}
\nonumber \mathrm X+\mathrm D_0 &\xrightleftharpoons[\varepsilon \alpha_{-}]{\varepsilon \alpha} \mathrm D_1 \\ \label{noncoop}
  \mathrm D_0 &\mathop\rightarrow\limits^{k_{0}} \mathrm D_0+\mathrm X, \\
  \mathrm D_1 &\mathop\rightarrow\limits^{k_{1}} \mathrm D_1+\mathrm X,\nonumber \\
\mathrm X &  \mathop\rightarrow\limits^{k_{-}} 0 \nonumber .
\end{align}
The network is activating if $k_{1}>k_{0}$, and repressing otherwise.

 {Referring to Figure \ref{f.block}}, this is a single-gene block with two states. 
  The reduced Markov chain is  a binary Bernoulli process with the rate $\alpha k_0/(\alpha k_-+\alpha k_0)$. Using \eqref{e.mixture_m} the stationary distribution is:
\begin{equation}\label{e.nc} \pi_1(x)= \frac {\alpha \rho_1}{\alpha_- +\alpha \rho} \mathbf P(x;k_1/k_-) + \frac {\alpha_-}{\alpha_-+\alpha \rho} \mathbf P(x;k_0/k_-), \end{equation}
where \begin{equation}\label{e.nc_rho}\rho=\mathbb E[X_2|D=0]=k_0/k_-.\end{equation}

Next, consider the same reaction network, but now with cooperativity:
 \begin{align}\label{coop}
\nonumber \mathrm X_2+\mathrm D_0 &\xrightleftharpoons[\varepsilon \alpha_{-}]{\varepsilon \alpha} \mathrm D_1 \\
  \mathrm D_0 &\mathop\rightarrow\limits^{k_{0}} D_0+\mathrm X, \\
  \mathrm D_1 &\mathop\rightarrow\limits^{k_{1}} \mathrm D_1+\mathrm X,\nonumber \\
\mathrm X &  \mathop\rightarrow\limits^{k_{-}} 0 \nonumber \\
\nonumber 2\mathrm X &\xrightleftharpoons[\beta_-]{\beta} \mathrm X_2.
\end{align}

In this case, the stationary distribution is still given by \eqref{e.nc} but the weighting parameter changes to \begin{equation}\label{e.c_rho}\rho_2=\mathbb E[X_2|D=0]=\frac{ k_0^2 \beta}{2k_-^2 \beta_-}.\end{equation}

In both cases the distribution  has modes  at $\frac{k_{1}}{k_-}$ and   $\frac{k_{0}}{k_-}$, where the height of the first mode is proportional to $\rho$.

Comparing \eqref{e.nc_rho} and \eqref{e.c_rho}, note that, in the non-cooperative case, if we want to increase the weight of the mode corresponding to the bound state keeping the dissociation ratio, then  the mode location needs to be changed.  {On the other hand,  the dimerization rates in \eqref{e.c_rho}}  can be used in order to tune the weights freely while keeping the modes and the binding to unbinding kinetic constants ratio unchanged. For instance,  the distribution can be made effectively unimodal with a sufficiently high dimerization ratio.

\paragraph{Comparison with the deterministic model:} \strut \\

Table \ref{table} compares the number of stable equilibria in the deterministic model with the number of modes in the stochastic model in the case of a single gene copy. It can be noted that there is no apparent correlation between the numbers of deterministic equilibria and stochastic modes.

Figure \ref{f.selfactivate_plot} depicts the transition from a unique mode with fast promoter kinetics to multiple modes with slow kinetics with cooperativity and leakiness  {for a self-activating gene}.

\begin{table}[t]
\centering
\begin{tabular}{ccccc}
\toprule & \multicolumn{2}{c}{\textbf{Non-Cooperative}} & \multicolumn{2}{c}{\textbf{Cooperative}}  \\
 & Leaky & {Non-Leaky}  & Leaky &{Non-Leaky} \\
 \toprule
 Stochastic  &  2 & 1   & 2 & 1  \\
 ({\scriptsize Slow promoter kinetics}) &  & {\scriptsize (at 0)}  & & {\scriptsize (at 0)} \\ \hline
 Deterministic&  1 &  {1} & 1-2 &  {1-2}\\ \bottomrule
\end{tabular}
\caption{Comparing the number of stable equilibria/modes for a self-regulating gene between stochastic with slow promoter kinetics and deterministic modelling frameworks.  {Details for reconstructing the table are given in SI-\S 5.3.
}   } \label{table}
\end{table}

\begin{figure}
\centering
\subfigure[]{\raisebox{0.8\height}{\includegraphics[width=0.09\columnwidth]{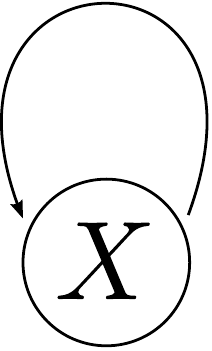}}}
\subfigure[]{\includegraphics[width=0.55\columnwidth]{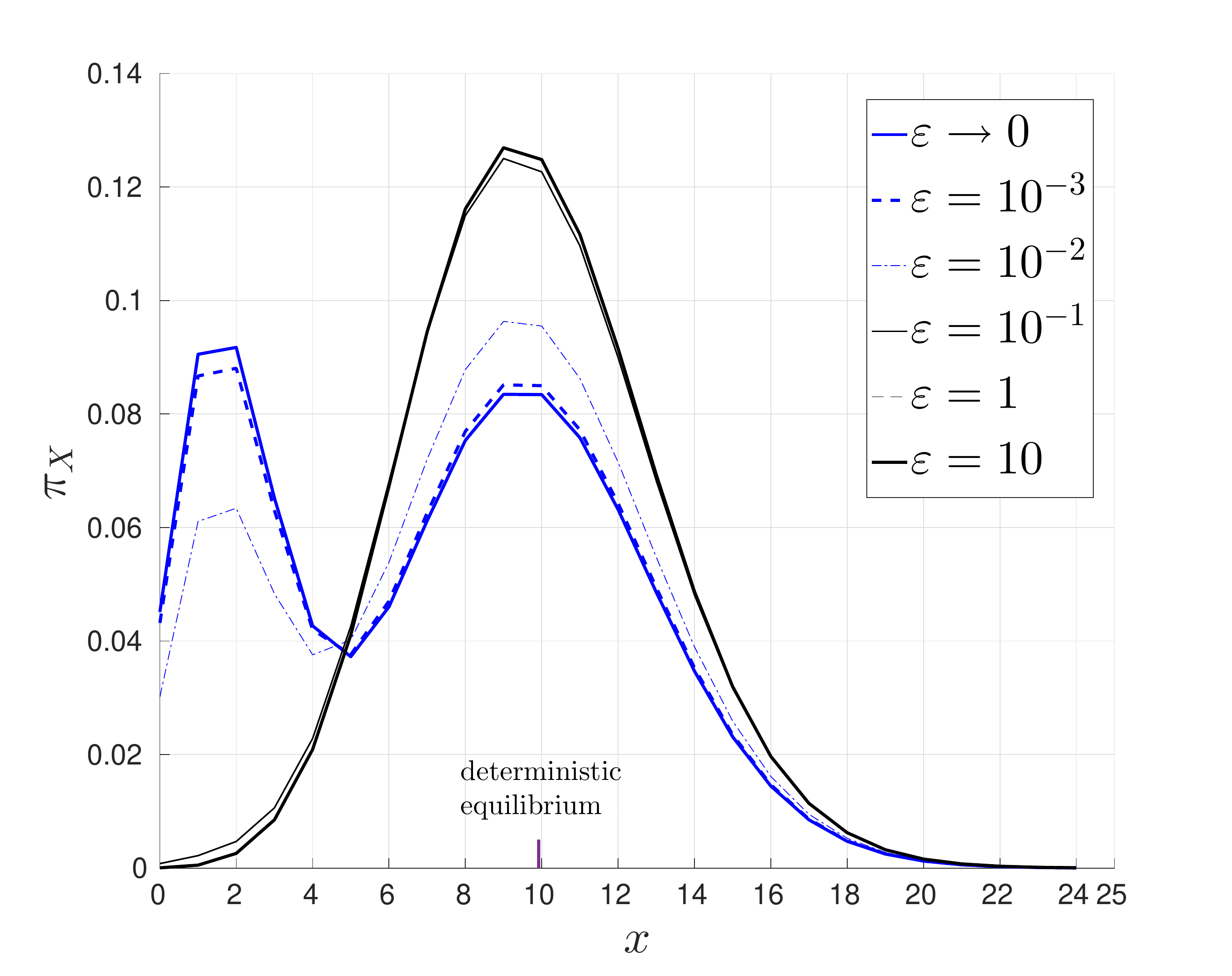}}
\caption{\textbf{More modes emerge due to slow promoter kinetics.} (a) A self-activating gene. (b) \red{The stationary probability distribution for different $\varepsilon$ which shows the} transition from fast promoter kinetics to slow promoter kinetics in a leaky cooperative self-activation of a gene  {given by \eqref{coop}}. \red{The slow kinetic limit is calculated via \eqref{e.mixture_m}, while the remaining curves are computed by a finite projection solution \cite{khammash06} of the CME. }  {The details and parameters are given in SI-\S 5.4 .}}

\label{f.selfactivate_plot}
\end{figure}

   \subsection*{The Toggle Switch}
A toggle switch is a basic GRN that exhibits deterministic multi-stability. It has two stable steady states and can switch between them with an external input or via noise. The basic design is a pair of two mutually repressing genes  {as in Figure \ref{f.toggle_switch}-a}. The ideal behavior is that only one gene is ``on'' at any moment in time. The network can be given by the following network with cooperativity indices $n,m$:
   \begin{align*}\begin{array}{rl}
\nonumber \mathrm Y_m+ \mathrm D_0^X &\xrightleftharpoons[\alpha_{-1}]{\alpha_1}  \mathrm D_1^X \\
  \mathrm D_0 &\mathop\rightarrow\limits^{k_{10}} \mathrm D_0^X+\mathrm X, \\
\mathrm X &   \mathop\rightarrow\limits^{k_{-1}} 0 \\
n\mathrm X  &\xrightleftharpoons[\beta_{-1}]{\beta_1}  \mathrm X_n \\ \end{array} \begin{array}{rl}
\nonumber \mathrm X_n+ \mathrm D_0^Y &\xrightleftharpoons[\alpha_{-2}]{\alpha_2} \mathrm D_1^Y  \\
\nonumber  \mathrm D_0^Y &\mathop\rightarrow\limits^{k_{20}} D_0^Y+\mathrm Y, \\ \mathrm Y & \mathop\rightarrow\limits^{k_{-2}} 0\\
mX  &\xrightleftharpoons[\beta_{-2}]{\beta_2}  \mathrm Y_m. \end{array}
\end{align*}
\rv{$\mathrm D^X$, $\mathrm D^Y$ denote the states of the promoters of the two genes expressing $X, Y$, respectively. } \\
 For the case $n,m=1$,  {there is no multi-merization reaction}.  {For consistency,} we choose $\beta_{ 1}=\beta_{-1}=\beta_2=\beta_{- 2}=1$ {in that case}.

  \red{Using the algorithm of} Proposition SI-2 we get that the distribution has three modes only (see SI-\S 3.3).  The stationary distribution for $X,Y$ is:
\begin{equation} \pi(x,y)= \frac 1{\frac{\alpha_1}{\alpha_{-1}} \rho_2 + \frac{\alpha_2}{\alpha_{-2}} \rho_1 + 1 }\left ( \mathbf P(y;{ \tfrac{k_{20}}{k_{-2}}}) \mathbf P(x;{ \tfrac{k_{10}}{k_{-1}}})  +  \frac{\alpha_1}{\alpha_{-1}} \rho_2  \mathbf P(y;{ \tfrac{k_{20}}{k_{-2}}})\delta(x) + \frac{\alpha_2}{\alpha_{-2}} \rho_1  \mathbf P(x;{ \tfrac{k_{10}}{k_{-1}}}) \delta(y) \right).  \end{equation}
where
\begin{equation}\label{e.rho_toggle} \rho_1 = \left ( \frac{k_{10}}{k_{-1}} \right ) ^n \frac{\beta_{1}}{n! \beta_{-1} }, \quad  \rho_2 = \left ( \frac{k_{20}}{k_{-2}} \right ) ^m \frac{\beta_{2}}{m! \beta_{-2} }. \end{equation}

 {Since the stationary distribution has three modes, it deviates from the ideal behavior of a switch where at most two stable steady states, under appropriate parameter conditions,  are possible. Nevertheless, } a bimodal distribution  can be achieved by minimizing the weight of the first mode at $(\frac{k_{10}}{k_{-1}},\frac{k_{20}}{k_{-2}})$. If we fix $\alpha_{1}/\alpha_{-1}, \alpha_{2}/\alpha_{-2}$, then this can be satisfied by tuning $n,m,\beta_{\pm 1},\beta_{\pm 2}$ to maximize $\rho_1,\rho_2$ in \eqref{e.rho_toggle}. Choosing higher cooperativity indices, subject to $n< k_{10}/k_{-1},m<k_{20}/k_{-2}$, achieves this. For instance, a standard design \cite{collins00} uses $n=2,m=3$. Figure \ref{f.toggle_switch} depicts the effect of cooperativity on achieving the desired behavior with the same dissociation constant and production ratios, and dimerization ratios equal to one. 
Notice that cooperativity allow us to minimize or maximize the weight of the mode corresponding to both proteins at high concentrations.

\begin{figure}
\begin{tabular}{cc}
\subfigure[Diagram of the toggle switch.]{\raisebox{0.25in}{\includegraphics[width=0.27\columnwidth]{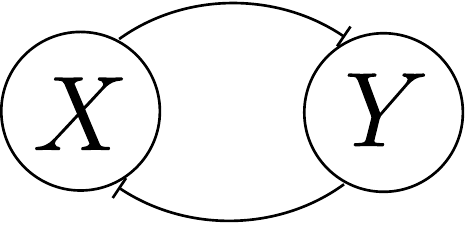}}}&
\subfigure[\red{Non-cooperative, } {$n,m=1$.}]{\includegraphics[width=0.5\columnwidth]{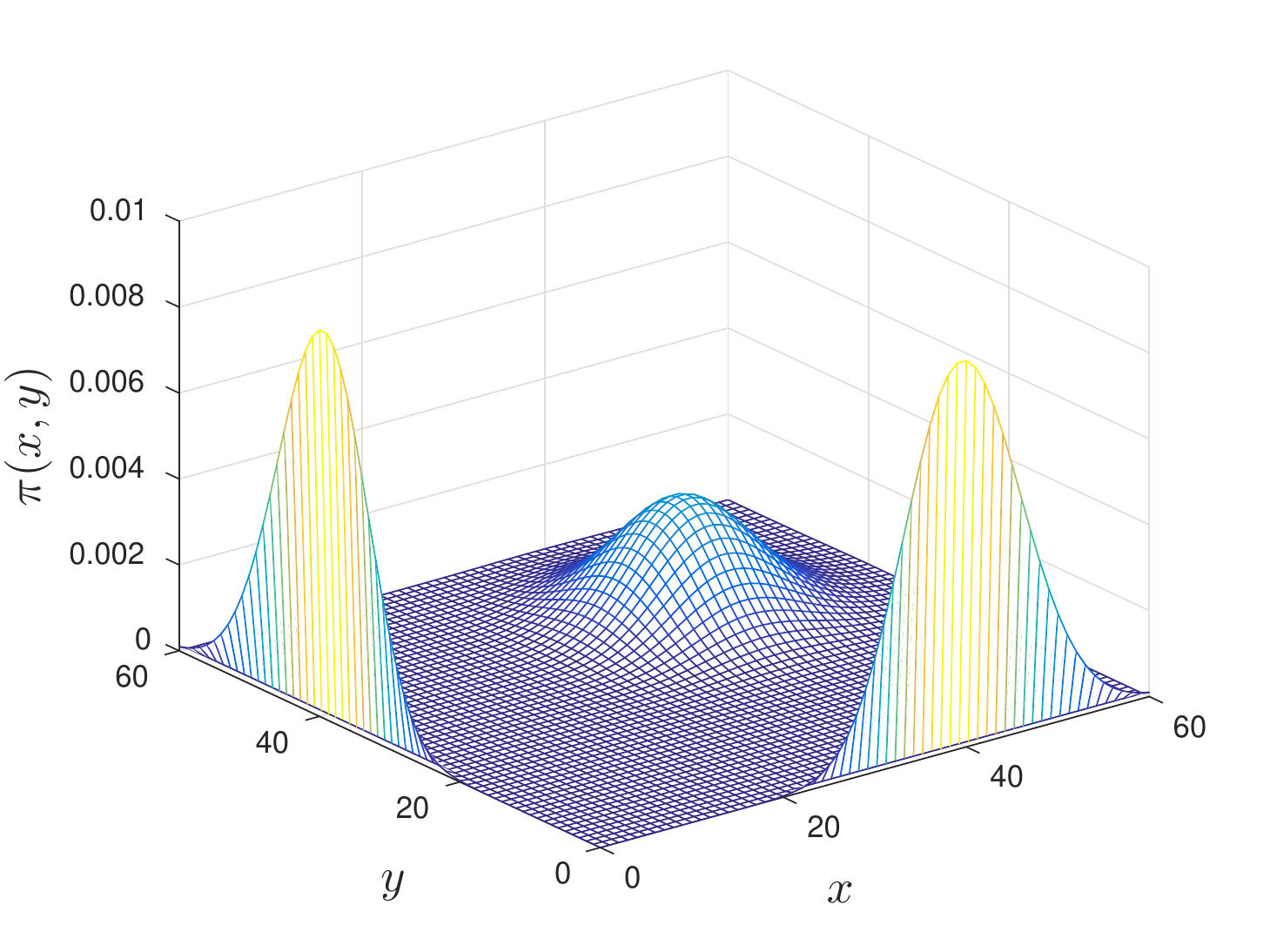}} \\
\subfigure[ {Cooperative, } \red{$\beta_1/\beta_{-1}=\beta_2/\beta_{-2}=1$.}]{\includegraphics[width=0.5\columnwidth]{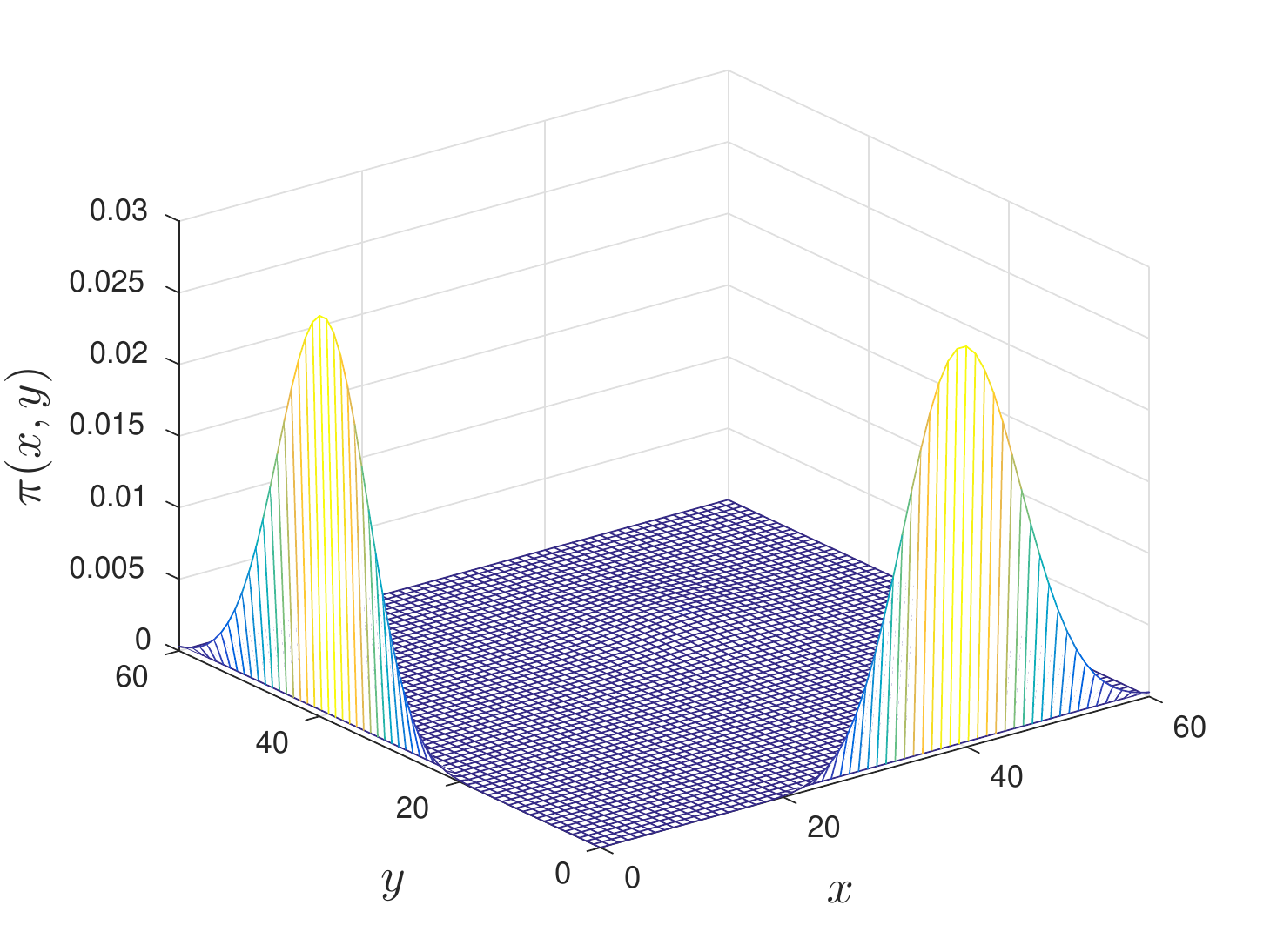}} &
\subfigure[ {Cooperative,} \red{$\beta_1/\beta_{-1}=\beta_2/\beta_{-2}=0.01$.}]{\includegraphics[width=0.5\columnwidth]{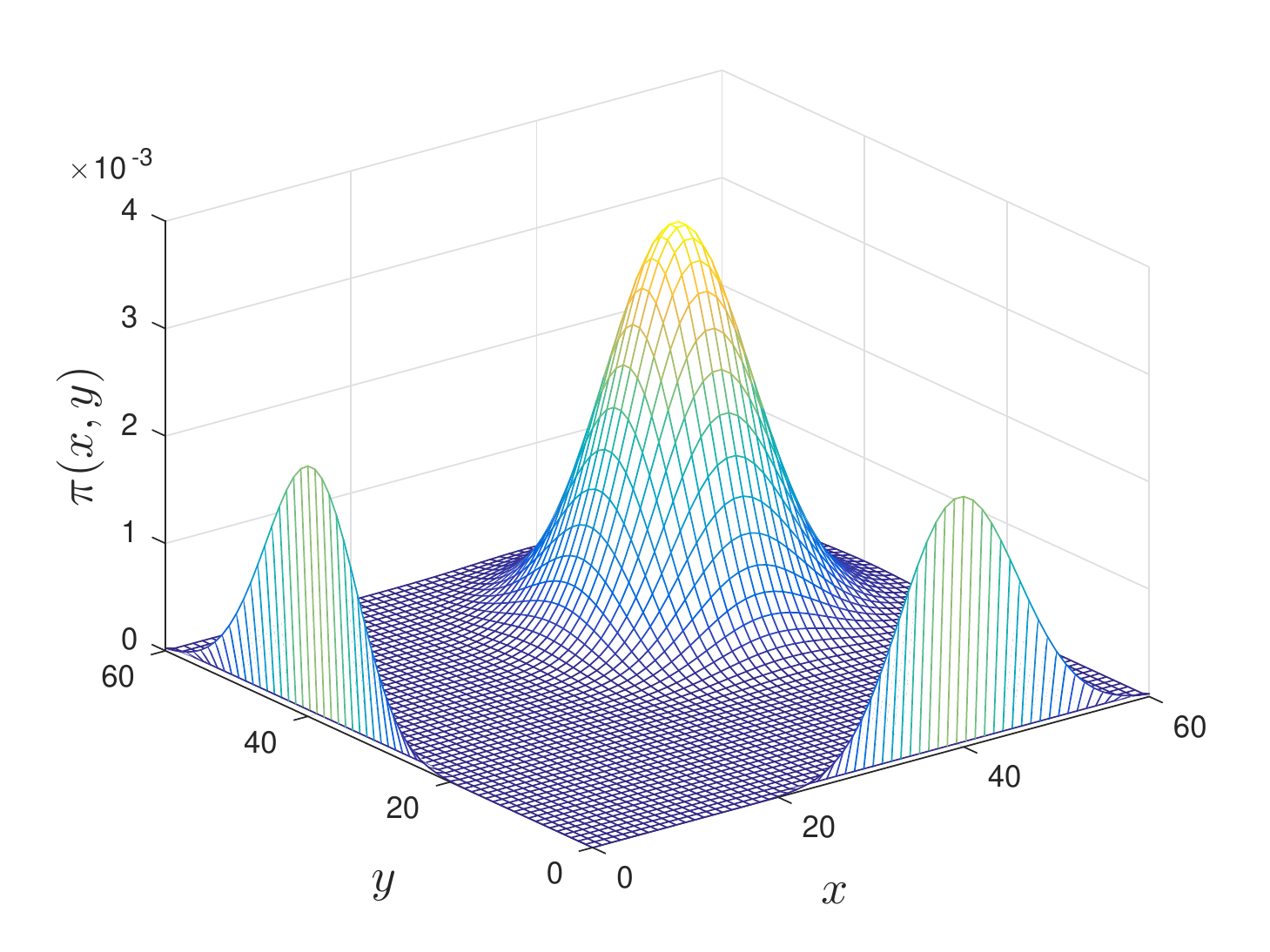}}\end{tabular}
\caption{\textbf{Cooperativity enables tuning of modes' weights.} Comparison of the stationary distribution between non-cooperative and cooperative binding. For all cases: $\alpha_{1}/\alpha_{-1}=\alpha_{2}/\alpha_{-2}=1/200, k_{10}/k_{-1}=k_{20}/k_{-2}=40$.  (a) Diagram of the toggle switch. (b) The stationary distribution corresponding to the non-cooperative case. (c) The stationary distribution corresponding to the cooperative case with $m=n=2, \beta_{1}/\beta_{-1}=\beta_2/\beta_{-2}=1$. (d) The stationary distribution corresponding to the cooperative case with $m=n=2, \beta_{1}/\beta_{-1}=\beta_2/\beta_{-2}=0.01$. All surfaces are plotted using \eqref{e.mixture_m}. }
\label{f.toggle_switch}
\end{figure}

The toggle switch has three modes regardless of the cooperativity index. This
is unlike the deterministic model where only one positive stable state is
realizable with non-cooperative binding,  {and two stable steady states are
  realizable with cooperative binding.} \rv{\S SI-3.3 contains further
  Monte-Carlo simulations that show that the predicted third mode appears with
  a 0.5 time scale separation. Experimentally, a recent pre-print has reported
  that the CRI-Cro toggle switch exhibits the third (high,high) mode and the authors proposed slow-promoter kinetics as a contributing mechanism \cite{jhu17}. }

\subsection*{Synchronization of interconnected toggle switches}
  We consider $N$ identical toggle switches:
\begin{equation}\begin{array}{rcl}
\nonumber \mathrm Y_{ic}+   \mathrm D_{0}^{xi}  &\xrightleftharpoons[\varepsilon \alpha_{-x} ]{\varepsilon \alpha_{x} }& \mathrm D_1^{xi} \\
  \mathrm D_{0}^{xi} &\mathop\rightarrow\limits^{k_x}& \mathrm X_i+ \mathrm D_{0}^{xi}, \\
\mathrm X_i &  \mathop\rightarrow\limits^{k_{-x}}& 0, \\
n\mathrm X_i &  \xrightleftharpoons[\beta_{-x} ]{\beta_x }& \mathrm X_{ic},
\end{array} \begin{array}{rcl}
\nonumber \mathrm X_{ic}+   \mathrm D_{0}^{yi}  &\xrightleftharpoons[\varepsilon \alpha_{-y} ]{\varepsilon \alpha_{y} }& \mathrm D_1^{yi} \\
  \mathrm D_{0}^{yi} &\mathop\rightarrow\limits^{k_y}& \mathrm Y_i+ \mathrm D_{0}^{yi}, \\
\mathrm Y_i &  \mathop\rightarrow\limits^{k_{-y}}& 0, \\
n\mathrm Y_i &  \xrightleftharpoons[\beta_{-y} ]{\beta_y }& \mathrm Y_{ic},
\end{array}\end{equation}
where $i=1,..,N$.
{We interconnect these systems through diffusion of
the protein species $\mathrm X_i,\mathrm Y_i$ among cells, modeled
through reversible
reactions with a diffusion coefficient $\Omega$:
 \begin{equation}\label{e.diffusion}
  \mathrm X_i \xrightleftharpoons[\Omega ]{\Omega } \mathrm X_j,
  \quad \mathrm Y_i \xrightleftharpoons[\Omega ]{\Omega } \mathrm Y_j, \quad i \ne j, \;i,j=1,..,N.
\end{equation}
We study this model as a very simplified version of a more complex quorum
sensing communication mechanism, in which orthogonal AHL molecules are
produced and by cells and act as activators of TFs in
receiving cells, as analyzed for example in \cite{sontag16}.}

Figure \ref{f.toggleNet}a depicts a block diagram
{of such a network.}

For a deterministic model, there exists a parameter range for which all toggle switches will synchronize into bistability for sufficiently high diffusion coefficient \cite{sontag16}.  {This implies each switch in the network behaves as a bistable switch, and it converges with all the other switches to the same steady-states. }

Our aim is to analyze the stochastic model at the limit of slow promot{e}r kinetics and compare it with the deterministic model.

This network is not in the form of the class of networks in Figure 1. Nevertheless, we show in SI-\S 4.1 that our results can be generalized to networks that admit weakly reversible deficiency zero conditional Markov chains.

 There are $4^N$ conditional Markov chains, and using Theorem 3 and Proposition 5, the stationary distribution is a mixture of $4^N-1$ Poissons.

 {Consider now the case of a high diffusion coefficient.} We show (see SI-\S 3.5) that as $\Omega\to \infty$, $X_1,..,X_N$ will synchronize in the sense that the joint distribution of $X_1,..,X_N$ is symmetric with respect to all permutations of the random variables. This implies that the marginal stationary distributions $p_{X_i}, i=1,..,N$ are identical. Hence, for sufficiently large $\Omega$,  the probability mass is concentrated around the region for which $X_1,..,X_N$ are close to each other.
Consequently, for large $\Omega$ we can replace the population of toggle switches with  a \emph{single toggle switch} with the \emph{synchronized protein processes} $X(t), Y(t)$,  {which are defined, for the sake of convenience, as $X(t):=X_1(t), Y(t):=Y_1(t)$.} Next, we describe the stationary distribution of $X(t),Y(t)$.

 The state of synchronized toggle switches does not depend on individual promoter configurations, and it depends only on the total number of unbound promoter sites in the network. Hence, the number of modes will drop from $4^N-1$ to $(N+1)^2-1$. 
 Note that similar to the single toggle switch, there are modes which have both $X,Y$ with non-zero copy number. On the other hand, there are many additional modes.  {Recall that in the case of a} single toggle switch, {we have tuned the cooperativity ratios such that the modes in which both genes are ON are suppressed}. Similarly, the undesired modes can be suppressed by tuning the cooperativity ratio  {which} can be achieved by choosing $\rho_{d_i}^X,\rho_{d_i}^Y, d=0,..,4^N-1$ sufficiently large. In particular, letting the multi-merization ratio $\beta_x/\beta_{-x},\beta_y/\beta_{-y} \to \infty$, the weights of modes in the interior of the positive orthant $\mathbb R_+^2$ approach zero.

 In conclusion, for sufficiently high $\Omega$ and sufficiently high multimerization ratio the population behaves as a \emph{multimodal switch },  {which means that the whole network can have either the gene $\mathrm X$ ON, or the gene $\mathrm Y$ ON. And every gene can take} $2N$ modes which are:
  \[ \left \{ \left ( \frac{i k_x}{Nk_{-x}}, 0   \right ),\left ( 0, \frac{i k_y}{Nk_{-y}}   \right ): i=1,..,N \right\}.\]
Comparing to the low diffusion case, the network will have up to $2^N-1$ modes with sufficiently high multimerization ratio.

In order to illustrate the previous results, consider a population of three toggle switches ($N=3$) and cooperativity  $n=2$.
 {For $\Omega$ greater then a certain threshold, the deterministic system bifurcates into bistabiliy. This means that all toggle switches converge to the same exact equilibria if $\Omega$ is greater than the threshold.  In contrast, the modes in the stochastic model of the toggle switches converge \emph{asymptotically} to each other. Hence,} we need to choose a {threshold} for  $\Omega$ that  { constitutes ``sufficient'' synchronization.}  {We choose to define this as} the protein processes synchronizing within one copy number. In other words, we require the maximum distance between the modes to be less than 1. It can be shown {(see SI-\S 3.5)} that the diffusion coefficient needs to satisfy:
\begin{equation}\label{q_inequality} \Omega\ge\frac 1N \max\{k_x -k_{-x}, k_y-k_{-y} \}. \end{equation}
The minimal  $\Omega$  that satisfies the inequality is $\Omega=75$ in this example.  The stationary distribution is depicted in Figure \ref{f.toggleNet}d.  The network has 15 modes, nine  {of} which are in the interior are suppressed due to cooperativity.
Comparing with the deterministic model,  it bifurcates into synchronization for $\Omega>0.5$. The stable equilibria of synchronized switch are $(149.98,0.022),(0.022,149.98)$.

The stochastic model with slow promoter kinetics adds four additional modes at $(0,100)$, $(100,0)$, $(50,0)$, $(0,50)$. This can be interpreted in the following manner. In the stochastic model,  the protein processes synchronize while the promoter configurations do not. The high states $(150,0),(0,150)$ correspond to the case when all the binding sites are empty. In the case when one binding site is empty, the first gene is producing while the second and the third are not. Due to diffusion,  {the first gene ``shares'' its expressed protein with the other two genes, which implies that each gene will receive a \emph{third} of the total protein copy numbers produced in the network.} A similar situation arises when two binding sites are empty.

\begin{figure*}
  \centering
  \subfigure[Network of toggle switches]{\raisebox{0.23 \height}{\includegraphics[width=0.2\textwidth]{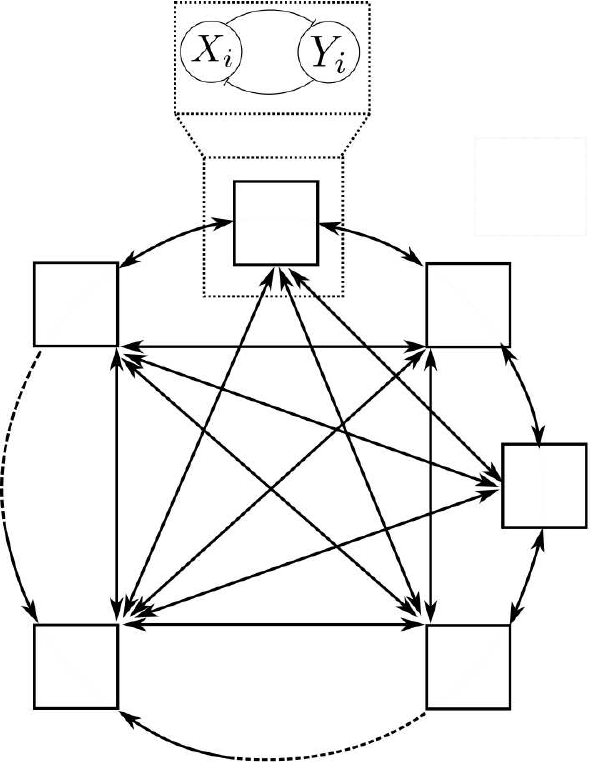}}}
   \raisebox{-0.08\height}{\includegraphics[width=0.7\textwidth]{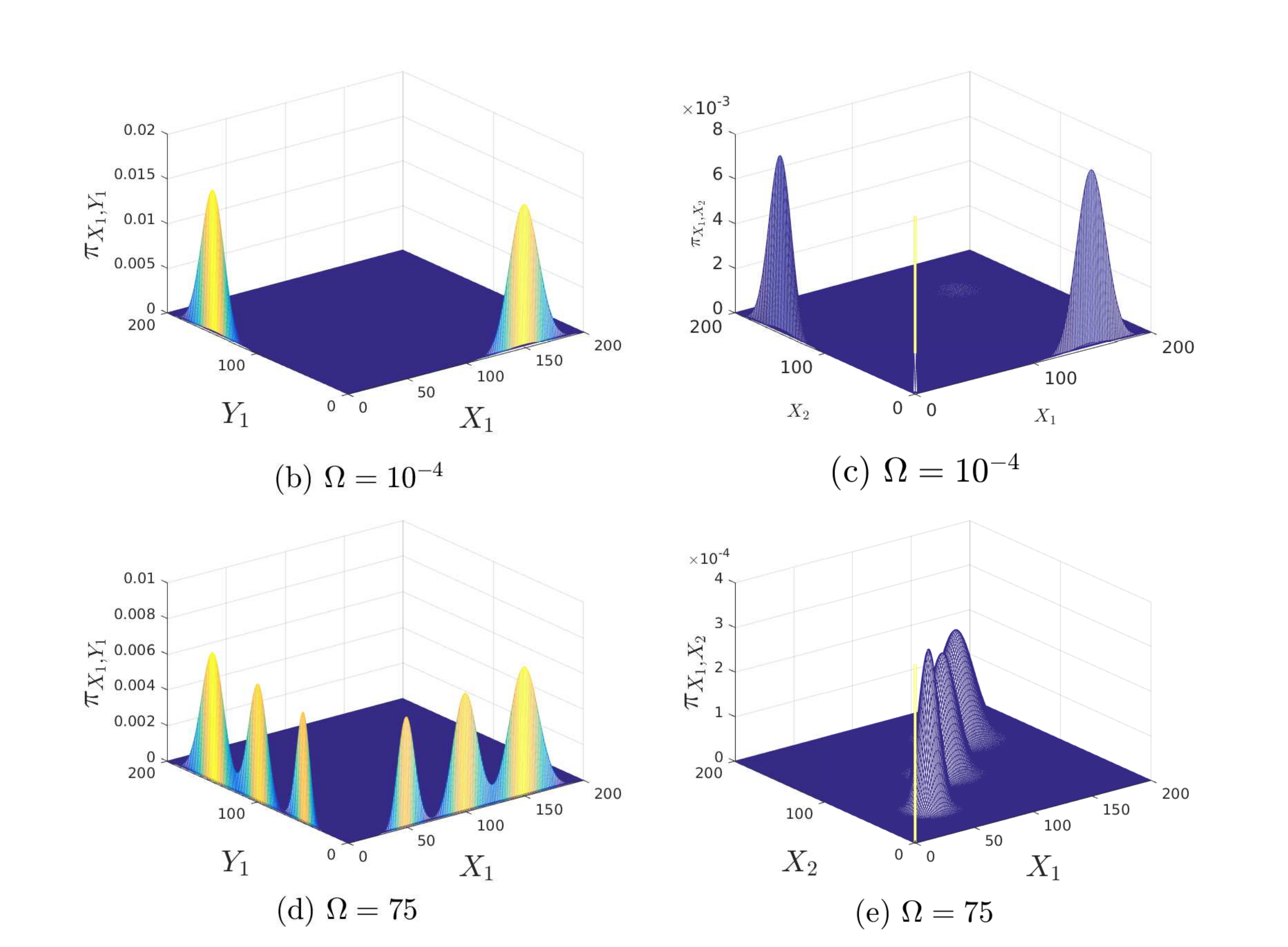}}
    \caption{\textbf{Slow promoter kinetics lead to the emergence of a multi- {modal} toggle switch} (a) A diagram of population of toggle switches.  {Arrows between blocks represent reversible diffusion reactions \eqref{e.diffusion}. Each block contains a toggle switch}. \red{The remaining subfigures show} stationary probability distributions for a population of three identical cooperative toggle switches. Due to the symmetries we plot joint distributions of $X_1,Y_1$ and $X_1,X_2$ only. Subplots (b),(c) depict the uncoupled toggle switches. It can be noted that the $X_1$ and $X_2$ are not synchronized. Subplots (d), (e) depict a high diffusion case. The toggle switches synchronize into a multi-modal toggle switch. \red{Parameters and details for calculating the plots are given in SI-\S 5.5  }}\label{f.toggleNet}
\end{figure*}

{\color{black}
 \subsection*{Trans-Differentiation Network}

We consider two networks for TF cross-antagonism in cell fate decision in this section. Both networks consist of two self-activating genes repressing each other as depicted in Figure \ref{f.pugata}-a \cite{zhou11}. The first network has independent cooperative binding of the TFs to the promoters. So it can be written as follows \cite{feng12}:
 \begin{equation}\label{fate1}
 \begin{array}{rl}
  \mathrm X_2+ \mathrm D_{00}^X \xrightleftharpoons[]{}  \mathrm D_{10}^X,\quad &    \mathrm X_2+ \mathrm D_{01}^X \xrightleftharpoons[]{}  \mathrm D_{11}^X,\\
 \mathrm Y_2+ \mathrm D_{00}^X \xrightleftharpoons[]{}  \mathrm  D_{01}^X, \quad &    \mathrm Y_2+  \mathrm  D_{10}^X  \xrightleftharpoons[]{}  \mathrm  D_{11}^X, \\
\mathrm X_2+ \mathrm D_{00}^Y \xrightleftharpoons[]{}  \mathrm  \mathrm D_{10}^Y, \quad&    \mathrm X_2+  \mathrm  D_{01}^Y  \xrightleftharpoons[]{}  \mathrm   D_{11}^Y, \\
\mathrm Y_2+ \mathrm D_{00}^Y  \xrightleftharpoons[]{}  \mathrm  D_{01}^Y,\quad&     \mathrm Y_2+  \mathrm  D_{10}^Y  \xrightleftharpoons[]{}   \mathrm D_{11}^Y, \\ 2\mathrm X  \xrightleftharpoons[\beta_{-x}]{\beta_x} \mathrm X_2, \quad& 2\mathrm Y  \xrightleftharpoons[\beta_{-y}]{\beta_y} \mathrm Y_2 \end{array}
\begin{array}{rl}
\mathrm D_{00}^X \mathop\rightarrow\limits^{k_{x0}} \mathrm  D_{00}^X + \mathrm X,\quad&  \mathrm D_{01}^X  \mathop\rightarrow\limits^{k_{x1}} \mathrm  D_{01}^X + \mathrm X,  \\ \mathrm D_{10}^X \mathop\rightarrow\limits^{k_{x2}}  \mathrm  D_{10}^X + \mathrm X, \quad &  \mathrm  D_{11}^X \mathop\rightarrow\limits^{k_{x3}} \mathrm  D_{11}^X \mathop\rightarrow\limits^{k_{x3}} \mathrm D_{11}^X + \mathrm X, \\
\mathrm D_{00}^Y \mathop\rightarrow\limits^{k_{y0}}   \mathrm  D_{00}^Y + \mathrm Y,\quad&  \mathrm  D_{01}^Y  \mathop\rightarrow\limits^{k_{y1}} \mathrm  D_{01}^Y + \mathrm Y, \\ \mathrm D_{10}^Y \mathop\rightarrow\limits^{k_{y2}}  \mathrm D_{10}^Y + \mathrm Y, \quad &  \mathrm  D_{11}^Y \mathop\rightarrow\limits^{k_{y3}}  \mathrm  D_{11}^Y \mathop\rightarrow\limits^{k_{y3}} \mathrm D_{11}^Y + \mathrm Y, \\
\mathrm X \mathop\rightarrow\limits^{k_{-x}} \emptyset, \quad& \mathrm Y \mathop\rightarrow\limits^{k_{-y}} \emptyset. \\ \strut
\end{array}
\end{equation}
 In order for the genes to be cross-inhibiting and self-activating we let: $k_{x1}=k_{y2}=0$. Also, $k_{x2}>k_{x0},k_{x3}$ and $k_{y1}>k_{y0},k_{y3}$.

 \begin{figure}
  \centering
    \subfigure[ {A cell fate  network}]{\raisebox{0.1\height}{\includegraphics[width=0.08\columnwidth]{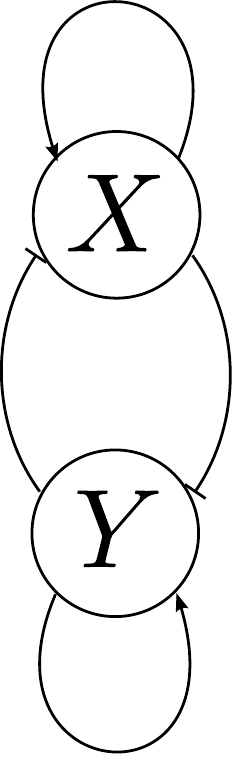}}}
  \subfigure[The stationary distribution of cell-fate circuit \eqref{fate1} ]{\includegraphics[width=0.43\columnwidth]{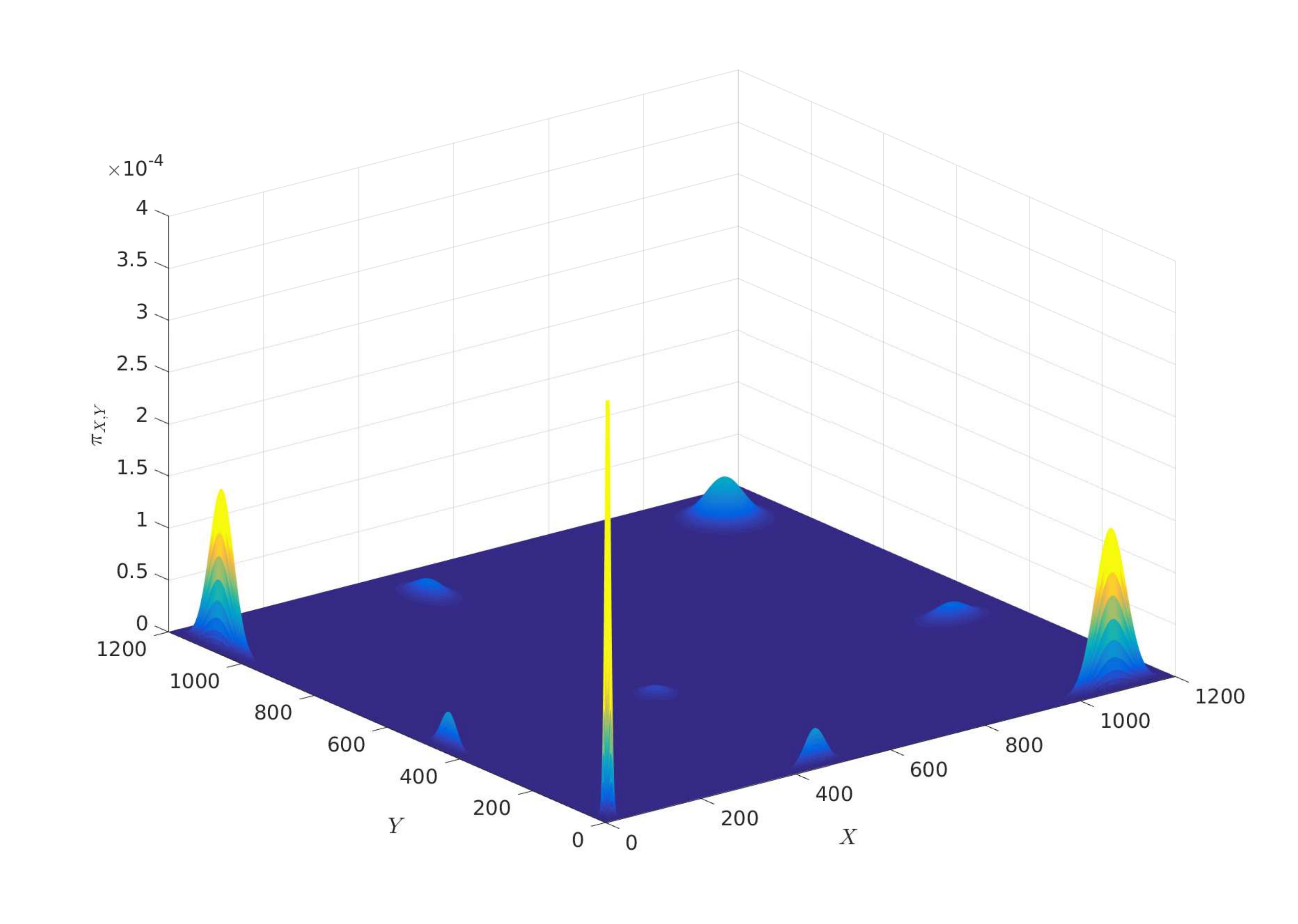}}
  \subfigure[The stationary distribution of a PU.1/GATA.1 network \eqref{fate2}]{\includegraphics[width=0.43\columnwidth]{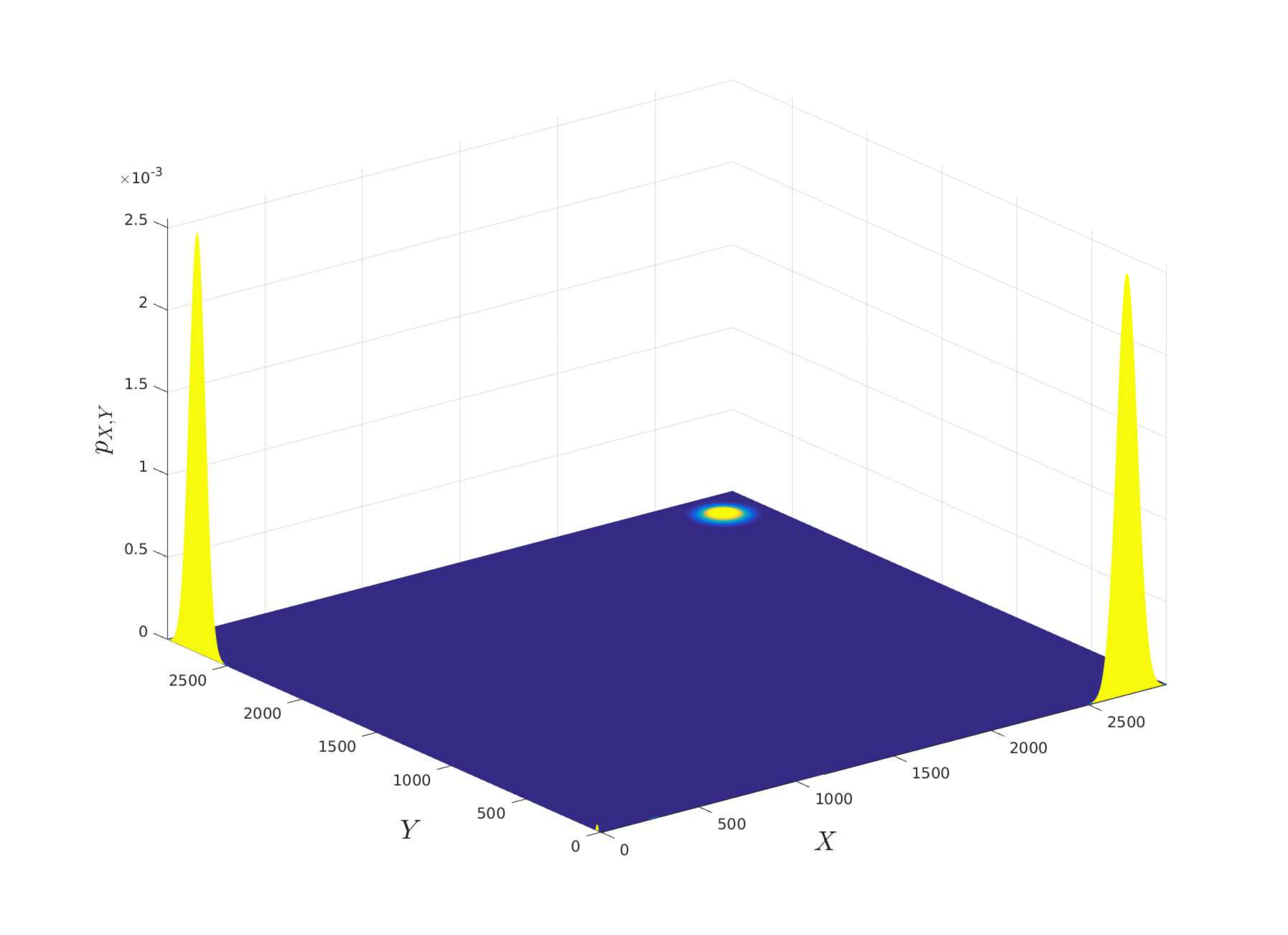}}
  \caption{\textbf{The cell-fate decision network with slow promoter kinetics has more modes than what a deterministic model predicts.} (a) A diagram of a generic cell-fate circuit that can describe both networks \eqref{fate1},  \eqref{fate2}, (b) The joint probability of an archetypical cell-fate circuit  \eqref{fate1} \red{computed using Theorem \ref{th}}. (c) The stationary distribution of a PU.1/GATA.1 circuit  \eqref{fate2}, where $\mathrm X$ denotes PU.1 and $\mathrm Y$ denotes GATA.1. Three modes can be seen. The parameters are given in SI-\S 5.6.} \label{f.pugata}
\end{figure}
The network can be analyzed with the proposed framework, as it consists of two genes each with two binding sites. Hence it can  {theoretically} admit up to 16 modes according to \eqref{e.mixture_m}.
The stationary distribution is depicted in Figure \ref{f.pugata}-b for an example parameter set. Note that despite the fact that we have 16 modes, only eight of them contribute to most of the stationary distribution.  This is to be contrasted with a deterministic model, which cannot produce more than 4 stable equilibria \cite{elife}.

The second network that  we study is a model of the PU.1/GATA.1 network, which is a lineage determinant in hematopoietic stem cells
 \cite{dore11}. Diagrammatically, it can also be presented by Figure \ref{f.pugata}-a. However, it differs from the first network presented above in several ways. First, PU.1 needs GATA.1 to bind to the promoter of GATA.1 \cite{zhang00}, and vice versa \cite{burda16}. In our modelling framework this means that the promoter configurations $\mathrm D_{01}^X,D_{10}^Y$ do not exist, where $\mathrm X$ stands for PU.1 and $Y$ stands for GATA.1. Hence, the network has nine gene states. Second, there is no evidence that PU.1 and GATA.1 form dimers to activate their own promoters cooperatively. In fact, it has been shown that self-activation for GATA-1 occurs primarily through monomeric binding \cite{crossley95}. Therefore, the PU.1/GATA.1 network can be written as follows:
\begin{equation} \label{fate2}
 \begin{array}{rl}
  \mathrm X_2+ \mathrm D_{00}^X \xrightleftharpoons[\alpha_{-x0}]{\alpha_{x0}}  \mathrm D_{10}^X,\quad & \mathrm Y_2+  \mathrm  D_{10}^X  \xrightleftharpoons[\alpha_{-x1}]{\alpha_{x1}}  \mathrm D_{11}^X, \\
\mathrm Y_2+ \mathrm D_{00}^Y  \xrightleftharpoons[\alpha_{-y0}]{\alpha_{y0}}   \mathrm D_{01}^Y,\quad&     \mathrm X_2+  \mathrm  D_{01}^Y  \xrightleftharpoons[\alpha_{-y1}]{\alpha_{y1}}  \mathrm D_{11}^Y,\\
\mathrm X \mathop\rightarrow\limits^{k_{-x}} \emptyset, \quad& \mathrm Y \mathop\rightarrow\limits^{k_{-y}} \emptyset, \\  \end{array}
\begin{array}{rl}
\mathrm D_{00}^X \mathop\rightarrow\limits^{k_{x0}} \mathrm  D_{00}^X + \mathrm X, & \mathrm D_{10}^X \mathop\rightarrow\limits^{k_{x1}}  \mathrm D_{10}^X + \mathrm X, \\ \mathrm D_{00}^Y \mathop\rightarrow\limits^{k_{y0}} \mathrm D_{00}^Y + \mathrm Y &
    \mathrm D_{01}^Y  \mathop\rightarrow\limits^{k_{y1}} \mathrm  D_{01}^Y + \mathrm Y, \\ \strut
\end{array}
\end{equation}
 More detailed discussion of the model is included in \cite{acc18}.

With lack of cooperativity, a deterministic model is only monostable and cannot explain the emergence of bistability for the above network \cite{acc18}. However, using our framework, up to nine modes can be realized. In order to simplify the landscape, we group the nine into four modes. This is possible since the states $\mathrm D_{11}^X, D_{00}^X$, $D_{11}^Y, D_{00}^Y$ have very low production rates. This gives a total of \emph{four} modes which are (low,low),(high,low),(low,high),(high,high). Using our model, we choose the parameters to realize bistability and tristability. Figure \ref{f.pugata}-c depicts the stationary distribution for a set of parameters that satisfies the assumptions and give rise  to a tristable distribution.
}

\subsection*{The Repressilator} \rv{A very different example is provided by a well-studied synthetic oscillator, the repressilator \cite{elowitz00}.   The slow gene activation model predicts oscillations even in the noncooperative binding regime in which the deterministic model as well as the adiabatic stochastic model do not oscillate.  See \S SI-3.4.  }
\section*{Methods}
\red{\subsection*{Numerical Simulation Software}
All calculations were performed using MATLAB R2016a, except for the
computation of deterministic solutions of the
{``quorum sensing''} numerical example where we used Bertini 1.5, which is software for solving polynomials numerically via homotopy methods.}

\section*{Discussion}

 Phenotypical variability in the absence of genetic variation
is a phenomenon of great interest in current biological and
translational research, as it plays an
important role in processes as diverse as
embryonic development \cite{CalvaneseFraga2012},
hematopoietic cell differentiation \cite{SharmaGurudutta2016},
and cancer heterogeneity \cite{EaswaranTsaiBaylin2014}.
A conceptual, and often proposed, unifying framework to explain non-genetic
variability is to think of distinct phenotypes as multiple ``metastable
states'' or ``modes'' in the complex energetic landscape associated to an
underlying GRN.
Following this point of view, we studied in this paper a general but
simplified mathematical model of gene regulation.  Our focus was on
stochastic slow promoter kinetics, the time scale relevant when
transcription factor binding and unbinding are affected by epigenetic
processes such as DNA methylation and chromatin remodeling.
In that regime, adiabatic approximations of promoter kinetics
are not appropriate.
In contrast to the existing literature, which largely confines itself
to numerical simulations, in this work we provided a rigorous analytic
characterization of multiple modes.

The general formal approach that we developed provides
insight into the
relative influence of model parameters on system behavior.  It also allows
making theoretical predictions of how changes in wiring of a gene regulatory
network, be it through natural mutations or through artificial interventions,
impact the possible number, location, and likelihood, of alternative states.
We were able to tease out the role of cooperative binding in stochastic models
in comparison to deterministic models, which is a question of great interest
in both the analysis of natural systems and in synthetic biology engineering.
Specifically, we found that, unlike deterministic systems, the number of modes
is independent of whether the TF-promoter binding is cooperative or not; on
the other hand, cooperative binding gives extra degrees of freedom for
assigning weights to the different modes.
More generally, we characterized the stationary distributions of CMEs for our GRNs as mixtures of Poisson
distributions, which enabled us to obtain explicit formulas for the locations
and probabilities of metastable states as a function of the parameters
describing the system.
One application of our mathematical results was to models of single or
communicating ``toggle switches'' in bacteria, where we showed that, for
suitable parameters, there are a very large number of metastable attractors. \rv{Indeed, stochastic effects have been shown before to lead to  multi-stability in a population of enzymatic reactions \cite{warmflash08}.}

This work was in fact
motivated
by our interest in
hematopoietic cell differentiation,
and in this paper we discussed two possible models of trans-differentiation
networks in mammalian cells.  In a first model, based on previous
publications, we uncovered more modes than had been predicted with different
analyses of the same model.  This implies that in practice there could be
unknown ``intermediate'' phenotypes that result from the network's dynamics,
which may be acquired by cells during the natural differentiation process or
which one might be able to induce through artificial stimulation.  The second
model included only binding reactions that have been experimentally
documented, and as such might be more biologically realistic than the first
model.  For this second model, a deterministic analysis predicts
monostability, which is inconsistent with the fact that the network should
control a switch between two stable phenotypes (erythroid and myeloid). This
suggests that stochasticity, likely due to low copy numbers and/or slow
promoter kinetics, might be responsible for the multiple attractors
(phenotypes) that are possible in cell differentiation gene regulatory
networks.

Our mathematical results, being quite generic, should also be useful in the
analysis of networks that have been proposed for understanding aspects of
cancer biology.
For example,
non-genetic heterogeneity has been recently recognized as an important
factor in cancer development and resistance to therapy, with stochastic
multistability in gene expression dynamics acting as a generator of phenotype
heterogeneity, setting a balance between mesenchymal, epithelial, and
cancer stem-cell-like states \cite{DeanFojoBates2005}
\cite{Huang_et_al_Ingber2005} \cite{Gupta_et_al_Lander2011} \cite{Huang2011},
and nongenetic variability due to multistability arising from mutually
repressing gene networks has been proposed to explain metastatic
progression \cite{Lee_at_al_Balazsi_Rosner}.

\section*{Acknowledgements} We thank Nithin S. Kumar for discussions regarding the PU.1/GATA.1 network and Cameron McBride for proofreading the manuscript. This work was supported by an AFOSR grant FA9550-14-1-0060.

\newpage

\begin{center} \Huge \bfseries Supporting Information
\end{center}

\vspace{1in}

\section{Review of the Master Equation and Markov Chains}%

\subsection{The Chemical Master Equation}
\rv{
 A generic reaction $\mathrm R_j \in \mathscr R$ takes the form: 
\begin{equation}\label{e.reaction}
 \mathrm R_j: \ \sum_{i=1}^{|\mathscr S|} \alpha_{ij} \mathrm Z_i \rightarrow  \sum_{i=1}^{|\mathscr S|} \beta_{ij} \mathrm Z_i,
\end{equation}
where $Z_i \in \mathscr S$, \red{$\alpha_{ij},$ and $\beta_{ij}$ are positive integers}.  The reactions that we consider are limited to at most two reactants. The reverse reaction of $\mathrm R_j$ is the reaction in which the products and reactants are interchanged. If the network contains both the reaction and its reverse then we use the short-hand notation to denote both of them as
\begin{equation}\label{e.reaction_r}
 \mathrm R_j: \ \sum_{i=1}^{|\mathscr S|} \mathrm \alpha_{ij} \mathrm Z_i  \rightleftharpoons  \sum_{i=1}^{|\mathscr S|} \beta_{ij} \mathrm Z_i.
\end{equation}}

\rv{The stoichiometry of a CRN can be summarized by a \emph{stoichiometry matrix} $\Gamma$ which is defined element-wise as follows:
\[  [\Gamma]_{ij} = \beta_{ij}-\alpha_{ij}.\]
The columns of the stoichiometry matrix $\gamma_1,..,\gamma_{|\mathscr R|}$ are known as the stoichiometry vectors.  {We say that} a nonzero nonnegative vector  $d$ {gives} a \emph{conservation law} for the stoichiometry if $d^T \Gamma =0$.}

The dynamics of the network refers to the manner in which the \emph{state} evolves in time, where {the} state $Z(t) \in \mathbf Z \subset \mathbb Z_{\ge 0}^{|\mathscr S|}$ is the vector of  copy numbers of the species of the network at time $t$. Since the collision of molecules is random in nature, the time-evolution of states is described mathematically by a stochastic process.  The standard stochastic model for a CRN is that of a continuous Markov chain. Let $\mathbf Z$ denote the state space. Consider a time $t$ and let the state be $Z(t)=z \in \mathbf Z$. Then, the probability that the $j^{\rm th}$ reaction fires in an interval $[t,t+\delta]$ is $R_j(z) \delta+o(\delta)$. If $\mathrm R_j$ fires, then the states changes from $z$ to $z+\gamma_j$, where $\gamma_j$ is the corresponding stoichiometric vector.

As $Z$ is a stochastic process we are interested in characterizing its qualitative behavior given by the joint probability distribution $p_z(t)=\Pr[Z(t)=z|Z(0)=z_0]$ for any given initial condition $z_0$.   The time-evolution of the probability distribution can be shown \cite{anderson15} to be given by a system of linear ordinary differential equations  known as the \emph{forward Kolmogorov equation} or the \emph{Chemical Master Equation}, given by:
\begin{equation}\label{e.MasterEquation}
\dot p_z(t)= \sum_{j=1}^{|\mathscr R|} R_j(z-\gamma_j) p_{z-\gamma_j}(t) - R_j(z)p_z(t),  {z \in \mathbf Z},
\end{equation}
where $\gamma_1,..,\gamma_{|\mathscr R|}$ are the columns of the stoichiometry matrix.

Since our species are either gene species or protein species, we split the stochastic process $Z(t)$ into two subprocesses: \emph{the gene process} $D(t)$ and \emph{the protein process} $X(t)$, as explained below.

  Consider the $i^{\rm th}$ gene. For each configuration species $\mathrm D_j^i \in B_i$, let  $D_j^i(t) \in \{0,1\}$ denote its occupancy, i.e. if $D_j^i(t)=1$, then at time $t$ the $i^{\rm th}$ gene is in a configuration $j \in B_i$.  It can be seen from gene reactions {(1)-(6)} that the network always has  {a} conservation law supported on $\{\mathrm D_j^i, j \in B_i\}$,  {so that:}
 \[ \sum_{j \in B_i} D_j^i(t) = 1, \]
 which reflects the physical constraint that the {promoter} can be in  {only} one
 configuration at  {any} given time.

 This conservation law enables us to introduce an equivalent reduced
 representation. For each gene we define one process $D_i$ such that $D_i(t)
 \in B_i$.  $D_i(t)=j$ if and only if $D_j^i(t)=1$.
Collecting these into a vector, define the gene process $D(t):=[D_1(t),...,D_N(t)]^T$ where $D(t) \in \prod_{i=1}^N B_i $. The $i^{\rm th}$ gene can be represented by $|B_i|$ states, so $L:={\prod_{i=1}^N |B_i|}$ is  the total number of promoter configurations in the GRN. {With} abuse of notation, we write also $D(t) \in \{0,..,L-1\}$ in the sense of the bijection between $\{0,..,L-1\}$ and $\prod_{i=1}^N B_i $ defined by interpreting $D_1...D_N$ as a binary representation of an integer. Hence, $d \in \{0,..,L-1\}$ corresponds to $(d_1,...,d_N) \in B_1 \times .. \times B_N$ and we write $d=(d_1,..,d_N)$.

 Since each gene expresses a corresponding protein,  we define $X_{i1}(t)\in \mathbb Z_{\ge 0}, i =1,..,N$ protein processes. If the multimerized version of the $i^{\rm th}$ protein participates in the network as an activator or repressor then we define $X_{ic}(t)$ as the corresponding multimerized protein process, and we denote $X_i(t):=[X_{i1}(t),X_{ic}(t)]^T$. If there is no multimerization reaction then we define $X_i(t):=X_{i1}(t)$. Since not all proteins are necessarily multimerized, the total number of protein processes is  $N \le M \le 2N$. Hence,
   the \emph{protein  process} is $X(t)=[X_1^T(t),..,X_N^T(t)]^T \in {\mathbb Z}_{\ge 0}^M $ and the state space can be written as $\mathbf Z= {\mathbb Z}_{\ge 0}^M \times \prod_{i=1}^N B_i$.

Consider the joint probability distribution:
\begin{equation}\label{e.jointpdfSI}
  p_{d,x}(t)= \Pr[X(t)=x, D(t)=d],
\end{equation}
\red{which represents the probability at time $t$ that the protein process $X$ takes the value} $x \in {\mathbb Z}_+^M$ \red{and the gene process $D$ takes the value} $d\in \{0,..,L-1\}$.  {Recall} \red{that $x$ is a vector of copy numbers for the protein processes while $d$ encodes the configuration of each promoter in the network.} Then, we can define for each fixed $d$:
\begin{equation}\label{p_dSI} p_d (t):= [p_{dx_0}(t),p_{dx_1}(t),....]^T, \end{equation}
\red{representing the vector enumerating  the probabilities \eqref{e.jointpdfSI} for all values of $x$ and for a fixed $d$, }where $x_0,x_1,..$ is an indexing of ${\mathbb Z}_{\ge 0}^M$. Note that $p_d(t)$ can be thought of as an infinite vector with respect to the aforementioned indexing. Finally, let \begin{equation}\label{e.vectordecompSI} p(t):=[p_0(t)^T,...,p_{L-1}^T(t)]^T \end{equation}
\red{representing a concatenation of the vectors \eqref{p_dSI} for $d=0,..,L-1$.} Note that $p(t)$ is a finite concatenation of infinite vectors.

\red{The joint stationary distribution $\bar\pi$ is defined as the following limit,  {which we assume to exist and be independent of the initial distribution:}
\begin{equation}\label{pi_defSI}
\bar\pi = \lim_{t \to \infty } p(t).
\end{equation}}
The stationary distribution $\bar \pi$ is a function of $\varepsilon$ also.

Consider a given GRN.  The master equation \eqref{e.MasterEquation} \red{is defined over a countable state space $\mathbf Z$ which can be enumerated with an arbitrarily chosen order. Hence, the  master equation can be interpreted as an infinite system of} differential equations. Its infinite infinitesimal generator matrix $\Lambda$   can be written succinctly entry-wise as:
\begin{equation}\label{e.Q} \lambda_{ z \tilde z}:=\left \{ \begin{array}{ll} R_j(z) & \mbox{if} \, \exists j \ \mbox{such that} \  \tilde z=z-\gamma_j  \\ -\sum_{\tilde z \ne z} \lambda_{z \tilde z} = -\sum_{j=1}^{|\mathscr R|} R_j(z) & \mbox{if} \ \tilde z= z \\ 0 & \mbox{otherwise}  \end{array} \right . , \end{equation}
\red{where $\lambda_{z\tilde z}$ refers to the rate of transition from $z$ to $\tilde z$.}
The matrix $\Lambda$ is stochastic, which means that it is Metzler and $1^T \Lambda=0$. A Metzler matrix is a matrix whose off-diagonal elements are non-negative.

\subsection{Irreducibility}
An important property in the context of  Markov chain analysis is that of \emph{irreducibility} \cite{norris}, and its significance stems from the fact that it is a necessary condition for the existence of a unique positive stationary distribution. Consider the Markov chain $Z(t)$ defined on $\mathbf Z$ with an associated infinitesimal generator $\Lambda$ as given in \eqref{e.Q}. Let $z,w \in \mathbf Z$. Then, it is said that $z$ leads to $w$ if there exist states $z_0,...,z_n \in \mathbf Z$ such that $\lambda_{zz_0} \lambda_{z_0z_1} ... \lambda_{z_nw}  >0$. A set $U \subset \mathbf Z$ is said to be  a \emph{communicating class} if for every $z_1,z_2 \in U$, $z_1$ leads to $z_2$ and $z_2$ leads to $z_1$. The state space $\mathbf Z$ can always be partitioned into a disjoint union of communicating classes \cite{norris}. The Markov chain is said to be \emph{irreducible} if the state space is a \emph{communicating class}. A communicating class $U$ is said to be \emph{closed} if $z \in U$, and $z$ leads to $w$ implies $w \in U$. A Markov chain is said to be \emph{weakly irreducible} if it has a unique closed communicating class $U$, and for all $z \in \mathbf Z$, $z$ leads to some element $U$.

We state the following result, under assumption A4:

\begin{proposition}\label{th.irreducible}Consider a gene regulatory network that consists of $N$ gene expression blocks. Then the associated Markov chain is {weakly irreducible}.
\end{proposition}
\begin{proof}
Consider the state $0 \in \mathbf Z$. We first show that for all $z \in \mathbf Z$, $z$ leads to $0$. Let $z=(x_1,d_1,..,x_N,d_N)$. We list the set of reactions, i.e transitions, that will lead to 0. Consider $d_i\ne 0$, if $d_i=1$ then we apply either the reaction (1) or (2). If $d_i=10$, then we apply reaction (3) and if $d_i=01$ we apply reaction (5). If $d_i=11$, then we apply reactions (3) and (4). Hence, $z$ leads to a state of the form $(x_1,0,x_2,0,...,x_N,0)$. Similarly, we can apply the decay reactions (8) and the reverse dimerization until we reach the origin.

Now we show that there exists a closed communicating class. If $0$ does not lead to any state then $\{0\}$ is a closed communicating class.  Otherwise, let $U$ be the smallest communicating class containing $0$. Note that $U$ is closed, since if there exists $z \in U$ that leads  to $w$, then $w$ leads to $0 \in U$.

In order to show that $U$ is unique, assume that there exists another closed communicating class $U'$. But this contradicts with the fact that all $z \in U'$ lead to $0 \in U$. We have shown that for all $z \in \mathbf Z$, $z$ leads to 0. Hence $z$ leads to $U$. \end{proof}

\begin{remark}
For finite Markov chains, weak irreducibility with appropriate stochastic stability assumptions are sufficient for the existence of a \emph{nonnegative} unique stationary distribution \cite{yin97}, while irreducibility is usually needed for the existence a \emph{positive} stationary distribution. Note that not all GRNs are irreducible. However, our subsequent results require weak irreducibility only, and investigation of irreducibility is out of the scope of this paper. Nevertheless, necessary and sufficient graphical conditions for irreducibility can be developed and are subject to future work.
\end{remark}

\section{Proofs of the Main Results}
In this section we include mathematical proofs of the main results in the main text.

\subsection{Decomposition of the Master Equation}
We include a proof for Proposition 1.

By the time-scale separation assumption, the gene reactions are slow and the protein reactions are fast. Then \eqref{e.Q} can be written as:
  \begin{equation}\label{e.Qdecomp} \lambda_{ z \tilde z}:=\left \{ \begin{array}{ll}  R_j^{(f)}(z)+ \varepsilon R_j^{(s)}(z) & \mbox{if} \, \exists j \ \mbox{such that} \  \tilde z=z-\gamma_j  \\ -\sum_{\tilde z \ne z} \lambda_{z \tilde z} = -\sum_{j=1}^{|\mathscr R|}   R_j^{(f)}(z)+\varepsilon R_j^{(s)}(z) & \mbox{if} \ \tilde z= z \\ 0 & \mbox{otherwise}  \end{array} \right ., \end{equation}
  where $(f),(s)$ denote fast and slow, respectively.

Hence, the summation in \eqref{e.MasterEquation} can decomposed into two terms. This implies  that the system matrix can be written as a sum of a fast matrix $\tilde \Lambda$ and a slow matrix $\varepsilon\hat \Lambda$ as in Eq. (14).

We now show that Eq. (15) holds, \red{which amounts to showing that $\tilde\Lambda$ is block diagonal}. Assume $\exists j$ such that $\tilde z=z-\gamma_j$. Let $z=(x,d),\tilde z=(\tilde x,\tilde d)$ with $\tilde d\ne d$.  As can be seen in Figure  2, protein reactions do not change the promoter configuration state $d$. \red{Hence, the transition rate $\lambda_{z\tilde z}$ has terms corresponding to gene reactions rate only, i.e.,} $\lambda_{z\tilde z}=\varepsilon R_j^{(s)}(z)$. \red{Hence, $\tilde \Lambda$ is block diagonal.}
 \hfill $\blacksquare$

 \subsection{Analytic Expression of the Conditional Probability Distributions}
 We include a proof of Proposition 2.

 As mentioned before, the stationary distribution is the product of the marginal stationary distributions, since the underlying conditional stochastic processes are independent.
If $\mathrm X_i$ does not form a multimer then it is known that the stationary distribution of the reaction network \eqref{e.birthdeathSI} is Poisson with mean $k_{id_i}/k_{-i}$ as in \eqref{e.cond_disSI},

Assume, instead, that $\mathrm X_i$ forms a multimer. In order to simplify notations, we drop the index $i$ and write $  \emptyset \xrightleftharpoons[ k_{-}]{ k} \mathrm X, \ n \mathrm X \xrightleftharpoons[  \beta_{-}]{ \beta} \mathrm X_{n}.$  Let $x_1,x_2$ denote the molecular counts of $\mathrm X,\mathrm X_n$. Then, the master equation is
\begin{align}\label{master2} \dot p_{x_1,x_2} &= \left (k p_{x_1-1,x_2}  -  k_- x_1 p_{x_1,x_2}  \right ) + \left (  k_-(x_1+1) p_{x_1+1,x_2}  - k p_{x_1,x_2}  \right ) \\ &+ \left (  \beta_{-2} (x_2+1) p_{x_1-2,x_2+1}  - \tfrac 1{n!}  \beta \prod_{k=0}^{n-1}(x_1-k) p_{x_1,x_2}\right ) \\ & + \left ( \tfrac 1{n!} \beta \prod_{k=1}^n(x_1+k) p_{x_1+n,x_2-1}   -  \beta_{-}x_2 p_{x_1,x_2}  \right )..\end{align}
We solve the recurrence equation assuming detailed balance, and then verify that   the obtained solution, which is given in \eqref{e.cond_disSI}, solves \eqref{master2} \hfill $\blacksquare$

\subsection{The Stationary Distribution as a Mixture of Poisson Distributions}
We include here the proof of Theorem 3.

Recall the slow-fast decomposition \red{of the master equation} in Eq. (14).  {Recall the joint stationary distribution \eqref{pi_defSI}. In order to emphasize the dependence on $\varepsilon$ we} \red{denote $\bar\pi^\varepsilon:=\bar \pi(\varepsilon)$}. \red{Hence,} $\bar\pi^\varepsilon$ is the unique stationary distribution that satisfies
 $ \Lambda_\varepsilon \bar\pi^\varepsilon=0$, $\pi^\varepsilon>0$, and $\sum_z \pi_z^\varepsilon = 1 $,  {where the subscript denotes the value of the stationary distribution at $z$}.

  Our objective is to characterize the stationary distribution as $\varepsilon \to 0$.
 Writing $\bar\pi_\varepsilon$ as an asymptotic expansion  {to first order} in terms of  $\varepsilon$, we have
\begin{equation}\label{e.expansionSI} \bar\pi^\varepsilon = \bar\pi^{(0)} + \bar\pi^{(1)} \varepsilon + o(\varepsilon).\end{equation}

Our aim is to find $\bar\pi^{(0)}$.
Substituting $\bar\pi^\varepsilon$ in Eq. (14),  and equating the coefficients of the powers of $\varepsilon$ to zero we obtain the following two equations:
\begin{align}\label{e.1steq}  {\tilde \Lambda \bar \pi^{(0)}}
& = 0 \\\label{e.2ndeq}    {\tilde \Lambda \bar \pi^{(1)}+ \hat\Lambda \bar\pi^{(0)}}
&= 0
\end{align}
 {where $\tilde \Lambda$ is given in Eq. (15).}
  \eqref{e.1steq} implies that $\bar\pi^{(0)} \in \ker \tilde\Lambda$\rv{, where $\ker$ denotes the kernel of $\tilde\Lambda$.}  We next show how to compute $\ker \tilde\Lambda$.

 Recall the conditional Markov chains with the associated infinitesimal generators as in Eq. (16). By the assumptions,  for each $d \in \{0,..,L-1\}$ there exists a unique {$\pi_{X|d}$} such that:
   {$ \Lambda_d \pi_{X|d} = 0,$
  $\pi_{X|d}>0$, and $\sum_x \pi_{X|d}(x) = 1 $}. {Recall} that $\pi_{X|d}$ is the stationary distribution of the Markov chain conditioned on $D(t)=d$.

  Defining the extended conditional distributions for $d=0,.., L-1$ as: \begin{equation} \label{e.concat}\bar \pi_{X|d}:=[ \overbrace{\mathbf 0^T \ ... \ \mathbf 0^T }^{d-1} \pi_{X|d}^T  \overbrace{\mathbf 0^T \ ... \ \mathbf 0^T }^{L-d} ] ^T.\end{equation}
  \rv{The stationary distribution above can be interpreted as a function as follows:} $\bar \pi_{X|d}(x,d)=\pi_{X|d}(x)$, and $\bar \pi_{X|d}(x,d')=0$ when $d'\ne d$.

 Then $\ker \tilde\Lambda = \mbox{span}\{\bar \pi_{X|0},..,\bar \pi_{X|L-1} \}$. Hence, we can write:
\[\bar\pi^{(0)}= \sum_{i=0}^{L-1} \lambda_i \bar \pi_{X|i},\]
for some $\lambda_0,...,\lambda_{L-1}\ge 0$. We normalize them to satisfy $\sum_{d=0}^{L-1} \lambda_d=1$.

In order to satisfy \eqref{e.2ndeq}, we utilize the fact that each $\Lambda_d$ is an infinitesimal generator which satisfies $\mathbf 1^T \Lambda_d=0$. Hence, we pre-multiply \eqref{e.2ndeq} by the vectors: $[\mathbf 1^T \ \mathbf 0^T \ ...   \mathbf 0^T]^T$, $[\mathbf 0^T \ \mathbf 1^T \ ...  \mathbf 0^T]^T$, $[\mathbf 0^T \ \mathbf 0^T \ ...    \ \mathbf 1^T]^T$ in order to get the following $L$-dimensional linear system:
\begin{equation}\label{e.solutionSI} \Lambda_r \lambda:=\begin{bmatrix}\mathbf 1^T & \mathbf 0^T &...  &  \mathbf 0^T \\ \mathbf 0^T & \mathbf 1^T &...  &  \mathbf 0^T \\ & & \ddots & \\ \mathbf 0^T & \mathbf 0^T &...  &  \mathbf 1^T \end{bmatrix} \hat\Lambda \, [\bar \pi_{X|0} \ \bar \pi_{X|1} \ ... \ \bar \pi_{X|L-1} ] \begin{bmatrix} \lambda_0 \\ \vdots \\ \lambda_{L-1} \end{bmatrix} = 0.
\end{equation}
   {Furthermore,}  we need the following {normalization} equation to find $\lambda_0,...,\lambda_{L-1}$ uniquely:
\begin{equation}\label{e.sum1}
\lambda_0+ ... + \lambda_{L-1} = 1.
\end{equation}

This is equivalent to  stating that $\lambda=[\lambda_0,...,\lambda_{L-1}]$ is the principal eigenvector of $\Lambda_r$. \hfill $\blacksquare$

\subsection{Computation of the Reduced-Order Markov Chain's Generator}
Recall that the generator of the reduced-order Markov chain can be written as follows:
\begin{equation}\label{e.solutionSI_explicit}
  \Lambda_r=\begin{bmatrix}\mathbf 1^T \hat \Lambda_{00} \pi_0 & \dots & \mathbf 1^T \hat \Lambda_{0{L-1}} \pi_{L-1} \\ \vdots &\ddots & \\ \mathbf 1^T \hat \Lambda_{{L-1}0} \pi_0 & \dots & \mathbf 1^T \hat \Lambda_{(L-1)(L-1)} \pi_{L-1}  \end{bmatrix}.
\end{equation}

The $(d',d)$ entry represents the probability of transition from the configuration $d'$ to configuration $d$, and it can be interpreted as a weighted conditional expectation of $\pi_d$.

 {Consider the reduced chain},  {and} fix a configuration $d$. Then, the maximum number of possible  {transitions out of $d$} is {given by the number of reactions which is} $\frac 12 \sum_{i=1}^N |B_i| $. Hence, $\Lambda_{r}$ is a sparse matrix  {for large $N$}. {Computation of} the infinite matrices and matrix product in \eqref{e.solutionSI_explicit} \red{can be cumbersome for networks with multiple genes. Hence,} we provide an algorithm for computing the nonzero entries in $\Lambda_r$.
This can be achieved by considering all the possible transitions from a configuration $d=(d_1,..,d_N)$. {Specifically, we consider a transition} from $d$ to  $d'$ by a gene reaction modifying a single promoter configuration. For instance consider $D_{d_i}$. Then for a constitutive or single TF-gene binding/unbinding, there can be only one transition starting  {from} $D_{d_i}$. This transition is  either the forward or reverse reaction in (1) or (2), respectively. For the case of two TFs, there can be two reactions among (3)-(6).

The algorithm can be described as follows:
\begin{proposition}\label{th.algorithm}The matrix $\Lambda_r$ in \eqref{e.solutionSI} can be computed via the algorithm below.
\end{proposition}
\begin{itemize}\itshape
\item For each $d \in \{0,..,L-1\}$ write $d=(d_1,..,d_N)\in \prod_{i=1}^N B_i$. Using the previously discussed identification:
\begin{itemize}
\item[-]  {Let $\mathscr R_d=\{\rm R_1,..,\rm R_{|\mathscr R_d|}\}$ the set of all gene reactions. Then, for each $j \in \{1,..,|\mathscr R_d|\}$}:
      \begin{enumerate}
   \item Let  {$D_{d_{i} }^i$, and} $D_{d_{i'} }^i$ be the  {reactant and product} configuration species of the {$\rm R_j$}.  {Hence, the reaction will cause a transition from $d$ to} $d'=(d_1,..,d_{i'},..,d_N)$. Let $\alpha$ be the kinetic constant of  {$\rm R_j$}. If  {$\rm R_j$} is a binding reaction, then let $\mathrm X_{\bar i }$, $\mathrm X_{\bar ic }$  denote the TF or the multi-merized TF,  {where $\bar i$ denotes the index of the gene that expresses the TF.}
  \item Then, the $(d',d)$ entry of $\Lambda_r$ can be written as:
     \begin{equation}\label{e.ratesSI}
         [\Lambda_r]_{d'd}= \left \{ \begin{array}{ll} \alpha, & \mbox{if the reaction is monomolecular}  \\
       \\  {\frac{\alpha}{n_{\bar i}!}  \frac{ \beta_{\bar i  }}{  \beta_{-\bar i}}\left (\frac{k_{\bar i d_{\bar i }}}{k_{-\bar i}} \right )^{n_{\bar i}}}, & \mbox{if the reaction is bimolecular}\end{array}\right .
     \end{equation}
   \end{enumerate}

\item[-] Set \begin{equation}\label{e.rate_d}\mathbf [\Lambda_r]_{d d}  = -\sum_{i'\ne i} \mathbf 1^T \hat \Lambda_{d_{i'} d_i } \pi_{d}.\end{equation}
\end{itemize}
\item Set the rest of the entries of $\Lambda_r$ to zero. \\
\end{itemize}

\emph{Proof.} Recall that in Eq. (14), the matrix $\hat\Lambda$ represents the slow matrix, which corresponds to the gene binding reactions. Hence, $\hat\Lambda_{dd'}$ represents the matrix corresponding to the transition between states of the form $(x,d)$ and $(x,d')$. Assume that $D(t)=d=(d_1,...,d_N)$. Consider the $i$th block. Note that it has one or two reactions that can fire. Specifically, there are $\frac 12 |B_i|$ gene reactions that can fire. Assume that such a reaction is in one of the forms:
\[ \mathrm D_{d_i}^i \mathop{\to}\limits^\alpha  \mathrm  D_{d_{i'}}, \mathrm  D_{d_i}^i \mathop{\to}\limits^\alpha \mathop{\to}\limits^\alpha \mathrm D_{d_{i'}}+\mathrm X_{\bar i},\]
where $\mathrm X_{\bar i}$ is a TF for that block. Then, \[ [\Lambda_r]_{d'd}=\mathbf 1^T \Lambda_{d'd}\pi_d =\alpha \mathbf 1 \pi_d = \alpha.\]
Now consider a reaction of the form:
\[ \mathrm X_{\bar i} +\mathrm D_{d_i}^i \mathop{\to}\limits^\alpha \mathrm D_{d_{i'}}^i.\]
Then,
\[ [\Lambda_r]_{d'd}=\mathbf 1^T \Lambda_{d'd}\pi_d =\alpha \sum_{x_{\bar i}=0}^\infty x_{\bar i} \pi_{d\bar i} \sum_{x_{ i \ne \bar i}} \pi_{d i} = \alpha \mathbb E [ X_{\bar i}(t)| D(t)=d] =  \alpha  \frac{k_{\bar i d_{\bar i }}^{n_{\bar i}} \beta_{\bar i }}{n_{\bar i}!k_{-\bar i d_{\bar i }}^{n_{\bar i}} \beta_{-\bar i}}.\]
 The last equality follows from evaluating the mean value of the Poisson distribution in \eqref{e.cond_disSI2}.

Finally, \eqref{e.rate_d} holds since $\mathbf 1^T\Lambda_r=0$, which follows from $\mathbf 1^T \hat \Lambda=0$.
 \hfill $\blacksquare$

\section{Detailed Discussion of Examples}

\subsection{The Gene Bursting Model}

We start with the simplest form of network, which is the {autonomous} TF-gene binding/unbinding model. It has been verified as a model for transcriptional bursting \cite{raj06}. This model has been studied analytically  using time-scale separation \cite{qian09},\cite{grima15}, \rv{Poisson-representations \cite{iyer14}, and the exact steady solution is known \cite{shahrezaei08}, \cite{wang14}.}

Consider:
  \begin{align}\label{aut_swiSI}
\nonumber \mathrm D_0 &\xrightleftharpoons[\varepsilon \alpha_{-}]{\varepsilon \alpha} \mathrm D_1 \\
  \mathrm D_1 &\mathop\rightarrow\limits^{k} \mathrm X+\mathrm D_1, \\ \nonumber
\mathrm X &  \mathop\rightarrow\limits^{k_-} 0.
\end{align}

 {Referring to Figure  2, we identify a single gene block with two states. Using \eqref{e.cond_disSI}, the conditional stationary distributions are:
\begin{align*}   \pi_{0}(x) &= \delta(x) \\
  \pi_{1}(x) &= \mathbf P(x; k/k_-).
\end{align*}
\red{In order to compute the}  {stationary distribution $\pi$, we need to find the} \red{generator for the reduced chain}. Since both reactions are monomolecular, we write the following  {using \eqref{e.ratesSI}:}
\[\Lambda_r =\begin{bmatrix} -\alpha & \alpha_- \\ \alpha & -\alpha_- \end{bmatrix}.\] Hence,} the reduced Markov chain is a binary Bernoulli process with a rate of $\alpha/(\alpha+\alpha_-)$.
Then the stationary distribution of $X$ can be written using \eqref{e.mixture_mSI} as:

\begin{equation}\label{e.bursting}\pi(x) =\frac{\alpha_{-}}{\alpha+\alpha_{-}}  \mathbf P(x;0) +  \frac{\alpha}{\alpha+\alpha_{-}} \mathbf P(x;k/k_-),\end{equation}
which is a bimodal distribution with peaks at 0 and $k/k_-$.
 {The fast promoter kinetics model is obtained, instead, by} reversing  {the} time-scale separation {such that the protein reactions become slow and gene reactions become fast. In that case, the resulting stationary distribution can be shown to be} a Poisson with mean $\frac{\alpha}{\alpha+\alpha_-} \frac{k}{k_-}$ which is the same as the deterministic equilibrium if we used the conservation law $D_1(t)+\mathrm D_0(t)=1$ for the above model. Finally, note that the  {mean of the} slow promoter kinetics model {is} the same  {as in the fast kinetics model} but the two stationary distributions differ drastically.

\rv{ Figure \ref{f.bursting} shows the transition from fast to slow promoter kinetics using the exact solution \cite{wang14} and compares it to the expression \eqref{e.bursting}.}
 \begin{figure}[t]
 \centering
 {\includegraphics[width=0.65\textwidth]{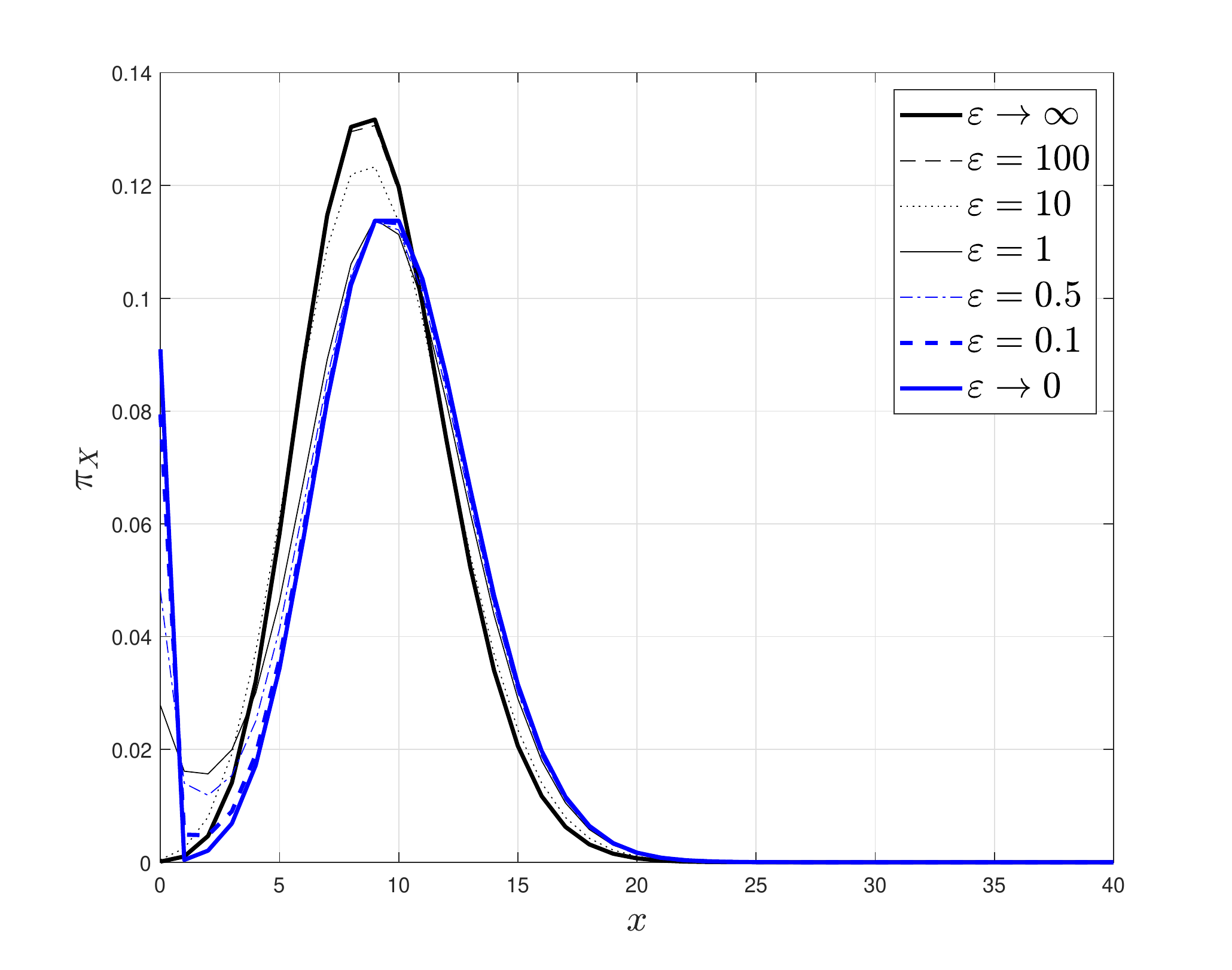}}
 \caption{\rv{The stationary probability distribution for different $\varepsilon$ which shows the transition from fast promoter kinetics,  {i.e., $\varepsilon \to \infty $}, to slow promoter kinetics,  {i.e., $\varepsilon \to 0 $}, in a single unregulated gene.  {The stationary distribution is bimodal for small $\varepsilon$, i.e. $\varepsilon \le 1$, and unimodal for large $\varepsilon$.} The deterministic equilibrium coincides with the fast kinetics mode at $ \frac{\alpha}{\alpha+\alpha_-}\frac k{k_-}$. The slow kinetic limit is calculated via \eqref{e.bursting}, the fast kinetics limit is a Poisson centered at the deterministic equilibrium, while the remaining curves are computed by evaluating the exact solution given in \cite{wang14} . The parameters are $\alpha_-=0.1, \alpha=1, k_-=2, k=20$.}} \label{f.bursting}
\end{figure}

\subsection{A Self-Regulating Gene}
Consider a non-cooperative self-regulating gene Eq. \eqref{noncoopSI}.

 {Referring to Figure 2}, this is a single-gene block with two states. At the limit of slow promoter kinetics,  {Remark 4 implies that} the gene binding/unbinding reaction can be written as in Eq. (29) \red{ as follows:
\[\mathrm D_0 \xrightleftharpoons[ \alpha_{-}]{  \alpha  k_0/k_-} \mathrm D_1.\]}
 {Using \eqref{e.ratesSI}}, \red{ the reduced generator can be written as:
\[\Lambda_r =\begin{bmatrix} -\alpha k_0/k_- & \alpha_- \\ \alpha k_0/k_- & -\alpha_- \end{bmatrix}.\]
}
 {Hence it} defines a binary Bernoulli process with the rate $\alpha k_0/(\alpha k_-+\alpha k_0)$. Using \eqref{e.mixture_mSI} the stationary distribution is a mixture of two Poisson distributions and can be written as:
\begin{equation}\label{e.ncSI} \pi_1(x)= \frac {\alpha \rho_1}{\alpha_- +\alpha \rho_1} \mathbf P(x;k_1/k_-) + \frac {\alpha_-}{\alpha_-+\alpha \rho_1} \mathbf P(x;k_0/k_-), \end{equation}
where \[\rho_1=\mathbb E[X_2|D=0]=k_0/k_-.\]

Next, consider the same reaction network, but now with cooperativity:
 \begin{align}
\nonumber \mathrm X_2+\mathrm D_0 &\xrightleftharpoons[\varepsilon \alpha_{-}]{\varepsilon \alpha} \mathrm D_1 \\ \label{coopSI}
  \mathrm D_0 &\mathop\rightarrow\limits^{k_{0}} D_0+\mathrm X, \\
  \mathrm D_1 &\mathop\rightarrow\limits^{k_{1}} \mathrm D_1+\mathrm X,\nonumber \\
\mathrm X &  \mathop\rightarrow\limits^{k_{-}} 0 \nonumber \\
\nonumber 2\mathrm X &\xrightleftharpoons[\beta_-]{\beta} \mathrm X_2.
\end{align}
In this case, the gene process is still a Bernoulli process, but with a different rate. The stationary distribution for $X$ can be written as:
\begin{equation}\label{e.c} \pi_{2}(x)= \frac {\alpha \rho_2}{\alpha_- +\alpha \rho_2} \mathbf P(x;k_1/k_-) + \frac {\alpha_-}{\alpha_-+\alpha \rho_2} \mathbf P(x;k_0/k_-), \end{equation}
where \[\rho_2=\mathbb E[X_2|D=0]=\frac{ k_0^2 \beta}{2k_-^2 \beta_-}.\]

Both distributions \eqref{e.ncSI}, \eqref{e.c} have modes  at $\frac{k_{1}}{k_-}$ and   $\frac{k_{0}}{k_-}$. The height of the first mode is proportional to $\rho_1$ for \eqref{e.ncSI}, and is proportional to $\rho_2$ for \eqref{e.c}. The network is activating if $k_{1}>k_{0}$, and repressing otherwise.

Comparing \eqref{e.ncSI} and \eqref{e.c}, note that, in the non-cooperative case, if we want to increase the weight of the mode corresponding to the bound state keeping the association ratio, then  the mode location needs to be changed.  {On the other hand,}  {the factor $\rho_2$ in   the dimerization rates in \eqref{e.c}}  can be used in order to tune the weights freely while keeping the modes and the binding to unbinding kinetic constants ratio unchanged. For instance, we can make the distribution effectively unimodal with a sufficiently high dimerization ratio.

\red{ A non-cooperative self-regulating gene with slow promoter kinetics has been studied in literature by deriving close-form expression \cite{hornos05} and using time-scale separation \cite{qian09} . However, the gene binding/unbinding reaction in both papers was approximated by an auto-catalytic reaction:
\[\mathrm X+\mathrm D_0  \xrightleftharpoons[\varepsilon \alpha_{-}]{\varepsilon \alpha} \mathrm D_1+\mathrm X.\]
This has the advantage of decoupling the slow and the fast processes. However, it is a simplification of the physical process.  We did not use such simplifications.
}

\paragraph{Special cases:} \strut \\
The model above {considers} a network with possibly non-zero production rates for both  {the unbound and bound} promoter configurations. We may also {consider} the special cases of pure self-activation or self-repression  {which are explained below:}
\begin{enumerate}
  \item \emph{Pure Self-Activation}, i.e. $k_{0}=0$ {in \eqref{noncoopSI} and \eqref{coopSI}}. Then the  peak  {corresponding to the bound configuration} disappears and we get only one peak at  {zero protein copy number} as it forms an absorbing state. This is a manifestation of the Keizer's paradox  {\cite{qian07}}. One way to circumvent this is to allow for a small transcriptional ``leak''. This amounts to taking  $k_0 \ll k_1$, and it allows us to recover the second mode.
  \item \emph{Pure Self-Repression}, i.e.  $k_{1}=0$ {in \eqref{noncoopSI} and \eqref{coopSI}}. Then we get two modes: one at 0 and the other at $k_{0}/k_-$.
\end{enumerate}
For both cases, in the non-cooperative case  {with a fixed dissociation ratio} the choice of this kinetic rate determines completely the relative weight of the modes as in \eqref{e.ncSI}. Cooperativity allows us to tune the weight of the mode corresponding  {to} the bound state without changing the location as mentioned before.

\paragraph{Comparison with Fast  {Promoter Kinetics}:} \strut \\
In order to demonstrate that slow switching is responsible for the emergence of new modes {compared to the deterministic model}, consider the non-cooperative {self-regulating gene}  network \eqref{noncoopSI} with fast  {promoter kinetics} modeled by letting $\varepsilon $ grow without bound in the first reversible TF-promoter binding/unbinding reactions.
We state the following proposition which is proved in the Methods section:
\begin{proposition}\label{th.fast}
As $\varepsilon \to \infty$, the stationary distribution of the network \eqref{noncoopSI} is given by:
\[\pi(m)=\lim_{t \to \infty}\Pr[X(t)=m] = \frac {\alpha_- w_m + \alpha mw_{m+1}}{\alpha m+\alpha_-},\]
           where
           $w_m$ satisfies the following recurrence relation:
\begin{align*}
w_{m+1}&=\frac{((k_1 m +(\alpha_-/\alpha) k_0)(m+(\alpha_-/\alpha)+1)}{k_- (m+1)( m+(\alpha_-/\alpha))^2} w_m, \ m\ge 0.
\end{align*}
where $w_0$ is chosen to satisfy $\sum_{m=0}^\infty  w_m=1$.
\end{proposition}
\begin{proof}
A decomposition dual to Eq. (14) can be written and it can be noted that the fast matrix is block-diagonal with respect to $\mathrm X+\mathrm D_1$ which is the slow variable, while the fast variable is  $\mathrm D_1$.

Expanding asymptotically the stationary distribution in terms of $\varepsilon$, and taking the limit as $\varepsilon$ goes to zero, we can find the distribution for the slow variable as follows:
\[ \Pr[X+D_1=m]= w_m,\]
where $w_m$ satisfies the following recurrence relation:
\begin{align*}
w_{m+1}&=\frac{((k_1 m +(\alpha_-/\alpha) k_0)(m+(\alpha_-/\alpha)+1)}{k_- (m+1)( m+(\alpha_-/\alpha))^2} w_m, \ m\ge 0.
\end{align*}
The joint distribution can be given as:
\begin{align*}
\Pr[X+D_1=m,D_1=0]&= w_m \frac{\alpha_-}{\alpha m+\alpha_-},\\
\Pr[X+D_1=m,D_1=1]&= w_m \frac{\alpha m}{\alpha m+\alpha_-}.
\end{align*}

Hence we can compute the marginal density of $X$ as follows:
\begin{align*} \Pr[X=m]&=\Pr[X+D=m,D=0]+\Pr[X+D=m+1,D=1] \\
           &= \frac {\alpha_- w_m + \alpha mw_{m+1}}{\alpha m+\alpha_-}. \end{align*}

\end{proof}

Since the ratio $w_{m+1}/w_m$ is a ratio of two polynomials and the denominator's degree is higher than the numerator, then stationary distribution is unimodal, while  slow TF-gene binding/unbinding was shown to give a bimodal distribution  {(see \eqref{e.ncSI})}.

\subsection{The Toggle Switch}

A toggle switch is a basic GRN that exhibits deterministic multi-stability. It has two stable steady states and can switch between them with an external input or via noise.   The ideal behavior is that only one gene is ``on'' at any moment in time. We now study the network with the slow promoter kinetics. Consider the following network with cooperativity indices $n,m$:
   \begin{align*}\begin{array}{rl}
\nonumber \mathrm Y_m+ \mathrm D_0^X &\xrightleftharpoons[\alpha_{-1}]{\alpha_1}  \mathrm D_1^X \\
  \mathrm D_0 &\mathop\rightarrow\limits^{k_{10}} \mathrm D_0^X+\mathrm X, \\
\mathrm X &   \mathop\rightarrow\limits^{k_{-1}} 0 \\
n\mathrm X  &\xrightleftharpoons[\beta_{-1}]{\beta_1}  \mathrm X_n \\ \end{array} \begin{array}{rl}
\nonumber \mathrm X_n+ \mathrm D_0^Y &\xrightleftharpoons[\alpha_{-2}]{\alpha_2} \mathrm D_1^Y  \\
\nonumber  \mathrm D_0^Y &\mathop\rightarrow\limits^{k_{20}} D_0^Y+\mathrm Y, \\ \mathrm Y & \mathop\rightarrow\limits^{k_{-2}} 0\\
m\mathrm Y  &\xrightleftharpoons[\beta_{-2}]{\beta_2}  \mathrm Y_m. \end{array}
\end{align*}
 For the case $n,m=1$,  {there is no multi-merization reaction}.  {For consistency,} we choose $\beta_{ 1}=\beta_{-1},\beta_2=\beta_{- 2}=1$  {in that case}.

 {Denote the promoter configuration species by $\mathrm D^X,D^Y$.} Then the network has four configurations $(\mathrm D^X,\mathrm D^Y)\in\{(0,0),(0,1)$ $,(1,0),(1,1)\}$. Using Corollary 4 we expect to have a stationary distribution with four modes $(k_{10}/k_{-1},k_{20}/k_{-2})$, $(k_{10}/k_{-1},0)$, $(0,k_{20}/k_{-2})$, $(0,0)$. \red{Using the algorithm of}  {Proposition 5,} the reduced-order Markov chain infinitesimal generator is:
\begin{equation}\label{e.switch_matrix}
  \Lambda_r=\begin{bmatrix} -\alpha_{1} \rho_2 -\alpha_2 \rho_1 & \alpha_{-2} & \alpha_{-1} & 0 \\ \alpha_2 \rho_1 & - \alpha_{-2} & 0 & \alpha_{-1} \\ \alpha_1 \rho_2 & 0 & -\alpha_{-1} & \alpha_{-2} \\ 0 & 0 & 0 & -\alpha_{-2}-\alpha_{-1}
  \end{bmatrix},
\end{equation}
where
\begin{equation}\label{e.rho_toggleSI} \rho_1 = \left ( \frac{k_{10}}{k_{-1}} \right ) ^n \frac{\beta_{1}}{n! \beta_{-1} }, \quad  \rho_2 = \left ( \frac{k_{20}}{k_{-2}} \right ) ^m \frac{\beta_{2}}{m! \beta_{-2} }. \end{equation}
We notice immediately  {from the last row in the matrix \eqref{e.switch_matrix}} that the transition rates towards the configuration (1,1) are zero, which implies that {the weight of the mode corresponding to $(1,1)$ is zero}. Hence, we have three modes only.   {The weights corresponding to the modes} can be found as the principal eigenvector of $\Lambda_r$ as given in Corollary 4. Hence, the stationary distribution for $X,Y$ is:
\begin{equation} \pi(x,y)= \frac 1{\frac{\alpha_1}{\alpha_{-1}} \rho_2 + \frac{\alpha_2}{\alpha_{-2}} \rho_1 + 1 }\left ( \mathbf P(y;{ \tfrac{k_{20}}{k_{-2}}}) \mathbf P(x;{ \tfrac{k_{10}}{k_{-1}}})  +  \frac{\alpha_1}{\alpha_{-1}} \rho_2  \mathbf P(y;{ \tfrac{k_{20}}{k_{-2}}})\delta(x) + \frac{\alpha_2}{\alpha_{-2}} \rho_1  \mathbf P(x;{ \tfrac{k_{10}}{k_{-1}}}) \delta(y) \right).  \end{equation}

 {Since the stationary distribution has three modes, it deviates from the ideal behavior of a switch where at most two stable steady states, under appropriate parameter conditions,  are possible. Nevertheless, } a bimodal distribution  can be achieved by minimizing the weight of the first mode at $(\frac{k_{10}}{k_{-1}},\frac{k_{20}}{k_{-2}})$. If we fix $\alpha_{1}/\alpha_{-1}, \alpha_{2}/\alpha_{-2}$, then this can be satisfied by tuning $n,m,\beta_{\pm 1},\beta_{\pm 2}$ to maximize $\rho_1,\rho_2$ in \eqref{e.rho_toggleSI}. Choosing higher cooperativity indices, subject to $n< k_{10}/k_{-1},m<k_{20}/k_{-2}$, achieves this.

The toggle switch has three modes regardless of the cooperativity index. This is unlike the deterministic model where only one positive stable state is realizable with non-cooperative binding,  {and two stable steady states are realizable with cooperative binding.}
However,   the toggle switch with fast switching can admit three modes in some parameter ranges.  {In contrast to the case of slow switching under consideration here}, the third stable state is the (low,low) state \cite{ma12}.

\rv{Monte-Carlo simulations via Gillespie's algorithm for the cooperative toggle switch have been performed to investigate the minimum time-scale separation to recover our predictions. We have chosen the parameters such that the (low,high), and (high,low) modes get 0.25 each, and the (high,high) modes gets 0.5 at the slow promoter limit. Figure \ref{f.toggle_monto}-a shows that the three modes predicted by the analysis have been recovered with promoter kinetics being just \emph{two times} slower than the protein decay rate. Figure \ref{f.toggle_monto}-b shows a sample trajectory where the three modes are visible.}

 { 
  \paragraph{Alternative model of the toggle switch}  \strut \\
Instead of modeling the toggle switch   as two independent genes with their own promoters, an alternative model consists of a signal promoter regulating two operons \cite{warren04}, \cite{warmflash07}.  In our modelling framework this amounts to a single gene expression block with two transcription factors and two expressed proteins, i.e., two inputs and two outputs. Despite the fact that our formalism accounts for a single expressed protein, we find no difficulty in applying our methods as shown below.

The alternative model can be written as:
\begin{equation}\label{fate1SI}
 \begin{array}{llll}
  \mathrm X_c+ \mathrm D_{00} \xrightleftharpoons[]{}  \mathrm D_{10},\quad &    \mathrm X_c+ \mathrm D_{01} \xrightleftharpoons[]{}  \mathrm D_{11}, &
 \mathrm Y_c+ \mathrm D_{00} \xrightleftharpoons[]{}  \mathrm  D_{01}, \quad &        \mathrm Y_c+  \mathrm  D_{10}  \xrightleftharpoons[]{}  \mathrm  D_{11}, \quad \\
  \mathrm D_{00} \mathop\rightarrow\limits^{k_{x}} \mathrm  D_{00}  + \mathrm X,\quad &  \mathrm D_{00}  \mathop\rightarrow\limits^{k_{y}} \mathrm D_{00} + \mathrm Y,&   \mathrm D_{01}  \mathop\rightarrow\limits^{k_{x}} \mathrm  D_{01} + \mathrm X, \quad &
     \quad \mathrm D_{10}  \mathop\rightarrow\limits^{k_{y}}  \mathrm  D_{10} + \mathrm Y, \quad \\ \mathrm X \mathop\rightarrow\limits^{k_{-x}} \emptyset, \quad& \mathrm Y \mathop\rightarrow\limits^{k_{-y}} \emptyset, \quad &
  n\mathrm X  \xrightleftharpoons[\beta_{-x}]{\beta_x} \mathrm X_c, &\quad  m\mathrm Y \xrightleftharpoons[\beta_{-y}]{\beta_y} \mathrm Y_c,

\end{array}
\end{equation}

Assuming slow-promoter kinetics, the stationary distribution can be decomposed into a mixture of Poisson distributions centered at $(\frac{k_{x}}{k_{-x}},\frac{k_{y}}{k_{-y}}), (\frac{k_{x}}{k_{-x}},0), (0,\frac{k_{y}}{k_{-y}})$. Despite the fact that model appears to be different, the resulting modes and the reduced-order Markov chain are similar to the model used before.
}

\begin{figure}
\centering
\subfigure[]{\includegraphics[width=0.75\textwidth]{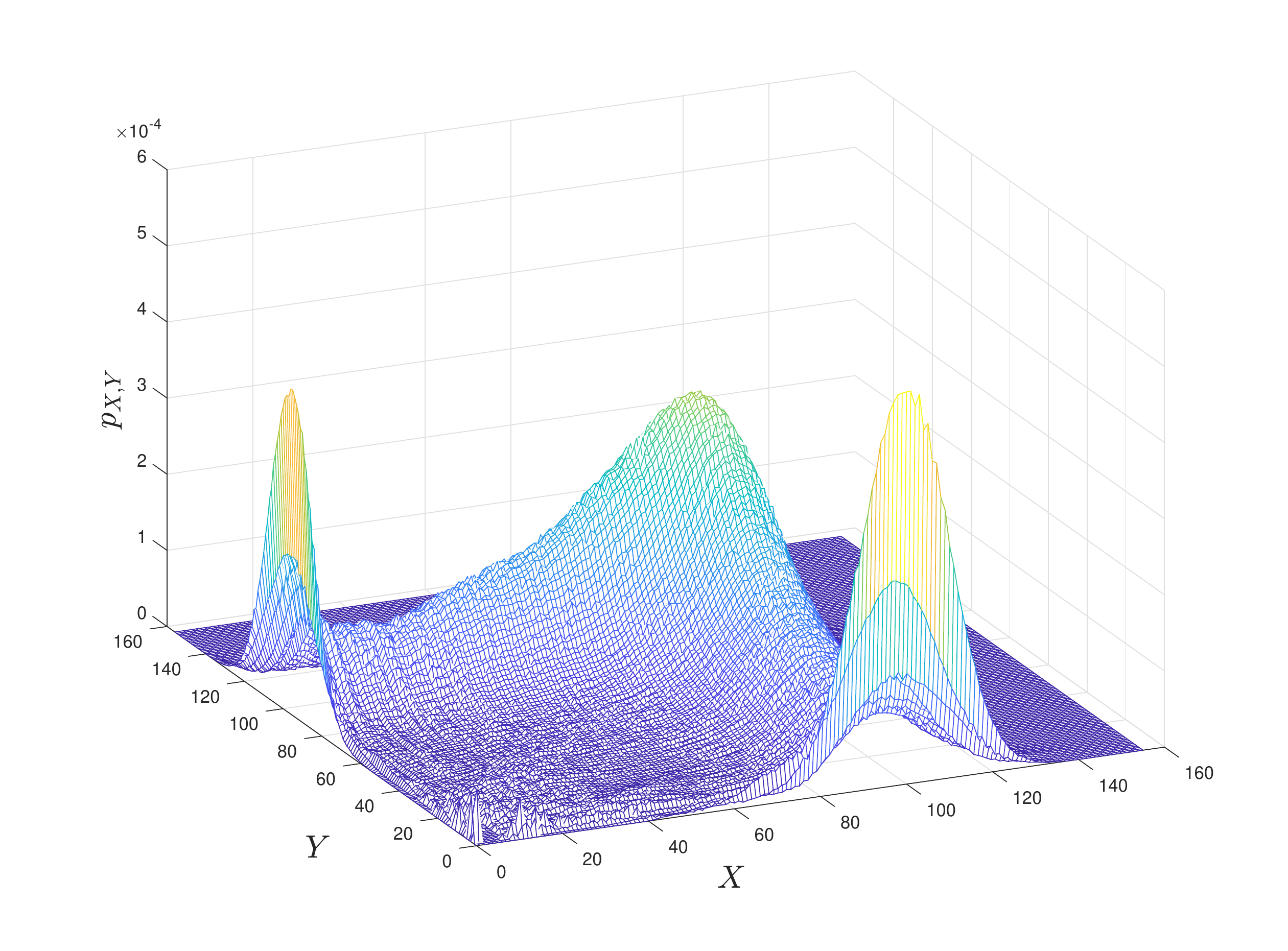}}
\subfigure[]{\includegraphics[width=0.49\textwidth,height=1.25in]{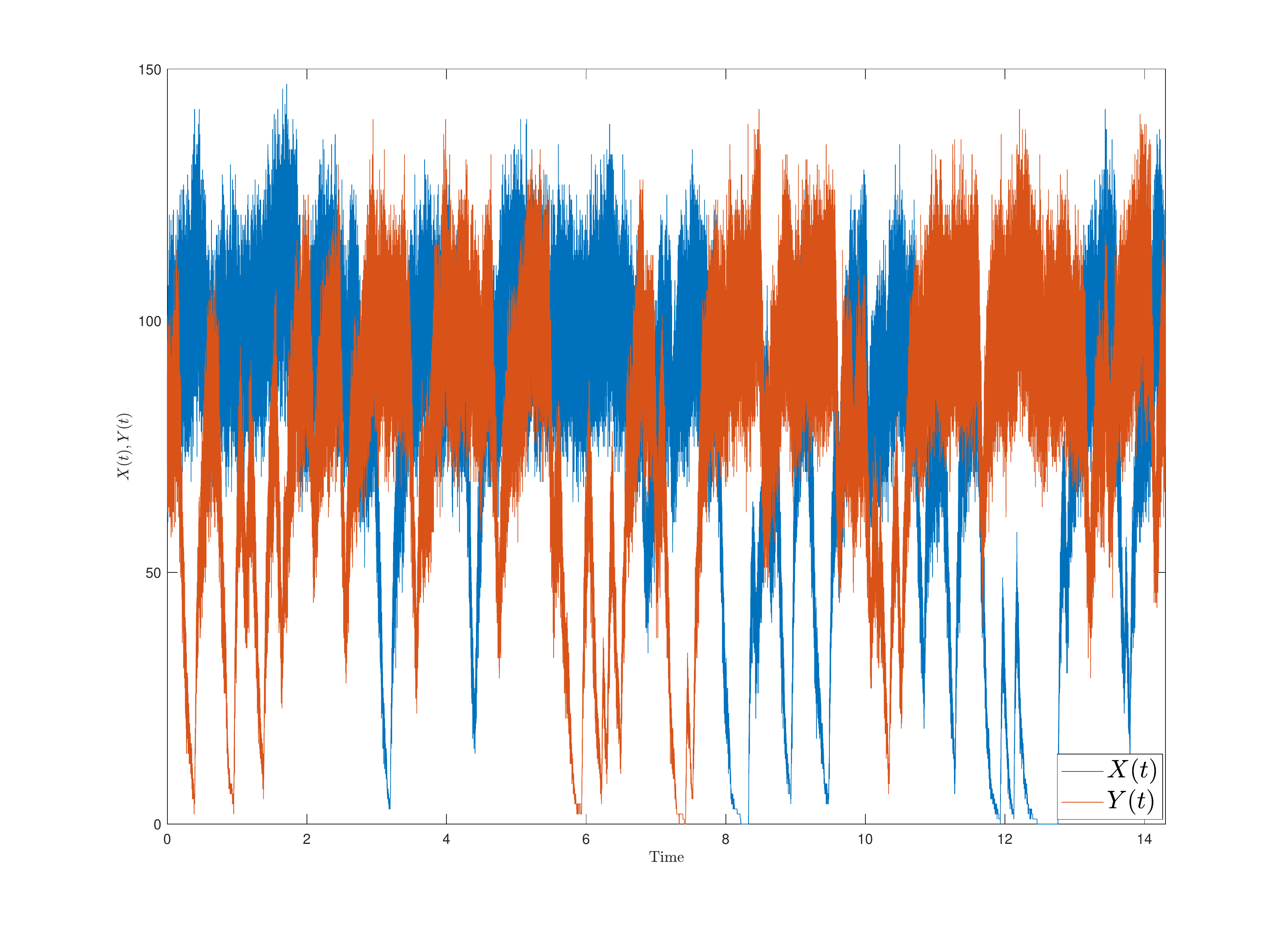}}
\subfigure[]{\includegraphics[width=0.49\textwidth,height=1.25in]{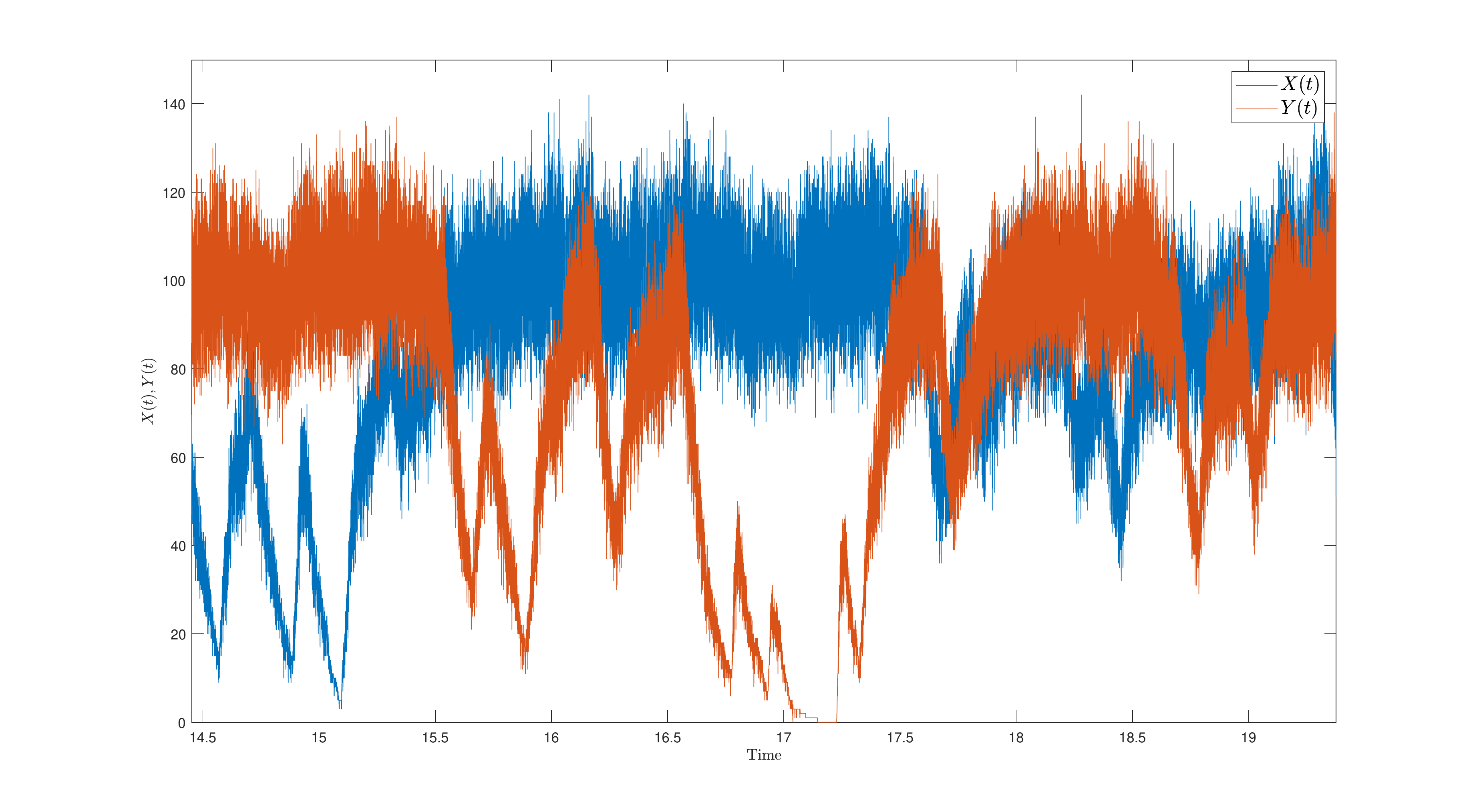}}
\caption{ \textbf{Analysis predicts the stationary distribution with 0.5 time scale separation.} (a) The empirical stationary distribution has been computed based on a realization of the stochastic process. The realization has been computed via Gillespie's algorithm and has been averaged over $70\times 10^6$ iterations. (b) A sample trajectory of the toggle switch. (c) A zoomed plot showing the three modes  Observe that between times 18 and 18.5 both modes are high. The parameters are: $\alpha_-=10, \alpha=0.1, k_-=20, k=2000, \beta=20, \beta_-=2000$. Note that the unbinding rate is two times slower than the protein decay rate. } \label{f.toggle_monto}
\end{figure}

{
\subsection{The Repressilator}
The repressilator is a synthetic biological circuit that implements a ring oscillator \cite{elowitz00}, and it has been simulated with slow-promoter kinetics [21]. It is a canonical example of a GRN that exhibits a limit cycle, i.e. sustained oscillation. The block diagram is shown in Figure \ref{f.repressilator}-a which shows three genes in a repression cycle. The list of reactions is given as follows:
\begin{align*}\begin{array}{rl}
\nonumber \mathrm Z_n+ \mathrm D_0^X &\xrightleftharpoons[\alpha_{-}]{\alpha}  \mathrm D_1^X \\
  \mathrm D_0 &\mathop\rightarrow\limits^{k_{}} \mathrm D_0^X+\mathrm X, \\
\mathrm X &   \mathop\rightarrow\limits^{k_{-}} 0 \\
n\mathrm X  &\xrightleftharpoons[\beta_{-}]{\beta}  \mathrm X_n \\ \end{array} \begin{array}{rl}
\nonumber \mathrm X_n+ \mathrm D_0^Y &\xrightleftharpoons[\alpha_{-}]{\alpha} \mathrm D_1^Y  \\
\nonumber  \mathrm D_0^Y &\mathop\rightarrow\limits^{k}  D_0^Y+\mathrm Y, \\ \mathrm Y & \mathop\rightarrow\limits^{k_{-}} 0\\
n\mathrm Y  &\xrightleftharpoons[\beta_{-}]{\beta}  \mathrm Y_n. \end{array}
\begin{array}{rl}
\nonumber \mathrm Y_n+ \mathrm D_0^Z &\xrightleftharpoons[\alpha_{-}]{\alpha} \mathrm D_1^Z  \\
\nonumber  \mathrm D_0^Z &\mathop\rightarrow\limits^{k} D_0^Z+\mathrm Z, \\ \mathrm Z & \mathop\rightarrow\limits^{k_{-}} 0\\
n \mathrm Z  &\xrightleftharpoons[\beta_{-}]{\beta}  \mathrm Z_n. \end{array}
\end{align*}

Deterministic analysis of the repressilator \cite{ddv_book} reveals that it does not oscillate with non-cooperative binding. We proceed to apply our analysis methods to study the behaviour in the stochastic case with slow promoter kinetics.

{Denote the promoter configuration species by $\mathrm D^X,\mathrm D^Y,\mathrm D^Z$.} Then the network has eight configurations $(\mathrm D^X,\mathrm D^Y,\mathrm D^Z)\in\{(0,0,0),(0,0,1)$ $,...,(1,1,1)\}$. Using Corollary 4 we expect to have a stationary distribution with eight modes. \red{Using the algorithm of}  {Proposition 5,} the reduced-order Markov chain infinitesimal generator is:
\begin{equation}\label{e.repressilator_matrix}
  \Lambda_r=\begin{bmatrix}
  -3 \alpha \rho & \alpha_- & \alpha_- & 0  & \alpha_- & 0 & 0 & 0 \\
  \alpha \rho & -\alpha_- - \alpha \rho & 0 & \alpha_- & 0   & \alpha_- & 0 & 0 \\
   \alpha \rho & 0 & -\alpha_- - \alpha \rho& \alpha_- & 0 & 0   & \alpha_- &0 \\
   0 &  \alpha \rho & 0 & -2 \alpha_- &  0 &0&0   & \alpha_- \\
   \alpha \rho &0&0&0& -\alpha_- - \alpha \rho  & \alpha_-& \alpha_-&0 \\
   0&0&0&0& \alpha \rho & -2\alpha_- & 0 & \alpha_- \\
   0 & 0 & \alpha \rho & 0 & 0 &0 & -2\alpha_-  & \alpha_- \\
   0&0&0&0  &0&0&0& -3\alpha_- .
  \end{bmatrix},
\end{equation}
where $ \rho =  \left ( \frac{k }{k_{- }} \right ) ^n \frac{\beta }{n! \beta_{- } }. $

The stationary distribution is:
\begin{align}\label{e.repressilator_pmf}\pi(x,y,z)&=\frac 1D \Big (  w\mathbf P(x,y,z;K,0,0)+w\mathbf P(x,y,z;0,K,0)+w\mathbf P(x,y,z;0,0,K)   \\ &\left . +\frac{1}{2w} \mathbf P(x,y,z;K,K,K)+ \mathbf P(x,y,z;K,K,0)+ \mathbf P(x,y,z;0,K,K)+ \mathbf P(x,y,z;K,0,K)\right ), \nonumber \end{align}
where $K=\frac{k}{k_-}, w=2\rho\frac{\alpha}{\alpha_-}, D=3w+3+\tfrac 1{2w}$. Notice that the Poisson distribution centered at (0,0,0) has zero weight. Hence, the network can admit up to seven modes.

We will study the behaviour of this network in two cases: when $w\gg 1$, and $w \approx 1$.

\paragraph{The ``Stochastic Oscillator''} \strut \\
Although the distribution is a mixture of seven modes, note that it reduces to three effective modes if \[ w= 2\frac{\alpha}{\alpha_-}   \left ( \frac{k }{k_{- }} \right ) ^n \frac{\beta }{n! \beta_{- } } \gg 1, \]
where the three modes are located at $d_x=(K,0,0),d_y=(0,K,0),d_z=(0,0,K)$.   The tri-modal stationary distribution is consistent with the classical oscillations of the repressilator, and this is independent of the cooperativity index. Note that this condition is analogous to the oscillation condition in the deterministic model \cite{ddv_book} (but with cooperativity only) which also requires ``large'' production ratio.

In order to study whether the network oscillates, we need to define a notion of limit cycle for a stochastic system. Due to randomness, the time-series can not be periodic. Nevertheless, since the stationary distribution is tri-modal, we say that the   network oscillates if the time trajectory continues to jump between the modes in the same order.

Assume $w \gg 1$. Let $d_x,d_y,d_z$ be the three dominant modes. We will show that if the reduced-order Markov chain is at mode $r_x$ then it is much more likely to transition to $d_y$ rather than to $d_z$. Similar arguments apply if we start from $d_y,d_z$.

 Let $Q(t)=e^{t\Lambda_r}$ be the probability transition matrix \cite{norris}. The $(i,j)$ entry in the transition matrix denotes the probability of being at $i$ at time $t$ if we start from $j$. Hence, $Q_{ij}(t)=\Pr[D(t)=i|D(0)=j]$.

We are interested in comparing the probabilities of transiting from $r_x$ to $r_y,r_z$. Hence we study small $t \ll 1$. We  utilize the power series expansion to evaluate $Q(t)$. Let $\Lambda_d=\mbox{diag}(\Lambda_r)$ and $\Lambda_+=\Lambda_r - \Lambda_d$, where $\mbox{diag}(\Lambda_r)$ is the diagonal matrix which contains the diagonal entries of $\Lambda_r$. Note that $\Lambda_+$ is positive. Since $\Lambda_d$ and $\Lambda_+$ commute, we can write
\[Q(t)= e^{\Lambda_d} e^{\Lambda_+}.\]
We can approximate $Q(t)$ for small $t$ by writing the Taylor expansion of $e^{\Lambda_+}$. Hence, using \eqref{e.repressilator_matrix} we write a third-order Taylor series for the transition probabilities that we need:
\begin{align*}
  Q_{d_x d_x} &= e^{-2t}+e^{-2t}o(t)=e^{-2t} ( 1+wt^2) + e^{-2t} o(t^3) \\
  Q_{d_y d_x} & = e^{-2t}o(t)= e^{-2t} (\tfrac 12 wt^2) + e^{-2t} o(t^3) \\
  Q_{d_z d_x} & = e^{-2t} o(t^3),
\end{align*}
where we assumed, w.l.o.g, that $\alpha_-=1$.

Using the expressions above, if the Markov chain is at $r_x$ then it is most likely to stay there. The transition is much more likely to happen to $d_y$ rather than $d_z$. Hence, we expect to see ``long'' periods of protein $X$ being expressed, and then it jumps to express protein $Y$, and then  protein $Z$.   Since the finite Markov chain is ergodic, the pattern repeats.

Note that the analysis above predicts that both the cooperative and the noncooperative repressilator are capable of oscillation with slow-promoter kinetics when $w \gg 1$.
We performed Monte-Carlo simulations via the Gillespie algorithm for both fast and slow kinetics. The results are shown in Figure \ref{f.repressilator_s}. We observe that the network always oscillates with cooperative binding. With non-cooperative binding, only the network with slow kinetics oscillates as predicted. The network with fast kinetics does not oscillate. Recall that the deterministic model with non-cooperative binding does not oscillate \cite{ddv_book}.

\paragraph{Multi-modality in the repressilator with slow-promoter kinetics} \strut \\
In the case that $w$ is close to 1, the seven modes in \eqref{e.repressilator_pmf} share the probability almost equally.  This is independent of the cooperativity index. Compare this the non-oscillatory deterministic model, which is mono-stable and it can't admit multiple stable equilibria.

\begin{figure}[t]
\centering
{ \raisebox{0\height}{\includegraphics[width=0.23\textwidth]{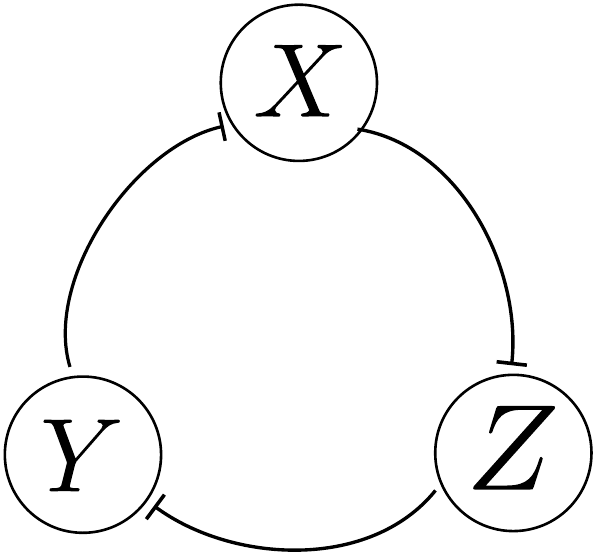}}}
\caption{  A \red{diagram of the} repressilator.   }
\label{f.repressilator}
\end{figure}

\begin{figure}[t]
\centering\subfigure[Cooperative, Slow-Promoter Kinetics.]{ {\includegraphics[width=0.5\textwidth,height=1in]{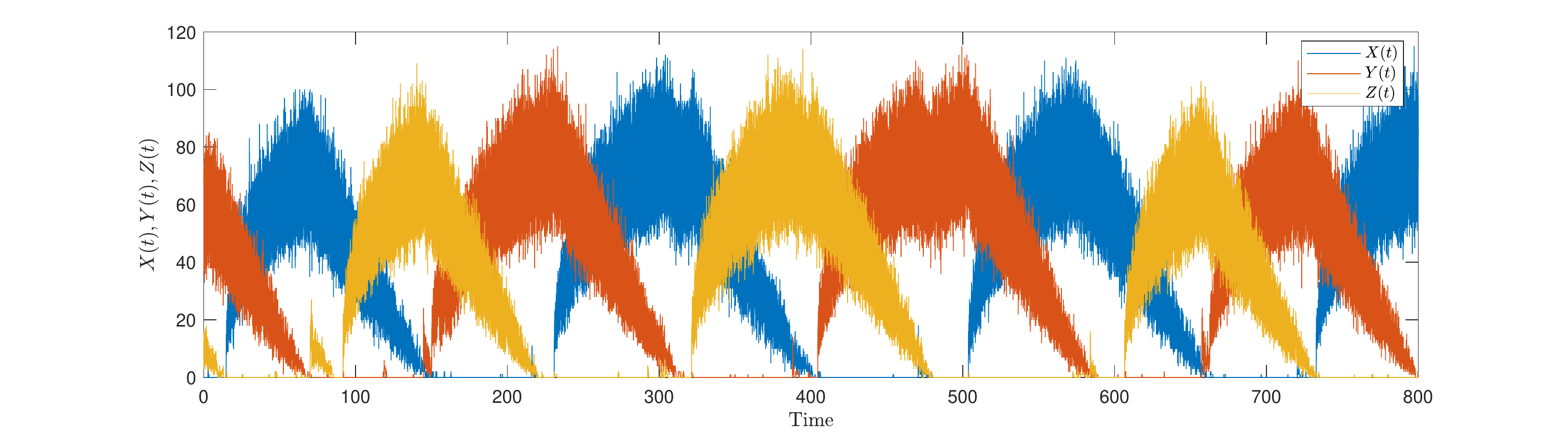}}}
\subfigure[Noncooperative, Slow-Promoter Kinetics.]{\includegraphics[width=0.475\textwidth,height=1in]{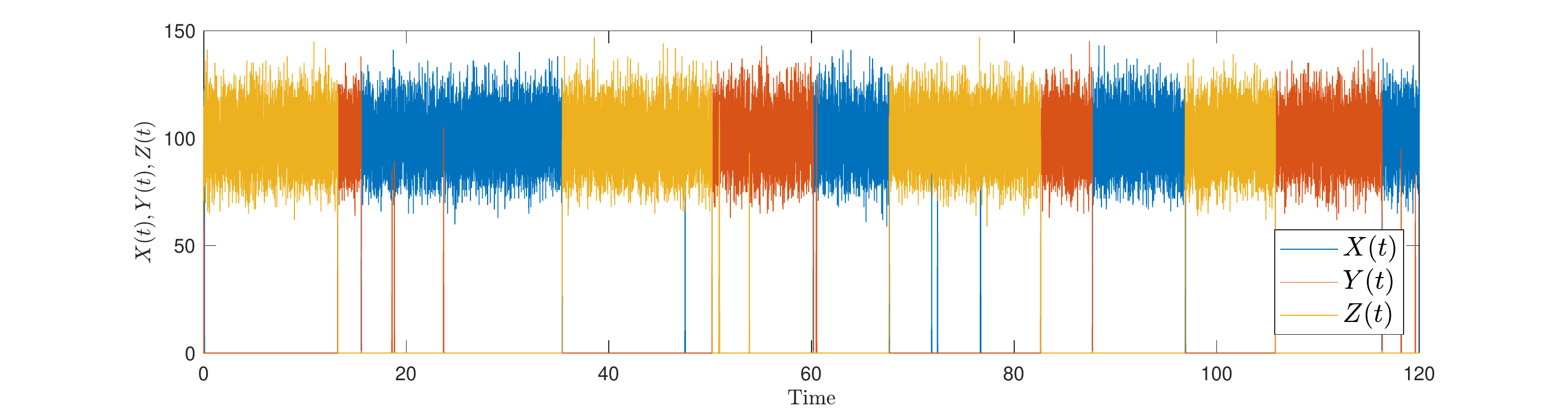}}
\subfigure[Cooperative, Fast-Promoter Kinetics.]{{\includegraphics[width=0.485\textwidth,height=1in]{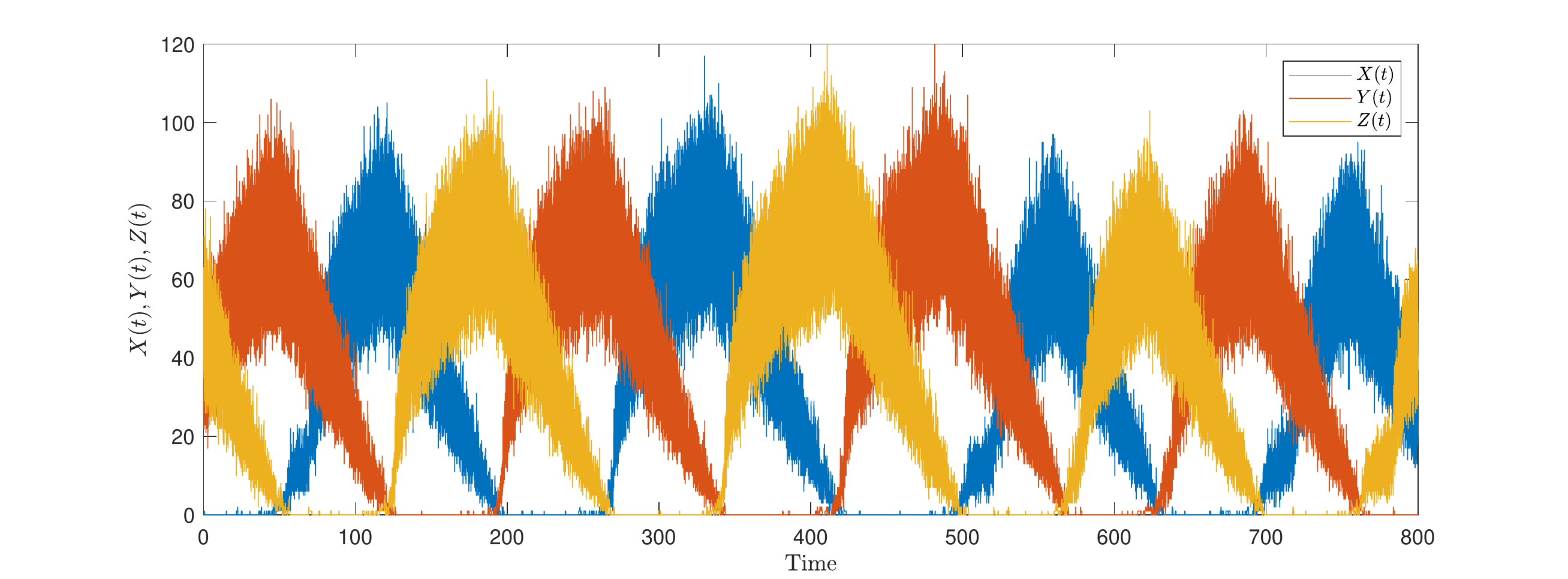}}}
\subfigure[Noncooperative, Fast-Promoter Kinetics.]{\includegraphics[width=0.485\textwidth,height=1in]{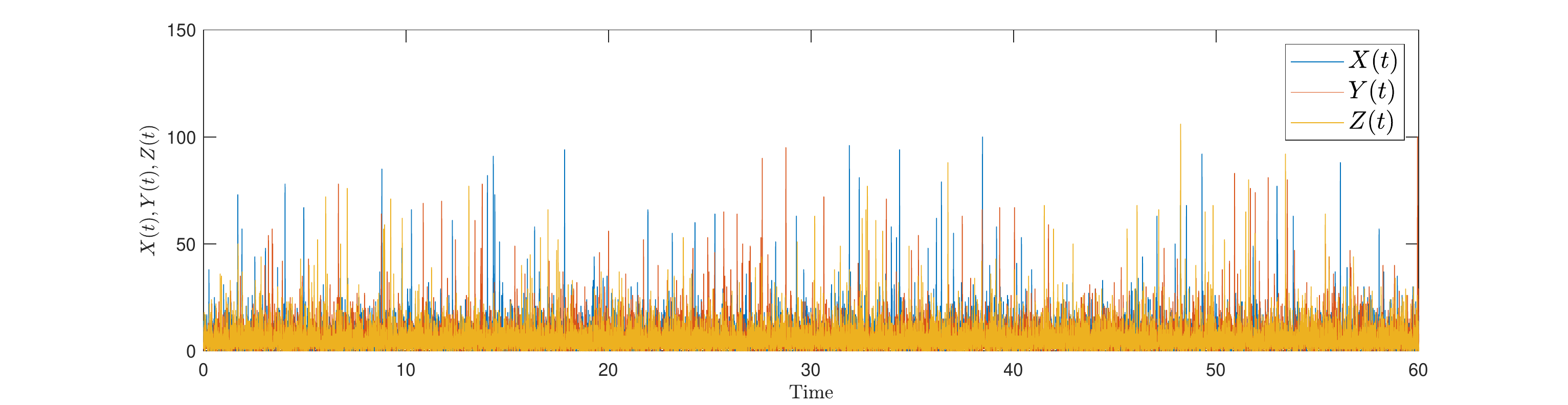}}
\caption{\rv{\textbf{ The noncooperative repressilator oscillates.  } (a) A time-series for the cooperative repressilator with cooperativity index 2, and slow promoter kinetics. (b) A time-series for the cooperative repressilator with cooperativity index 2, and fast promoter kinetics. (c) A time-series for the noncooperative repressilator, and slow promoter kinetics. (d) A time-series for the noncooperative repressilator, and fast promoter kinetics. The plots were generated by stochastic simulation via the Gillespie algorithm. For all the figures, the parameters are: $ \alpha=5\varepsilon, \alpha_-=1\varepsilon, k=2000, k_-=20, \beta_{\pm}=1$, where $\varepsilon=0.1$ for slow kinetics, and $\varepsilon=1000$ for fast kinetics.}}
\label{f.repressilator_s}
\end{figure}

}
 \subsection{Synchronization of interconnected toggle switches}
  We consider $N$ identical toggle switches:
\begin{equation}\begin{array}{rcl}
\nonumber \mathrm Y_{ic}+   \mathrm D_{0}^{xi}  &\xrightleftharpoons[\varepsilon \alpha_{-x} ]{\varepsilon \alpha_{x} }& \mathrm D_1^{xi} \\
  \mathrm D_{0}^{xi} &\mathop\rightarrow\limits^{k_x}& \mathrm X_i+ \mathrm D_{0}^{xi}, \\
\mathrm X_i &  \mathop\rightarrow\limits^{k_{-x}}& 0, \\
n\mathrm X_i &  \xrightleftharpoons[\beta_{-x} ]{\beta_x }& \mathrm X_{ic},
\end{array} \begin{array}{rcl}
\nonumber \mathrm X_{ic}+   \mathrm D_{0}^{yi}  &\xrightleftharpoons[\varepsilon \alpha_{-y} ]{\varepsilon \alpha_{y} }& \mathrm D_1^{yi} \\
  \mathrm D_{0}^{yi} &\mathop\rightarrow\limits^{k_y}& \mathrm Y_i+ \mathrm D_{0}^{yi}, \\
\mathrm Y_i &  \mathop\rightarrow\limits^{k_{-y}}& 0, \\
n\mathrm Y_i &  \xrightleftharpoons[\beta_{-y} ]{\beta_y }& \mathrm Y_{ic},
\end{array}\end{equation}
where $i=1,..,N$.
{We interconnect these systems through diffusion of
the protein species $\mathrm X_i,\mathrm Y_i$ among cells, modeled
through reversible
reactions with a diffusion coefficient $\Omega$:
\begin{equation}\label{e.diffusionSI}
  \mathrm X_i \xrightleftharpoons[\Omega ]{\Omega } \mathrm X_j,
  \quad \mathrm Y_i \xrightleftharpoons[\Omega ]{\Omega } \mathrm Y_j, \quad i \ne j, \;i,j=1,..,N.
\end{equation}
}

For a deterministic model, there exists a parameter range for which all toggle switches will synchronize into bistability for sufficiently high diffusion coefficient \cite{sontag16}.  {This implies each switch in the network behaves as a bistable switch, and it converges with all the other switches to the same steady-states. }

Our aim is to analyze the stochastic model at the limit of slow promot{e}r kinetics and compare it with the deterministic model. The network has $4^N$ promoter configurations. Consider a  {promoter} configuration \[ d=(d^X,d^Y):=(d_1^X,..,d_N^X,d_1^Y,...,d_N^Y) \in \left\{0,..,4^{N}-1\right\},\] where $d_i^X,d_j^Y \in \{0,1\}$. As discussed before, we need to represent the Markov chain conditioned on $D(t)=d$. Each conditional Markov chain can be represented as follows:
\begin{align*}\label{e.conditionaltoggle}
\begin{array}{rcl}
\emptyset &  \xrightleftharpoons[k_{-x} ]{d_i^X k_{x} }& \mathrm X_{i},
\\
n\mathrm X_i &  \xrightleftharpoons[\beta_{-x} ]{\beta_x }& \mathrm X_{ic},
\end{array} \begin{array}{rcl}
\emptyset &  \xrightleftharpoons[k_{-y} ]{d_i^Y k_{y} }& \mathrm Y_{i},
 \\
n\mathrm Y_i &  \xrightleftharpoons[\beta_{-y} ]{\beta_y }& \mathrm Y_{ic},
\end{array}\\
\mathrm X_i \xrightleftharpoons[\Omega ]{\Omega } \mathrm X_j,
\quad \mathrm Y_i \xrightleftharpoons[\Omega ]{\Omega } \mathrm Y_j, \quad  i,j=1,..,N.
\end{align*}
This conditional Markov is not the forms \eqref{e.birthdeathSI}, \eqref{e.birthdeathSIDim} that arise from the class of GRNs defined previously as in Figure 1. Nevertheless, it can be observed that it is a reversible zero-deficiency network \cite{feinberg87}. We show in the SI-\S 4.1 that our results can be generalized to networks that admit weakly reversible deficiency zero conditional Markov chains.

As shown in the SI-\S 4.1, {there are $4^N$ conditional Markov chains, and their stationary distributions are Poisson distributions with the following set of modes}:
\begin{equation}\label{e.toggleNetModes}
  \left\{ \left (\frac{k_x}{k_{-x}} \omega_{d_1}^X,\dots,\frac{k_x}{k_{-x}} \omega_{d_N}^X, \frac{k_y}{k_{-y}} \omega_{d_1}^Y,\dots,\frac{k_y}{k_{-y}} \omega_{d_N}^Y   \right ) : d=(d^X,d^Y) \in \left\{0,..,4^{N}-1\right\} \right \},
\end{equation}
where
\begin{equation}\label{e.toggleNetModes2} \omega_{d_i}^X = \frac{\sum_{i=1}^N \bar d_i^X + \bar d_i^X k_{-x}/\Omega}{N+k_{-x}/\Omega}, i=1,..,N, \end{equation}
in which $\bar d_i^X=1-d_i^X$. $\omega_{d_i}^Y$ is defined similarly.

Using Theorem 3, the stationary distribution is a mixture of $4^N$
Poisson distributions, and the weights can be found by finding the principal
eigenvector of the  {infinitesimal generator of the reduced-order} Markov chain as before. {Proposition 5, gives the procedure to find the matrix $\Lambda_r$, where Eq. (27) is replaced by \eqref{e.ratesSI_g} with \eqref{e.cond_disSI}}. As before, the Markov
states are the configurations $d=(d^X,d^Y)\in \{0,..,4^N-1\}$.  Assume
that $D(t)=d$. Then, the state transitions are given by the following
reactions:
\begin{equation}\label{e.togglenet_slow}  \mathrm D_{0}^{xi}  \xrightleftharpoons[
    \alpha_{-x} ]{ \alpha_{x} \rho_{d_i}^{X} } \mathrm D_1^{xi},
  \quad
  \mathrm D_{0}^{yi}  \xrightleftharpoons[ \alpha_{-y} ]{ \alpha_{y} \rho_{d_i}^{Y} } \mathrm D_1^{yi},
\end{equation}
where
\begin{equation} \label{e.toggleNetWeight}\rho_{d_i}^{X}=\mathbb E[\mathrm X_{ic}|D(t)=d]=\frac 1{n!} \frac{\beta_x}{\beta_{-x}}\left(\frac{k_x}{k_{-x}}\omega_{d_i}^X \right)^n,\end{equation}
and $\rho_{d_i}^{Y}$ is defined analogously.

Note that the  {mode corresponding to the state in which all the TFs are bound to the promoters} has no incoming transitions in the reduced Markov chain and hence it has zero weight, hence the network has $4^N-1$ modes.

 {We consider now the case of a high diffusion coefficient.} Note from \eqref{e.toggleNetModes}, \eqref{e.toggleNetModes2}, \eqref{e.toggleNetWeight} that as $\Omega\to \infty$, $X_1,..,X_N$ will synchronize in the sense that the joint distribution of $X_1,..,X_N$ is symmetric with respect to all permutations of the random variables. This implies that the marginal stationary distributions $p_{X_i}, i=1,..,N$ are identical. Hence, for sufficiently large $\Omega$  the probability mass is concentrated around the region for which $X_1,..,X_N$ are close to each other.
Consequently, for large $\Omega$ we can replace the population of toggle switches with  a \emph{single toggle switch} with the \emph{synchronized protein processes} $X(t), Y(t)$,  {which are defined, for the sake of convenience, as $X(t):=X_1(t), Y(t):=Y_1(t)$.} Next, we describe the stationary distribution of $X(t),Y(t)$.

 From \eqref{e.toggleNetModes2} it can be seen that $\omega_{d_i}^X$ does not depend on $d_i$ for large $\Omega$. Instead it depends only on $\sum_{i=1}^N \bar d_i^X$, which is the total number of unbound promoter sites in the genes producing $\mathrm X_1,..,\mathrm X_N$. The same holds for $\omega_{d_i}^Y$. Hence, the number of modes will drop from $4^N-1$ to $(N+1)^2-1$. Hence the joint distribution of $X,Y$ is a mixture of Poisson distributions with the following modes:
 \[ \left \{ \left ( \frac{i k_x}{Nk_{-x}}, \frac{j k_y}{Nk_{-y}}   \right ): i,j=0,..,N, (i,j) \ne (0,0) \right\}.\]

 Note that similar to the single toggle switch, there are modes which have both $X,Y$ with non-zero copy number. On the other hand, there are many additional modes.  {Recall that in the case of a} single toggle switch,  {we have tuned the cooperativity ratios such that the modes in which both genes are ON are suppressed}. Similarly, the undesired modes can be suppressed by tuning the cooperativity ratio {which} can be achieved by choosing $\rho_{d_i}^X,\rho_{d_i}^Y, d=0,..,4^N-1$ sufficiently large. In particular, letting the multi-merization ratio $\beta_x/\beta_{-x},\beta_y/\beta_{-y} \to \infty$, the weights of modes in the interior of the positive orthant $\mathbb R_+^2$ approach zero.

 In conclusion, for sufficiently high $\Omega$ and sufficiently high multimerization ratio the population behaves as a \emph{multimodal switch },  {which means that the whole network can have either the gene $\mathrm X$ ON, or the gene $\mathrm Y$ ON. And every gene can take} $2N$ modes which are:
  \[ \left \{ \left ( \frac{i k_x}{Nk_{-x}}, 0   \right ),\left ( 0, \frac{i k_y}{Nk_{-y}}   \right ): i=1,..,N \right\}.\]
Comparing to the low diffusion case, the network will have up to $2^N-1$ modes with sufficiently high multimerization ratio.

\paragraph{Comparison with the deterministic model}
For $\Omega$ greater then a certain threshold, the deterministic system bifurcates into bistabiliy. This means that all toggle switches converge to the same exact equilibria if $\Omega$ is greater than the threshold.  As we have seen before, this is not the case for the stochastic system, since the toggle switches converges \emph{asymptotically} to each other. Hence, we need to choose a  {a threshold} for  $\Omega$ that  { constitutes ``sufficient'' synchronization.}  {We choose to define this as} the protein processes synchronizing within one copy number. In other words, we require the maximum distance between the modes in \eqref{e.toggleNetModes} to be less than 1.

This amounts to requiring the diffusion coefficient needs to satisfy:
\begin{equation}\label{q_inequalitySI} \Omega\ge\frac 1N \max\{k_x -k_{-x}, k_y-k_{-y} \}. \end{equation}

We derive \eqref{q_inequalitySI} below.\\
In order for the maximum distance between the modes in \eqref{e.toggleNetModes} to be less than 1, we need to satisfy this inequality:
\[\max_{d\in \{0,..,4^{n-1}\}} \max_{i,j=1,..,N} \left\{\left|\frac{k_x}{k_{-x}} \omega_{d_i}^X-\frac{k_x}{k_{-x}} \omega_{d_j}^X\right|, \left|\frac{k_y}{k_{-y}} \omega_{d_i}^Y-\frac{k_y}{k_{-y}} \omega_{d_j}^Y\right|\right\} < 1.\]

Let us consider first maximizing the term containing the variables related to the gene $\mathrm X$. Note from \eqref{e.toggleNetModes2} that $\omega_{d_i}^X$ and $\omega_{d_j}^X$ can differ only in a $k_{-x}/\Omega$ term in the numerator. Hence, the maximization can be simplified to:

\[\max \left\{\frac{k_x}{k_{-x}} \left( \frac{ k_{-x}/\Omega}{N+ k_{-x}/\Omega} \right ), \frac{k_y}{k_{-y}}\left( \frac{k_{-y}/\Omega}{N+ k_{-y}/\Omega} \right )\right\} < 1.\]
Solving for $\Omega$ yields \eqref{q_inequalitySI}.

\rv{\section{Extension to GRNs with Complex-Balanced Conditional Markov Chains}}
In the main text we have included assumptions to simplify the mathematical treatment and notations while retaining the same qualitative features of the problem.  In this section, we elucidate the manner in which the results can be generalized.

The class of networks defined in the main text does not allow direct interactions between the proteins. However, our main results hold for a more general class of networks.
Figure 2 depicts a generalized gene expression block where the proteins expressed by other genes participate in the gene reactions block. The gene expression block remains unchanged.
The basic theory presented in the main text (Proposition 1, Theorem 3) remains unchanged. Proposition 5 needs can be modified by replacing Eq. (27) by the following:
   \begin{equation}\label{e.ratesSI_g}
         [\Lambda_r]_{d'd}= \left \{ \begin{array}{ll} \alpha, & \mbox{if the reaction is monomolecular}  \\
       \\ \alpha \mathbb E[X_{\bar ic}|D=d], & \mbox{if the reaction is bimolecular}\end{array}\right .
     \end{equation}

However, since there are no closed form formulae for computing the conditional expectation in general, the expression above is of limited utility.
Nevertheless, we can define a more general class of network such that the conditional stationary distributions can be computed analytically as  obtained in Proposition 2.
In order to define our generalized class of networks, we need some further \rv{notation and background which are summarized in the following subsection}.

\subsection{Complex-Balancing and the Concept of Deficiency}
Recall that a reaction network consists of the set of species $\mathscr S$ and the set of reactions $\mathscr R$. A reaction $\mathrm R_j \in \mathscr R$ can be written as Eq. \eqref{e.reaction}.

The formal linear combination in the left-side of the reaction is called the \emph{reactant complex}, while the linear combination on the right-side of the reaction is called the \emph{product complex}. The set of all complexes in the network is denoted by $\mathscr C$. Each complex can be interpreted as a vector belonging to the Euclidean vector space whose unit vectors correspond to the species.

Define the following matrix:
\[ \tilde \Gamma = [A|B], \, \mbox{where} \, [A]_{ij}=\alpha_{ij}, [B]_{ij}=\beta_{ij},\]
where $\alpha_{ij}, \beta_{ij}$ are the stoichiometry coefficients. If we remove the redundant columns of $\tilde\Gamma$, this yields an $|\mathscr S| \times |\mathscr C|$ matrix $\Gamma_c$, where each column corresponds to a certain complex. Hence, it can be verified that there exists a matrix $Y \in \{0,\pm 1\}^{\mathscr |C| \times |\mathscr  R | } $ such the that the stoichiometry matrix $\Gamma$ admits the following factorization:
\begin{equation}\label{Stoich_factor} \Gamma = \Gamma_c Y. \end{equation}

\begin{figure}
  \centering
  \includegraphics[width=0.7\textwidth]{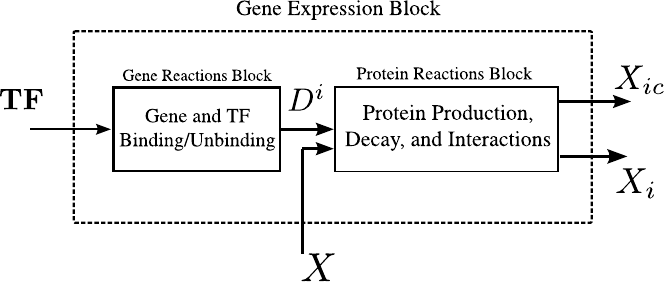}
  \caption{A generalized representation of the gene expression block. $X$ represent the vector of all protein monomers in the network.}\label{f.blockdiagram_g}
\end{figure}
There is significant literature studying the reaction networks using the notion of complexes, initially for deterministic systems, \cite{horn72,feinberg87}, and then for stochastic systems \cite{anderson10,sontag10}. The ordinary differential equation for a reaction network can be written as:
\begin{equation}\label{ode} \dot x = \Gamma R(x), \end{equation}
where $x$ is the concentration vector, and $R$ is a reaction rate function which is commonly defined using mass-action kinetics as follows:
\[R_j(x)=k_j \prod_{i=1}^{|\mathscr S|} x_i^{a_{ji}},\]
where $k_j$ is reaction rate constant.
 A conservation law is a nonnegative vector $v$ such that $v^T \Gamma=0$. The stoichiometric compatibility class containing $x_0$ is defined as $(x_\circ + \Im \Gamma) \cap \mathbb R_{\ge 0}^n$. If there are no conservation laws then $\mathbb R_{\ge 0}^n$ is a stoichiometric compatibility class.

In order to state the subsequent results, we define a \emph{modified} mass-action kinetics.  For a given a reaction rate function $R$, the modified reaction rate function is given as
\[\tilde  R_j(x)=\tilde k_j \prod_{i=1}^{|\mathscr S|} x_i^{a_{ji}}, \, \mbox{with} \, \tilde k_j:= \frac{k_j}{\prod_{i=1}^{|\mathscr S|} a_{ji}! }.\]

\rv{As per \eqref{Stoich_factor}, the complex formation rate is defined as $R_c(x):=YR(x)$. }  An equilibrium $\hat x$ of \eqref{ode} is called \emph{complex balanced} if  it satisfies $R_c(x)=0$ also. The existence of a single complex-balanced equilibrium guarantees that all equilibria of the network are complex balanced \cite{horn72}.

\rv{ Physically, reaction rates in $\ker Y$ do not change the complex-formation rate. Hence, if $\ker{\Gamma}=\ker{Y}$, then every equilibrium is complex-balanced. This ensures that a network is complex balanced regardless of kinetic rate constants. Hence, the \emph{deficiency} of the network is defined as $\delta:= \mbox{dim}( \ker(\Gamma_c) \cap \mbox{Image}(Y))$. If $\delta=0$ then every equilibrium is complex-balanced, and the network is said to be a zero-deficiency network.}

A network is weakly reversible if: \rv{existence of a} directed path from complex $C_i$ to complex $C_j$, \rv{implies the existence} a directed path from $C_j$ to $C_i$. \rv{This is equivalent to the existence of a strictly positive vector in $\ker Y$.}

\rv{For weakly-reversible networks, a simple graphical characterization of the deficiency is given by \cite{feinberg87}: } \begin{equation}\label{deficiency}\delta=|\mathscr C|-\ell-\rank(\Gamma),\end{equation} where $\ell$ is the number of strongly connected components in the graph of complexes.

\rv{A main result in the theory of deterministic complex-balanced networks is stated as:}

\begin{lemma}[\cite{feinberg87}] If a reaction network with mass-action kinetics is weakly reversible and has deficiency zero, then there exists a unique positive equilibrium in every stoichiometric compatibility class.
\end{lemma}

\rv{Hence if the network is complex-balanced and without conservation laws, then it has a unique equilibrium for the deterministic system. A parallel result exists for the Markov chain model with associated chemical master equation. It is stated as follows:}

\begin{lemma}[\cite{anderson10,sontag10}] Assume a network is weakly reversible, has deficiency zero and has no conservation laws. Let $\hat x=[\hat x_{1},..,\hat x_{n}]^T$ be the unique complex balance equilibrium for the  deterministic system with the modified mass-action kinetics. Then the stationary distribution for the corresponding master equation is:
\[ \pi(x)=\prod_{i=1}^n \mathbf P(x_i;\hat x_{i}).\]

\end{lemma}

\subsection{Extension of the Main Result}

\rv{We are ready to define the generalized class of networks as gene regulatory networks whose conditional Markov chains are complex-balanced. In particular, if all the conditional Markov chains are weakly reversible and zero-deficiency, then the GRN can be treated with our methods.}

Consider a set of $N$ genes, binding sets $\{B_i\}_{i=1}^N$, and kinetics constants $k_j$. A \emph{generalized gene expression block} is as shown in Figure \ref{f.blockdiagram_g}. Then, a gene regulatory network is an arbitrary interconnection of gene expression blocks subject to Assumption A2.
A gene regulatory network admits conditional product-form distributions if, for each $d \in \{0,..,L-1\}$, the reaction network corresponding to the conditional Markov chain is weakly reversible and deficiency zero.
Hence, we can restate Proposition 2:
\begin{proposition}\label{th.cond_disSI} Fix $d \in \{0,..,L-1\}$. Let $\hat x^d,\hat x^{d2}$ be the complex-balanced equilibria for the proteins and multi-merized proteins, respectively. Consider eq. (16), then there exists a conditional joint stationary distribution $\pi_{X|d}^{(J)}$ and it is given by
\begin{equation}\label{e.cond_disSI} \pi^{(J)}_{X|d}(x)= \prod_{i=1}^N \pi_{X|di}(x_i),\end{equation}
where
\begin{equation}\label{e.cond_disSI2} \pi^{(J)}_{X|d_i}(x_i)= \left \{ \begin{array}{ll}\displaystyle \mathbf P\left ( x_{i1},x_{i2}; \hat x^d_i, \hat x^{d2}_i \right ) 
& \mbox{if}\  X_i \  \mbox{is multimerized} \\
\mathbf P\left ( x_{i}; \hat x^d_i \right ) ,& \mbox{otherwise}
\end{array}\right . ,\end{equation}
where {$(J)$ refers to the joint distribution in multimerized and non-multimerized processes, $x_{i1}$ refers to the copy number of $X_i$, while $x_{i2}$ refers to the copy number of $X_{ic}$,} $\mathbf P(x;a):= \frac {a^x}{x!} e^{-a},  \mathbf P(x_1,x_2;a_1,a_2):= \frac {a_1^{x_1}}{x_1!}\frac {a_2^{x_2}}{x_2!} e^{-a_1-a_2}$.
\end{proposition}

\subsection{Specific Examples}
\subsubsection{Cooperative Binding}
Cooperative binding is the process in which a TF forms as a hetero-dimer consisting of two different proteins. Consider two proteins $\mathrm X,\mathrm Y$, then the hetero-dimerization reaction:
\[\mathrm  X+\mathrm Y  \leftrightarrow \mathrm {XY} . \]

Consider a conditional Markov chain that consists of uncoupled birth-death processes, with multi-merization reactions. It is reversible and has deficiency zero. If a hetero-dimerization reaction is added, then this involves adding two complexes, adding a new strongly connected component to the complexes' graph, and adding a new linearly-independent row to the stoichiometry matrix. According to \eqref{deficiency} this will not change the deficiency. Hence, our results can be extended immediately to networks with cooperative binding.

\subsubsection{Diffusion}
We model diffusion between two gene expression blocks as a reversible reaction between the proteins as:
\[ \mathrm X  \leftrightarrow \mathrm Y . \]

Consider a conditional Markov chain that consists of uncoupled birth-death processes, with multi-merization reactions. It is reversible and has deficiency zero. If a diffusion reaction is added, then none of the number of complexes, the number of strongly connected components, or the rank of the stoichiometry matrix change. According to \eqref{deficiency} this will not change the deficiency. Hence, our results can be extended immediately to networks with diffusion. We have already applied this to the communicating toggle switches in the main text.
\rv{\subsubsection{Multi-Step Multi-merization}}
\rv{In the main text we have assumed that the multi-merization process occurs in one-step only as follows:
\begin{equation}\label{multimer} nX \xrightleftharpoons[ \beta_-]{ \beta} X_n. \end{equation}
As an alternative model, consider the following multi-step or \emph{sequential} multi-merization process:
\begin{align*}
2X & \xrightleftharpoons[ \beta_{-2}]{ \beta_2} X_2 \\
X_2+X &   \xrightleftharpoons[ \beta_{-3}]{ \beta_3} X_3 \\
\vdots \\
X_{n-1}+X &   \xrightleftharpoons[ \beta_{-n}]{ \beta_n} X_n \end{align*}
Consider a conditional Markov chain, conditioned on $D=d$, that consists of the above network with a birth and death process:
\[ 0 \xrightleftharpoons[ k_-]{ k} X.\]
In order to evaluate the deficiency \eqref{deficiency} of the network, note that $|\mathscr C|=2+2n, \ell=n+1, \rank(\Gamma)=n+1$. Hence, $\delta=0$. Therefore, the network is complex balanced and the stationary distribution is a product of Poisson distributions.}

\rv{Since $X_n$ acts as a transcription factor we are interested in finding $\mathbb E[X_n|D=d]$. Using Lemma SI-2, we can find the required quantity by finding the a complex-balanced equilibrium of the associated deterministic system with modified mass-action kinetics. Hence, we find:
\[\mathbb E[X_n|D=d]= \frac 1{2}\left( \frac {k}{k_-} \right )^n  \prod_{i=2}^n \frac{\beta_i}{\beta_{-i}}.\]
From the perspective of the slow promoter kinetics, the multi-step multi-merization reactions can be replaced with a single reaction \eqref{multimer} with $\beta/\beta_-= \frac 12 n! \prod_{i=2}^n {\beta_i}/{\beta_{-i}}$.}

\rv{\section{Extension to Networks with Multiple Copies of the Genes}}
In the main text, we have assumed that there exists one copy only per gene. Nevertheless, our theory is not limited by this assumption. In this subsection we show how the theory can be generalized.

Consider the network in Figure 2. Assume that the $i$th gene has $M_i$ copies.
Let $B_i$ be the corresponding binding-set, and denote $b_i=|B_i|-1$. We need to track the promoter configurations for all copies of the $i$th gene. Hence, we define the \emph{total binding set} $P_i$. The new set $P_i$ describes all possible $|B_i|$-tuples of nonnegative integers such that they sum to $M_i$.  It can be shown that $|P_i|=\binom{M_i+b_i}{b_i}$. To explain this, consider the following:
\begin{itemize}
  \item If one TF binds to a promoter or the promoter changes its configuration autonomously, then $B_i=\{0,1\}$, while $P_i := \{(0,M_i),(1,M_i),..,(M_i,0)\}$. Hence $|P_i|=M_i+1$.
  \item If two promoters bind independently to a promoter then $B_i=\{00,01,10,11\}$, while $P_i := \{(M_i,0,0,0), (M_i-1,1,0,0),....,(0,0,0,M_i)\}$, and $|P_i|=(M_i+3)(M_i+2)(M_i+1)/6$. Each tuple $(d_{i0},..,d_{i3})$ is interpreted as follows: there are $d_{i0}$ copies of the gene at configuration 00, $d_{i1}$ at configuration 01, etc. The tuple entries are ordered to correspond to 00,01,10,11, respectively.
  \item If two TFs bind competitively  to a promoter then the binding set is $B_i := \{00,10,01\}$, while $P_i := \{(M_i,0,0), (M_i-1,1,0),....,(0,0,M_i)\}$, and $|P_i|=(M_i+2)(M_i+1)/2$. The tuple entries are ordered to correspond to 00,01,10, respectively.
\end{itemize}

Consider the master equation \eqref{e.MasterEquation}.
Let us consider that the $i$th gene has $M_i$ copies. We again split the stochastic process $Z(t)$ into two subprocesses: \emph{the gene process} $D(t)$ and \emph{the protein process} $X(t)$. We will redefine $D(t)$ as explained below.

Consider the $i^{\rm th}$ gene. We define the process $D_i$ such that $D_i(t)
 \in P_i$, so it informs us exactly the copy numbers of genes at each promoter configuration.
 Collecting these into a vector, define the gene process $D(t):=[D_1(t),...,D_N(t)]^T$ where $D(t) \in \prod_{i=1}^N P_i $. The $i^{\rm th}$ gene can be represented by $|P_i|$ states, so $L:={\prod_{i=1}^N |P_i|}$ is  the total number of configurations in the GRN. {With} abuse of notation, we write also $D(t) \in \{0,..,L-1\}$ in the sense of some fixed bijection between $\{0,..,L-1\}$ and $\prod_{i=1}^N P_i $. Hence, $d \in \{0,..,L-1\}$ corresponds to $(d_1,...,d_N) \in P_1 \times .. \times P_N$ and we write $d=(d_1,..,d_N)$. Note that $d_i=(d_{i0},..,d_{ib_i})$.
 With the redefined $D$, the protein processes and the master equation are defined in an identical manner to the main text.

 In order to apply Proposition 2, we need to reexamine the conditional Markov chains.  Fix $D(t)=d$, and consider the $i$th gene. Let $d_i=(d_{i0},..,d_{ib_i})$, which informs us the copy number of genes with a certain configuration. Hence, the production rate for the protein is the sum of the production rates of the genes weighted by their respective copy numbers. Then, the conditional Markov chain is given as:
 \begin{equation}\label{e.birthdeathSI}
  \emptyset \xrightleftharpoons[  k_{-i}]{  k_{id_i}} \mathrm  X_i,
\end{equation}
or, with multimerization, as:
\begin{equation}\label{e.birthdeathSIDim}
  \emptyset \xrightleftharpoons[ k_{-i}]{ k_{id_i}} \mathrm X_i, \ n_i \mathrm  X_i \xrightleftharpoons[  \beta_{-i}]{ \beta_{i}} \mathrm X_{ic}.
  \end{equation}
 where
  \begin{equation}\label{ki} k_{id_i}:= \sum_{j=0}^{b_i} d_{ij} k_{ij}.\end{equation}

Theorem 3 and Corollary 4 hold as in the main text. Hence we conclude that
\begin{equation}\label{e.mixture_mSI}\pi(x)= \sum_{d=0}^{L-1} \lambda_d \pi_{X|d}(x)=\sum_{d=0}^{L-1} \lambda_d \prod_{i=1}^N \mathbf P\left(x_i;\frac{k_{id_i}}{k_{-i}}\right).\end{equation}

In order to find the weights $\lambda_i$'s, we need to restate Proposition SI-2 as follows:

\begin{proposition}\label{proposition_si} The matrix $\Lambda_r$ in eq. \eqref{e.solutionSI} can be computed via the algorithm below.

 \begin{itemize}\itshape
\item For each $d \in \{0,..,L-1\}$ write $d=(d_1,..,d_N)=((d_{10},...,d_{1b_1}),..,(d_{N0},...,d_{Nb_N}))\in \prod_{i=1}^N P_i$.  Using the previously discussed identification:
\begin{itemize}
\item[-]  {Let $\mathscr R_d=\{\rm R_1,..,\rm R_{|\mathscr R_d|}\}$ be the set of all gene reactions. Then, for each $\nu \in \{1,..,|\mathscr R_d|\}$}:
\begin{itemize}
  \item[*]
      \begin{enumerate}
   \item Let  {$\mathrm D_{j}^i$, and} $\mathrm D_{j'}^i$ be the  {reactant and product} configuration species of the {$\rm R_\nu$}, where $j \in B_i$.  {The reaction will cause a transition from $d$ to} a state $d'$ in which $d_{ij}^{'}=d_{ij}-1$ $d_{ij'}^{'}=d_{ij'}+1$, and $d_{ij}^{'}=d_{ij}$ otherwise. Let $\alpha$ be the kinetic constant of  {$\rm R_\nu$}. If  {$\rm R_\nu$} is a binding reaction, then let $\mathrm X_{\bar i }$, $\mathrm X_{\bar ic }$  denote the TF or the multi-merized TF,  {where $\bar i$ denotes the index of the gene that expresses the TF.}
  \item Then, the $(d',d)$ entry of $\Lambda_r$ can be written as:
     \begin{equation}\label{e.ratesSI_m}
         [\Lambda_r]_{d'd}= \left \{ \begin{array}{ll} \alpha, & \mbox{if the reaction is monomolecular}  \\
       \\ {\frac{\alpha}{n_{\bar i}!}  \frac{ \beta_{\bar i  }}{  \beta_{-\bar i}}\left (\frac{k_{\bar i d_{\bar i }}}{k_{-\bar i}} \right )^{n_{\bar i}}}, & \mbox{if the reaction is bimolecular}\end{array}\right .
     \end{equation}
     where $k_{\bar i  d_{\bar i }}$ is defined in \eqref{ki}.
   \end{enumerate}
   \end{itemize}
\item[-] Set \begin{equation}\label{e.rate_dm}\mathbf [\Lambda_r]_{d d}  = -\sum_{i'\ne i} \mathbf 1^T \hat \Lambda_{d_{i'} d_i } \pi_{d}.\end{equation}
\end{itemize}
\item Set the rest of the entries of $\Lambda_r$ to zero. \\
\end{itemize}
\end{proposition}

 \paragraph{Example: The Self-Regulating Gene} \strut \\
 Let us consider again the non-cooperative self-regulating gene,
 \begin{align}
\nonumber \mathrm  X+\mathrm D_0 &\xrightleftharpoons[\varepsilon \alpha_{-}]{\varepsilon \alpha} \mathrm D_1 \\ \label{noncoopSI}
  \mathrm D_0 &\mathop\rightarrow\limits^{k_{0}} \mathrm D_0+\mathrm X, \\
  \mathrm D_1 &\mathop\rightarrow\limits^{k_{1}} \mathrm D_1+\mathrm X,\nonumber \\
\mathrm X &  \mathop\rightarrow\limits^{k_{-}} 0 \nonumber .
\end{align}

Assume that the gene has $M$ copies. Hence, the total binding set $P_i=\{(M,0)...,(0,M)\}$, and we have $M+1$ states. Using Proposition \ref{proposition_si} the reduced order Markov chain can be represented as follows:
\[(M,0) \xrightleftharpoons[ \alpha_{-}]{  \alpha M k_0/k_-} (M-1,1) \xrightleftharpoons[ \alpha_{-}]{  \alpha (k_1+ (M-1) k_0)/k_-} (M-2,2) \dots  \xrightleftharpoons[ \alpha_{-}]{  \alpha Mk_1 k_0/k_-} (0,M) .\]

The resulting stationary distribution is a mixture of Poissons with modes at $Mk_0, (M-1)k_0+k_1, ..., Mk_1$.
Figure \ref{f.higher_M} shows the stationary distribution for several values of $M$ for the network above. The number of modes increases with $M$, however, the stationary distribution converges  for high $M$.

\begin{figure}
  \centering
  \includegraphics[width=\textwidth]{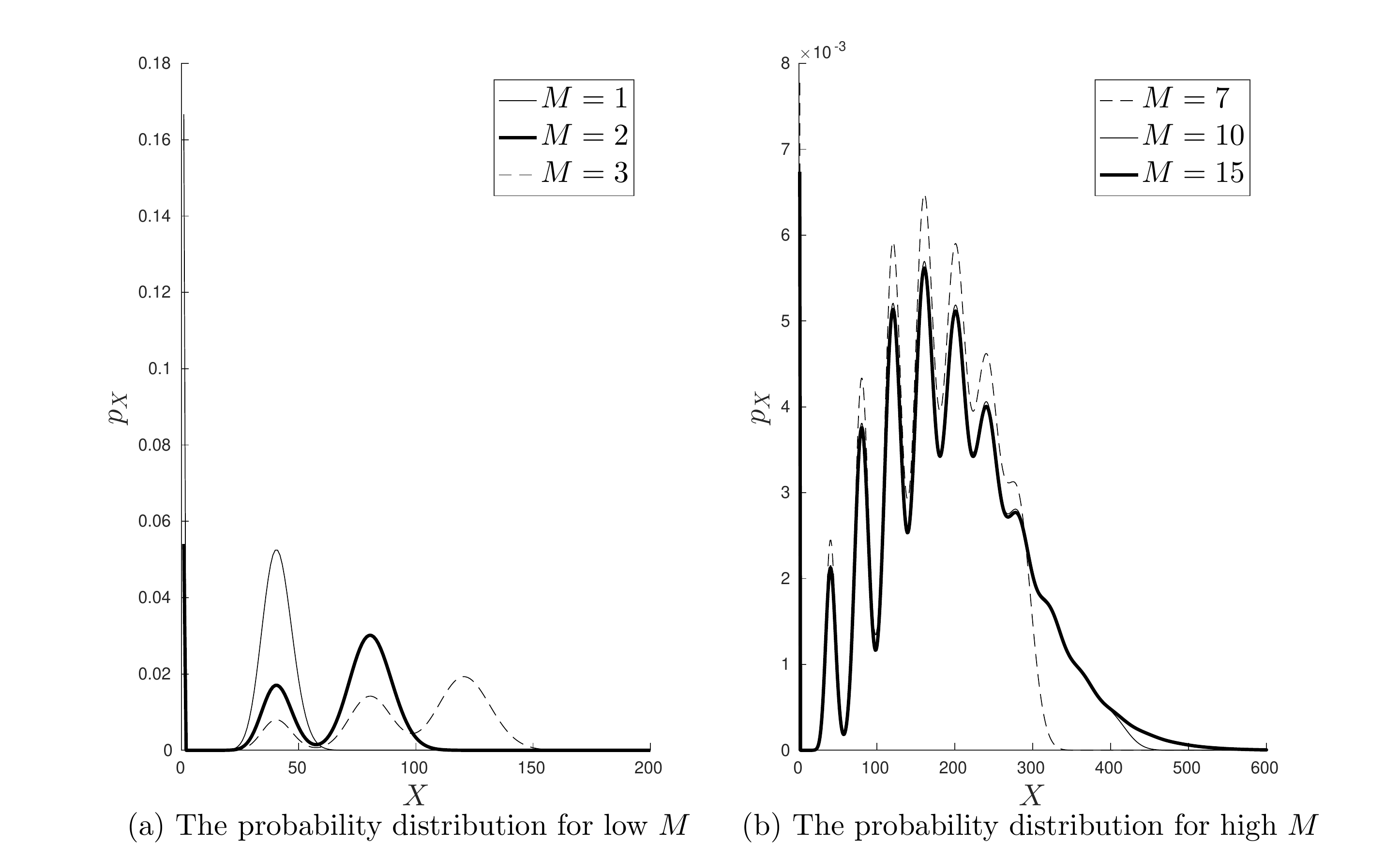}
  \caption{\textbf{The number of modes increase with the number of gene copies.}  (a) The stationary distribution for one, two and three copy numbers. The parameters are chosen to be identical to the one in Figure 1. (b) The stationary distribution for seven, ten and fifteen. Note coefficients of the high modes are very small. }\label{f.higher_M}
\end{figure}

\section{Additional Information for Tables and Figures in the Main Text}

\subsection{Figure 1-b}
A self-repressing gene is given as:
\begin{align}
\nonumber \mathrm X+\mathrm D_0 &\xrightleftharpoons[\varepsilon \alpha_{-}]{\varepsilon \alpha}\mathrm  D_1 \\
  \mathrm D_0 &\mathop\rightarrow\limits^{k_{0}} \mathrm D_0+\mathrm X, \\
\mathrm X &  \mathop\rightarrow\limits^{k_{-}} 0 \nonumber .
\end{align}

The ordinary differential equations for the deterministic model can be written as:
\[ \begin{bmatrix}\dot d_0 \\ \dot d_1 \\ \dot x \end{bmatrix} = \begin{bmatrix} -1 & 1 & 0 & 0 \\ 1 & -1 & 0 & 0 \\ -1 & 1 & 1 & -1 \end{bmatrix} \begin{bmatrix} \alpha d_0 x \\ \alpha_- d_1 \\ k_0 d_0 \\ k_- x \end{bmatrix},  \]
with $d_0(t)+d_1(t)=1$ for all $t$.

Setting the derivatives to zero, and solving for $x$ gives the following charcterisitic equation:
\[\alpha k_- x^2 + \alpha_- k_- x - \alpha_- k_0 =0.\]
Since the third term has a negative sign, the network has a unique positive equilibrium for all parameter values.
The parameters used in the simulation are:
\[ \alpha=\varepsilon/200, \alpha_-=\varepsilon, k_0=40, k_-=1.\]

The deterministic equilibrium is $\hat x \approx 34.1641$.
The curve for the fast promoter kinetic has been computed by applying the recurrence relation in Proposition 7.
The curve for the slow promoter kinetic has been computed by Eq. (32).
The remainder of the curves have been computed using a finite state projection truncated at $x=65$. This amounts to writing the master equation as a finite linear differential equation of the form:
\[ \dot p = \Lambda p, \]
then evaluating the principal eigenvector of $\Lambda$ using the command \texttt{null} in MATLAB, and normalizing the resulting vector.

\subsection{Figure 1-d}
A hybrid repression-activation network is given as:
\begin{align*}\begin{array}{rl}
\nonumber \mathrm Y+ \mathrm D_0^X &\xrightleftharpoons[\alpha_{-1}]{\alpha_1}  \mathrm D_1^X \\
 \mathrm  D_0 &\mathop\rightarrow\limits^{k_{10}}\mathrm  D_0^X+\mathrm X, \\
\mathrm X &   \mathop\rightarrow\limits^{k_{-1}} 0 \\ \end{array} \begin{array}{rl}
\nonumber \mathrm X+ \mathrm D_0^Y &\xrightleftharpoons[\alpha_{-2}]{\alpha_2} \mathrm D_1^Y  \\
\nonumber  \mathrm D_1^Y &\mathop\rightarrow\limits^{k_{21}} \mathrm D_1^Y+\mathrm Y, \\ \mathrm Y & \mathop\rightarrow\limits^{k_{-2}} 0. \end{array}
\end{align*}

The parameters used in the simulation are:
\[ \alpha_1=\varepsilon/70, \alpha_{-1}=\varepsilon, \alpha_2=\varepsilon/20, \alpha_{-2}=\varepsilon, k_{10}=k_{21}=40, k_{-1}=k_{-2}=1.\]

Modelling the network deterministically, there exists a unique positive equilibrium where the equilibrium values of $X,Y$ are given as $ \hat x \approx 25.4727, \hat y\approx39.9216 .$
The stationary distribution has been computed using Corollary 4 and Proposition 5.

\subsection{Table 1}
Let us consider the non-cooperative self-regulating gene. The reaction network is given as:
\begin{align}
\nonumber \mathrm  X+\mathrm D_0 &\xrightleftharpoons[\varepsilon \alpha_{-}]{\varepsilon \alpha} \mathrm D_1 \\
 \mathrm  D_0 &\mathop\rightarrow\limits^{k_{0}} \mathrm D_0+\mathrm X, \\
 \mathrm  D_1 &\mathop\rightarrow\limits^{k_{1}}\mathrm  D_1+\mathrm X,\nonumber \\
\mathrm X &  \mathop\rightarrow\limits^{k_{-}} 0 \nonumber .
\end{align}
The deterministic model gives the following characteristic equation for $x$:
\[ \alpha k_- x^2 +(\alpha_- k_- - \alpha k_1  ) x -  \alpha_- k_0 =0.  \]

Since the third term has a negative sign, the network has a unique positive equilibrium for all parameter values.
For the stochastic model:   With no leakiness, i.e. $k_0=0$, the Markov chain has an absorbing state at $(X,D_1)=(0,0)$. Hence, the network has a single mode at zero.
With leakiness,  we have shown in the main text that the slow promoter kinetics give rise to two modes (Corollary 4). While for fast promoter kinetics, we have mentioned that it gives rise to a uni-modal distribution after Proposition 7.
Let us consider the cooperative self-regulating gene. The reaction network is given as:
\begin{align}
\nonumber \mathrm  X_c+\mathrm D_0 &\xrightleftharpoons[\varepsilon \alpha_{-}]{\varepsilon \alpha} \mathrm D_1 \\
  \mathrm D_0 &\mathop\rightarrow\limits^{k_{0}} \mathrm D_0+\mathrm X, \\
  \mathrm D_1 &\mathop\rightarrow\limits^{k_{1}} \mathrm D_1+X,\nonumber \\
\mathrm X &  \mathop\rightarrow\limits^{k_{-}} 0 \nonumber \\
2\mathrm X & \xrightleftharpoons[\varepsilon \beta_{-}]{\varepsilon \beta}\mathrm  X_c
\end{align}

The deterministic model gives the following characteristic equation for $x$:
\begin{equation}\label{cubic} - \alpha \beta  k_- x^3 + \alpha \beta k_1 x^2 - \alpha_- \beta _- k_- x + \alpha_-\beta_-k_0 =0.\end{equation}

Since the constant term is negative, this means that \eqref{cubic} has at least one positive real solution.
The equation \eqref{cubic} has three solutions if
\[ \alpha_1 k_{1} k_-^2 \beta \alpha_{-1} \beta_{-1}  (k_{1} + 18 k_{0}) > 4 \alpha_1^2 k_{1}^3 \beta^2 k_{0} + 27 \alpha_1 k_-^2 \beta k_{0}^2 \alpha_{-1} \beta_{-1} + 4 k_-^4 \alpha_{-1}^2 \beta_{-1}^2 , \]
and one solution otherwise. Using the Routh-Hurwitz cireterion, all  solutions have positive real parts if and only if $k_1 > k_0$.   This means that the network can have either one or two positive \emph{stable} equilibria only.

For the stochastic model:   with no leakiness, the discussion is identical to the previous case.
With leakiness,  we have shown in the main text that the slow promoter kinetics give rise to two modes independent of the cooperativity index (Corollary 4). While for fast promoter kinetics, a single mode has been verified by solving the master equation numerically for many parameter sets including the ones that are multi-stable for the deterministic system.

\subsection{Figure 3-b}
The reaction network is given in Eq. (34).
The parameters are:
\[ \alpha=\varepsilon, \alpha_-=\varepsilon, k_0=200, k_1=1000, k_-=100, \beta=100, \beta_-=100.\]

For these parameters, the deterministic system has a unique stable positive equilibrium at $\hat x \approx 9.9195$.
The curve for the slow promoters limit has been computed using Eq. (36).
The remaining curves have been computed by a numerical solution of the master equation as articulated in SI-5.1.

\subsection{Figure 5b-e}
 We consider a population of three toggle switches ($N=3$) with $k_x=k_y=150, k_{-x}=k_{-y}=1, \beta_{x}=\beta_{y}=1, \beta_{-x}=\beta_{-y}=1, \alpha_x=\alpha_y=0.3\varepsilon, \alpha_{-x}=\alpha_{-y}=\varepsilon$ and $n=2$.

 The surface plots have been computed using Eq. (25).

 \subsection{Figure 6b-c}
 The parameters for Figure 6b are as follows:
 The dissociation ratio is fixed at 1/2000 for all binding/unbinding reactions and we let the dimerization ratio be 1/90.  The production ratios are $k_{x2}/k_{-x}=k_{y1}/k_{-y}=2700$. We assume that the inhibition and activation actions allow for leaks and hence we let $k_{x0}/k_{-x}=k_{y0}/k_{-y}=1080$, $k_{x3}/k_{-x}=k_{y3}/k_{-y}=675$ and $k_{x2}/k_{-x}=k_{y1}/k_{-y}=20$.

 The parameters for Figure 6c are as follows:
 The dissociation ratios are $\alpha_{x0}/\alpha_{-x0}=10 , \alpha_{x1}/\alpha_{-x1}=1/2700, \alpha_{y0}/\alpha_{-y0}=10 , \alpha_{y1}/\alpha_{-x1}=1/2700$. The maximal production ratio is $k_{x1}/k_{-x}=k_{y1}/k_{-y}=2700$, while $k_{x0}/k_{-x}27,k_{x2}/k_{-y}=0$, $k_{y2}/k_{-y}=0,k_{y0}/k_{-y}=270$.

 We ratio of the production to the leak were chosen as $k_{x1}:k_{x0}\approx 100:1$, and $k_{y1}: k_{y0} \approx 5-10:1$ \cite{nishikawa02},\cite{okuno05}.

\newpage
\bibliographystyle{unsrt}

\end{document}